\PassOptionsToPackage{dvipsnames}{xcolor}

\documentclass[acmsmall]{acmart}
\usepackage[utf8]{inputenc}

\newcommand{\citerepl}{\cite{zenodoreplpackage}}

\usepackage{amsfonts}
\usepackage{stmaryrd} 
\usepackage{listings} 
\usepackage{amsthm} 
\usepackage[capitalise]{cleveref} 
\usepackage{mathtools} 

\usepackage{algorithm}
\usepackage[noend]{algorithmic} 

\usepackage{eso-pic} 
\usepackage{tikz}
\newcommand{\errnote}[1]{%
  \tikz[remember picture,overlay,baseline]
    \node[anchor=base west, font=\scriptsize\sffamily\it, text=red,
          text width=1.4cm, align=left, inner sep=0pt]
      at ([xshift=-1.4cm]current page.east |- 0,0) {#1};}
      
\definecolor{darkyellow}{RGB}{125, 99, 45}
\definecolor{blue}{RGB}{0, 24, 255}
\definecolor{purple}{RGB}{175, 0, 219}
\definecolor{lightblue}{RGB}{38, 128, 153}
\definecolor{green}{RGB}{8, 128, 0}
\definecolor{red}{RGB}{163, 21, 21}
\definecolor{lightgrey}{RGB}{204, 204, 204}
\definecolor{red}{RGB}{163, 21, 21}
\definecolor{errorred}{RGB}{241,156,153}

\lstdefinestyle{Python}{
  language=Python,
  captionpos=b,
  basicstyle=\ttfamily\footnotesize,
  numbers=left,
  numberstyle=\small\color{lightgray},
  tabsize=2,
  columns=fixed,
  showstringspaces=false,
  showspaces=false,
  showtabs=false,
  keepspaces,
  commentstyle=\color{green},
  stringstyle=\color{red},
  keywordstyle=\color{purple},
  morekeywords={with,as,then,do,end},
  deletekeywords={def,not,in,and,or},
  keywords={[2]@invariant,sample,score,resample,read,read_trace,revisit,visit},
  keywordstyle={[2]\color{lightblue}}, 
  keywords={[3]@invariant,
  len,
  },
  keywordstyle={[3]\color{darkyellow}},
  keywords={[4]@invariant,def,not,and,or,in},
  keywordstyle={[4]\color{blue}},
}

\lstdefinestyle{Julia}{
  language=Python,
  captionpos=b,
  basicstyle=\scriptsize\textnormal\ttfamily,
  numbers=left,
  numberstyle=\scriptsize\color{lightgray},
  tabsize=2,
  columns=fixed,
  showstringspaces=false,
  showspaces=false,
  showtabs=false,
  keepspaces,
  commentstyle=\color{green},
  stringstyle=\color{red},
  keywordstyle=\color{purple},
  morekeywords={end, begin, elseif},
  keywords={[2]@invariant,},
  keywordstyle={[2]\color{lightblue}}, 
  keywords={[3]@invariant,
  @model, @gen
  },
  keywordstyle={[3]\color{darkyellow}},
  keywords={[4]@invariant, function},
  keywordstyle={[4]\color{blue}},
}

\lstdefinestyle{C++}{
  language=C++,
  captionpos=b,
  basicstyle=\scriptsize\textnormal\ttfamily,
  numbers=left,
  numberstyle=\scriptsize\color{lightgray},
  tabsize=2,
  columns=fixed,
  showstringspaces=false,
  showspaces=false,
  showtabs=false,
  keepspaces,
  commentstyle=\color{green},
  stringstyle=\color{red},
  keywordstyle=\color{purple},
  morekeywords={data,parameters,model},
  deletekeywords={},
  keywords={[2]@invariant,},
  keywordstyle={[2]\color{lightblue}}, 
  keywords={[3]@invariant,
  len,
  },
  keywordstyle={[3]\color{darkyellow}},
  keywords={[4]@invariant,def,not,and,or,in},
  keywordstyle={[4]\color{blue}},
}

\newtheorem{theorem}{Theorem}
\newtheorem{corollary}{Corollary}
\newtheorem{proposition}{Proposition}
\newtheorem{lemma}{Lemma}
\newtheorem{definition}{Definition}

\newenvironment{statement}{\hspace{\parindent}\textsc{Statement}.\itshape}{}

\newcommand{\hbra}{
  \hbox to \columnwidth{\vrule width0.3mm height 1.8mm depth-0.3mm
    \leaders\hrule height1.8mm depth-1.5mm\hfill
    \vrule width0.3mm height 1.8mm depth-0.3mm}}
\newcommand{\hket}{
  \hbox to \columnwidth{\vrule width0.3mm height1.5mm
    \leaders\hrule height0.3mm\hfill
    \vrule width0.3mm height1.5mm}}

\newcommand{\ratio}{.35}

\newenvironment{display}[2][\ratio]{
  \begin{tabbing}
    \hspace{0.1em} \= \hspace{1.5em} \= \hspace{#1\linewidth-3.2em} \= \hspace{1.5em} \= \kill
    \textbf{#2}\\[-.8ex]
    \\[-.8ex]
  }
  {\\[-.8ex]
  \end{tabbing}}
  
\newcommand{\entry}[2]{\>\>$#1$\>\>#2} 
\newcommand{\clause}[3][]{\>$#2$\>#1\>#3}
\newcommand{\Category}[2]{\clause{#1::=}{#2}}

\newcommand{\kw}[1]{\mbox{\lstinline$#1$}}
\newcommand{\pdf}{\textnormal{pdf}}

\newcommand{\medcup}{\textstyle \bigcup}

\newcommand{\Reals}{\mathbb{R}}
\newcommand{\pdfvalues}{\Reals_{\ge 0}}
\newcommand{\Nats}{\mathbb{N}}
\newcommand{\Ints}{\mathbb{Z}}
\newcommand{\powerset}{\mathscr{P}}

\newcommand{\true}{\mathtt{true}}
\newcommand{\false}{\mathtt{false}}
\newcommand{\none}{\mathtt{null}}

\newcommand{\cfg}{G}
\newcommand{\unrolledcfg}{H}

\newcommand{\measurespace}{\mathcal{M}}
\newcommand{\sigmaalgebra}{\Sigma}
\newcommand{\refmeasure}{\nu}

\newcommand{\values}{\mathcal{V}}
\newcommand{\exprs}{\mathbf{Exprs}}

\newcommand{\progstates}{\Sigma}
\newcommand{\strings}{\mathbf{Strings}}
\newcommand{\types}{\mathbf{Types}}
\newcommand{\dict}{\mathtt{tr}}
\newcommand{\obsdict}{\mathtt{obs\_tr}}
\newcommand{\dicts}{\mathcal{T}}
\newcommand{\addr}[1]{\textnormal{"#1"}}

\newcommand{\undefval}{\textnormal{undefined}}

\newcommand{\opsemrel}{\Downarrow^{\dict}}
\newcommand{\progstate}{\sigma}

\newcommand{\probvar}{\textbf{p}}
\newcommand{\cfgrel}{\overset{\dict}{\longrightarrow}}
\newcommand{\cfgreloverset}[1]{\overset{#1}{\longrightarrow}}
\newcommand{\staticprov}{\textnormal{prov}}
\newcommand{\staticprovnode}{\textnormal{provnode}}

\newcommand{\evalf}[2]{V^{#1}_{#2}}
\newcommand{\restricedfs}[2]{[\dicts|_{#1} \to #2]}

\newcommand{\RD}{\textnormal{RD}}
\newcommand{\CP}{\textnormal{BP}}
\newcommand{\cfgnode}{\iota_\textnormal{cfg}}

\newcommand{\dirac}[1]{\delta_{#1}}

\newcommand{\milliseconds}{ms}
\newcommand{\microseconds}{\mu s}

\newcommand{\megafastarrow}{\uparrow\uparrow}
\newcommand{\fastarrow}{\nearrow}
\newcommand{\samearrow}{\rightarrow}
\newcommand{\slowarrow}{\searrow}
\newcommand{\megaslowarrow}{\downarrow\downarrow}

\newcommand{\expectation}[2]{\mathbb{E}_{#1}{\left[#2\right]}}

\newcommand{\revcolor}{}

\newcommand{\revvcolor}{}

\AtBeginDocument{%
  }

\setcopyright{cc}
\setcctype{by}
\acmDOI{10.1145/3798223}
\acmYear{2026}
\acmJournal{PACMPL}
\acmVolume{10}
\acmNumber{OOPSLA1}
\acmArticle{115}
\acmMonth{4}
\received{2025-10-08}
\received[accepted]{2026-02-17}
\received[corrected]{2026-11-02}
\begin{document}
\AddToShipoutPictureFG*{%
  \AtPageUpperLeft{\raisebox{-1.1cm}{\hspace{1.6cm}\parbox{0.8\paperwidth}{\footnotesize\color{red}
  \textbf{Note (arXiv v3):} This version adds errata on p.~\pageref{sec:errata}
  correcting Corollary~3. The text above the references is identical to the
  version of record, \url{https://doi.org/10.1145/3798223}.}}}}

\title{Static Factorisation of Probabilistic Programs with User-Labelled Sample Statements and While Loops}

\author{Markus Böck}
\orcid{0009-0001-6704-5903}
\affiliation{%
  \institution{TU Wien}
  \city{Vienna}
  \country{Austria}
}
\email{markus.h.boeck@tuwien.ac.at}

\author{Jürgen Cito}
\orcid{0000-0001-8619-1271}
\affiliation{%
  \institution{TU Wien}
  \city{Vienna}
  \country{Austria}
}
\email{juergen.cito@tuwien.ac.at}


\begin{abstract}
It is commonly known that any Bayesian network can be implemented as a probabilistic program, but the reverse direction is not so clear.
In this work, we address the open question to what extent a probabilistic program with user-labelled sample statements and while loops -- features found in languages like Gen, Turing, and Pyro -- can be represented graphically.
To this end, we extend existing operational semantics to support these language features.
By translating a program to its control-flow graph, we define a sound static analysis that approximates the dependency structure of the random variables in the program.
As a result, we obtain a static factorisation of the implicitly defined program density, which  is equivalent to the known Bayesian network factorisation for programs without loops and constant labels, but constitutes a novel graphical representation for programs that define an unbounded number of random variables via loops or dynamic labels.
We further develop a sound program slicing technique to leverage this structure to statically enable three well-known optimisations for the considered program class:
we reduce the variance of gradient estimates in variational inference and we speed up both single-site Metropolis Hastings and sequential Monte Carlo.
These optimisations are proven correct and empirically shown to match or outperform existing techniques.
\end{abstract}




\maketitle

\section{Introduction}

Probabilistic programming provides an intuitive means to specify probabilistic models as programs.
Typically, probabilistic programming systems also come with machinery for automatic posterior inference.
This enables practitioners without substantial knowledge in Bayesian inference techniques to work with complex probabilistic models.

Besides these practical aspects, there is also a long history of studying probabilistic programming languages (PPLs) from a mathematical perspective~\cite{kozen1979semantics, barthe2020foundations}.
It is commonly known that any Bayesian network~\cite{koller2009pga} can be implemented as a probabilistic program \cite{gordon2014probabilistic-overview}, but the reverse direction is not so clear.
The understanding is that first-order probabilistic programs without loops can be compiled to Bayesian networks, because the number of executed sample statements is fixed and finite, such that each sample statement corresponds to exactly one random variable \cite{van2018introppl, borgstrom2011measure-transformer, sampson2014assertions}.
Generally, a graphical representation of a probabilistic program is desirable as it lends itself for optimising inference algorithms by exploiting the dependency structure between random variables.
In fact, many probabilistic programming languages require the user to explicitly construct the model in a graphical structure~\cite{salvatier2016pymc,minka2012infer,mccallum2009factorie,lunn2000winbugs,milch2007blog}.

However, in the universal probabilistic programming paradigm probabilistic constructs are embedded in a Turing-complete language~\cite{goodman2012church}.
While loops and recursive function calls easily lead to an unbounded number of executed sample statements which makes {\revcolor the typical compilation to a graph with one node per random variable impossible}.
Even more, many PPLs allow the user to label each sample statement with a dynamic address.
This further complicates matters, because in this case even a single sample statement may correspond to a potentially infinite number of random variables.
As a consequence, many approaches resort to dynamic methods to leverage the dependency structure of this kind of probabilistic programs for efficient inference~\cite{wu2016swift,yang2014generatingefficient,chen2014sublinear}.

{\revcolor
Sample statements with dynamic address expressions constitute a valuable language feature, as they allow the definition of random variables whose identifiers differ from the program variables to which values are assigned to.
This permits the specification of a broader model class relative to PPLs that lack this feature.
Below we give examples of such user-labelled sample statements in Gen~\cite{cusumano2019gen}, Pyro~\cite{bingham2019pyro}, and Turing~\cite{ge2018turing}, where the address expressions are highlighted with }{\color{red}\underbar{red}}:
\begin{lstlisting}[style=Python, numbers=none, moredelim={[is][\color{red}\underbar]{@}{@}}, moredelim={[is][\color{black}]{?}{?}}],
    x = {@:x => i@} ~ Normal(0.,1.)             # Gen
    x = pyro.?sample?(@f"x{i}"@, Normal(0.,1.))   # Pyro
    @x[i]@ ~ Normal(0.,1.)                      # Turing
    x = sample(@"x"+str(i)@, Normal(0.,1.))     # our formal syntax
\end{lstlisting}

{\revcolor In this work, we address the open question to what extent a probabilistic program with \emph{user-labelled sample statements and while loops} can be represented graphically.
We prove that that the density of a probabilistic program with $K$ user-labelled sample statements decomposes into $K$ factors.
This finite factorisation holds even if the program contains while loops.
Moreover, the dependency structure implied by the factorisation can be computed \emph{statically} which, as we will show, enables new opportunities for the optimisation of inference algorithms.
}

\subsection{Overview}
\textbf{\cref{sec:semantics}. Address-Based Semantics For Probabilistic Programs} {\revcolor generalises the operational semantics of Core Stan~\cite{gorinova2019slicstan}.}
Instead of interpreting the program density $p$ as a function over only a fixed set of finite variables, we extend this approach to interpret the density as a function over traces.
In contrast to \citet{gorinova2019slicstan}, this allows us to consider programs with dynamic addresses and while loops.
These language features greatly expand the class of expressible probabilistic models, for instance, models with an infinite number of random variables.
{\revcolor The described semantics can be seen as an operational analog to the denotational models of Pyro~\cite{lee2019towards} and Gen~\cite{lew2023stochasticprobs}.
}

\textbf{\cref{sec:without-loops}. Static Factorisation for Programs Without Loops} examines the factorisation of the program density for the simplified case where programs do not contain while loops.

In \cref{sec:staticprov}, {\revcolor building on standard data and control flow analysis techniques}, we present a static provenance analysis that finds the dependencies of a program variable $x$ at any program state $\progstate$.
The dependencies are described by specifying for which addresses a trace~$\dict$ needs to be known to also know the value of $x$ in $\progstate$.
This analysis forms the basis for finding the factorisation over addresses and operates on the control-flow graph (CFG) of a program.

In \cref{sec:theorem-without-loops}, we {\revcolor derive a novel factorisation from the static analysis}, $p(\dict) = \prod_{k=1}^K p_k(\dict)$, for programs without loops by explicitly constructing one factor~$p_k$ per sample statement.
If we further restrict the program to constant addresses, we arrive at the language of \citet{gorinova2019slicstan}, where our factorisation theorem implies {\revcolor a known} equivalence to Bayesian networks.
{\revcolor However, we} show that already in the case of dynamic addresses, programs are only equivalent to Markov networks, because of the possibility of cyclic dependencies and an infinite number of random variables.


\textbf{\cref{sec:loops}. Static Factorisation for Programs With Loops} describes a way to generalise the factorisation result of \cref{sec:without-loops} to the much more challenging case where programs may contain while loops.
{\revcolor The idea is to mathematically unroll while loops to obtain a directed acyclic graph, the \emph{unrolled CFG}, that essentially represents all possible program paths.
Our main theoretical contribution is to show that even in this case, the density still factorises into a finite number of terms -- one factor for each of the $K$ sample statements}.
Like dynamic addresses, while loops may also be used to model an infinite number of random variables and thus, the resulting factorisation in general implies the equivalence of programs to Markov networks.

\textbf{\cref{sec:casestudy}. {\revcolor Optimising Posterior Inference Algorithms}}  {\revcolor presents a novel and sound program slicing method building on the factorisation theory established in \Cref{sec:loops}.
By generating sub-programs for each factor, we can 1) decrease the runtime of a Metropolis Hastings algorithm by speeding up the acceptance probability computation; 2) construct a lower-variance gradient estimator for Black-Box Variational Inference; 3) facilitate an iterative implementation of sequential Monte Carlo.
Each of the considered optimisations are known, but have been either realised dynamically or were reserved for PPLs that require the explicit graphical construction of probabilistic models. 
Our approach is the first to enable them statically for the considered program class.
We evaluate our approach on a benchmark set of 16 probabilistic programs, the majority of which make use of while loops to define an unbounded number of random variables.


}

\subsection{Contributions}
{\revcolor We briefly summarise the key contributions of this work:
\begin{itemize}
    \item A static provenance analysis in the context of our formal PPL building on standard data and control flow techniques,
    \item A novel density factorisation for the considered program class, which can be computed statically and subsumes a known equivalence to Bayesian networks in a restricted setting,
    \item A novel and sound program slicing method derived from the factorisation theory,
    \item An application of the slicing method to enable well-known optimisations \emph{statically and provably correct} for the considered program class and three inference algorithms: Metropolis Hastings, variational inference, and sequential Monte Carlo,
    \item An open-source implementation and empirical evaluation of our proposed approach~\citerepl{}.
\end{itemize}
}

\section{Address-Based Semantics for Probabilistic Programs}\label{sec:semantics}

First, we introduce the syntax of our formal probabilistic programming language in terms of expressions and statements similar to \citet{gorinova2019slicstan} and \citet{hur2015provablycorrectsampler}.
We assume that the constants of this language range over a set of values including booleans, integers, real numbers, immutable real-valued vectors, finite-length strings, and a special $\none$ value.
We define $\values$ to be the set of all values.
A program contains a finite set of program variables denoted by $x$, $x_i$, $y$, or $z$.
We further assume a set of built-in functions $g$.
We assume that functions $g$ are total and return $\none$ for erroneous inputs.
The only non-standard construct in the language are sample statements.
For sample statements, the expression $E_0$ corresponds to the user-defined sample address, which is assumed to evaluate to a string.
The symbol $f$ ranges over a set of built-in distributions which are parameterised by arguments $E_1,\dots,E_n$.

\noindent
\begin{minipage}[c]{1.5\textwidth}
\begin{display}[.15]{Syntax of Expressions:}
	\Category{E}{expression}\\
	\entry{c}{constant}\\
	\entry{x}{variable}\\
	\entry{g(E_1,\dots,E_n)}{function call}\\
	\\
	\\
\end{display}
\end{minipage}\hspace{5mm}%
\begin{minipage}[c]{\textwidth}
\begin{display}[.4]{Syntax of Statements:}
	\Category{S}{statement}\\
	\entry{x = E}{assignment}\\
	\entry{S_1; S_2}{sequence}\\
	\entry{\kw{if}\;E\;\kw{then}\;S_1\,\kw{else}\;S_2}{if statement} \\
	\entry{\kw{skip}}{skip}\\
    \entry{\kw{while}\;E\;\kw{do}\;S}{while loop}\\
    \entry{x = \kw{sample}(E_0, f(E_1,\dots,E_n))}{sample}
\end{display}
\end{minipage}

\subsection{Operational Semantics}

The meaning of a program is the density implicitly defined by its sample statements.
This density is evaluated for a \emph{program trace} and returns a real number.
A program trace is a mapping from finite-length addresses to numeric values, $\dict \colon \strings \to \values$.
We denote the set of all traces with $\dicts$.
How the value at an address is determined is not important for this work and in practice typically depends on the algorithm that is used to perform posterior inference for the program.
For instance, if an address $\alpha$ corresponds to an observed variable, the value $\dict(\alpha)$ is fixed to the data.
If an address $\alpha$ does not appear in a program, then its value can be assumed to be $\dict(\alpha) = \none$.

For each trace $\dict$, the operational semantics of our language are defined by the big-step relation $(\progstate,S) \opsemrel \progstate'$,
where $S$ is a statement and $\progstate \in \progstates$ is a program state -- a finite map from program variables to values
\begin{displaymath}
    \progstate ::= x_1\mapsto V_1,\dots, x_n \mapsto V_n, \quad \, x_i\;\text{distinct},\; V_i \in \values.
\end{displaymath}
A state can be naturally lifted to a map from expressions to values $\progstate\colon \exprs \to \values$:
\begin{displaymath}
    \progstate(x_i) = V_i, \quad \progstate(c) = c, \quad \progstate(g(E_1,\dots,E_n)) = g(\progstate(E_1),\dots,\progstate(E_n)).
\end{displaymath}
The operational semantics for the non-probabilistic part of our language are standard:
\vspace{1mm}
\begin{displaymath}
    \frac{}{(\progstate, \kw{skip}) \opsemrel \progstate}\quad
    \frac{}{(\progstate, x = E) \opsemrel \progstate[x \mapsto \progstate(E)]}\quad
    \frac{(\progstate, S_1)\opsemrel \progstate' \quad (\progstate', S_2)\opsemrel \progstate''}{(\progstate, S_1;S_2) \opsemrel \progstate''}
\end{displaymath}
\vspace{2mm}
\begin{displaymath}
    \frac{\progstate(E) = \true \quad (\progstate, S_1) \opsemrel \progstate'}%
{(\progstate, \kw{if}\;E\;\kw{then}\;S_1\,\kw{else}\;S_2) \opsemrel \progstate'} \quad
    \frac{\progstate(E) = \false \quad (\progstate, S_2) \opsemrel \progstate'}%
{(\progstate, \kw{if}\;E\;\kw{then}\;S_1\,\kw{else}\;S_2) \opsemrel \progstate'}
\end{displaymath}
\vspace{2mm}
\begin{displaymath}
    \frac{\progstate(E)= \false}%
{(\progstate, \kw{while}\;E\;\kw{do}\;S) \opsemrel \progstate}\quad
\frac{\progstate(E) = \true \quad (\progstate, (S; \kw{while}\;E\;\kw{do}\;S)) \opsemrel \progstate'}%
{(\progstate, \kw{while}\;E\;\kw{do}\;S) \opsemrel \progstate'}
\end{displaymath}
\vspace{1mm}

Above, $\progstate[x \mapsto V]$ denotes the updated program state $\progstate'$, where $\progstate'(x) = V$ and $\progstate'(y) = \progstate(y)$ for all other variables $y \neq x$.
The inference rules are to be interpreted inductively such that the semantics of a while loop can be rewritten as 
\begin{eqnarray*}
     (\progstate, \kw{while}\;E\;\kw{do}\;S) \opsemrel \progstate' \iff & \exists n \in \Nats_0\colon \left((\progstate, \textnormal{repeat}_n(S)) \opsemrel \progstate' \land \progstate'(E) = \false\right) \land \\
    & \quad \forall m < n\colon \left((\progstate, \textnormal{repeat}_m(S)) \opsemrel \progstate'_m \land \progstate'_m(E) = \true\right),
\end{eqnarray*}
where $\textnormal{repeat}_n(S) = (S;...;S)$ repeated $n$-times. In particular, if a while loop does not terminate for state $\progstate$ then $\nexists \progstate'\colon (\progstate, \kw{while}\;E\;\kw{do}\;S) \opsemrel \progstate'$.
This makes the presented operational semantics partial, where errors and non-termination are modelled implicitly.
Formally, $\opsemrel$ is the least-relation closed under the stated rules.

At sample statements, we inject the value of the trace at address $V_0 = {\revcolor\progstate(E_0)}$.
We multiply the density, represented by the reserved variable $\probvar$, with the value of function $\pdf_f$ evaluated at $\dict(V_0)$.
Note that the function $\pdf_f$ may be an arbitrary function to real values.
However, in the context of probabilistic programming, we interpret $\pdf_f$ as the density function of distribution $f$ with values in $\pdfvalues$.
Lastly, the value $\dict(V_0)$ is stored in the program state at variable $x$.
The semantics of sample statements are summarised in the following rule:
\begin{equation}\label{eq:sample-semantics}
    \frac{\forall i \colon \progstate(E_i) = V_i \land  V_i \neq \none \quad \quad V_0 \in \strings \quad \quad V = \dict(V_0) \land V \neq \none}%
{(\progstate, x = \kw{sample}(E_0, f(E_1,\dots,E_n))) \opsemrel \progstate[x \mapsto V, \probvar \mapsto \progstate(\probvar) \cdot \pdf_f(V; V_1,\dots,V_n)]}
\end{equation}
Finally, the meaning of a program $S$ is defined as the mapping of trace to density.
This definition is well-defined, because the semantics are deterministic for each trace.
If a trace leads to errors or non-termination, then the density is undefined.

\begin{definition}\label{def:prog-semantics}
{\revcolor The density of a program $S$ is the function $\bar{p}_S\colon \dicts \times \progstates \to \pdfvalues$ given by 
\begin{displaymath}
    \bar{p}_S(\dict, \sigma_0) \coloneqq 
    \begin{cases}
        \sigma(\probvar) & \text{if }\;\exists\, \progstate\in \progstates \colon (\sigma_0, S) \opsemrel \progstate \\
        \undefval & \text{otherwise}
    \end{cases}
\end{displaymath}
For the initial program state $\sigma_0 = (x_1\mapsto \none,\dots,x_n\mapsto \none, \probvar \mapsto 1)$, mapping all program variables to $\none$, we define $p_S\colon \dicts \to \pdfvalues$ with $p_S(\dict) = \bar{p}_S(\dict,\sigma_0).$
}
\end{definition}


\begin{figure}[h]
\centering
    \begin{minipage}[b]{.45\columnwidth}
\begin{lstlisting}[style=Python, numbers=none, caption=Simple probabilistic program with stochastic branching., label=lst:semantics]
p = sample("p", Uniform(0,1))
x = sample("x", Bernoulli(p))
if x == 1 then
    y = sample("y", Bernoulli(0.25))
else
    z = sample("z", Bernoulli(0.75))
\end{lstlisting}   
    \end{minipage}%
    \hspace{5mm}
    \begin{minipage}[b]{.45\columnwidth}
\begin{lstlisting}[style=Python, numbers=none, caption={Program of \cref{lst:semantics} rewritten with dynamic addresses.}, label=lst:semantics2]
p = sample("p", Uniform(0,1))
x = sample("x", Bernoulli(p))
if x == 1 then
    addr = "y"; prob = 0.25;
else
    addr = "z"; prob = 0.75;
r = sample(addr, Bernoulli(prob))
\end{lstlisting}  
    \end{minipage}
\end{figure}
Consider the simple probabilistic program in \cref{lst:semantics} {\revcolor (Note that we slightly deviate from the formal syntax in listings for better readability)}.
The density defined by our operational semantics is given below, where $\dirac{v}$ denotes the delta function whose value is 1 if $v$ equals $\true$ and 0 otherwise.
\begin{align*}
    p(\dict) =&\,\, \pdf{}_\textnormal{Uniform}(\dict(\addr{p});0,1) \\
        &\cdot \pdf_\textnormal{Bernoulli}(\dict(\addr{x});\dict(\addr{p})) \\
        &\cdot \left(\dirac{\dict(\addr{x})}\pdf_\textnormal{Bernoulli}(\dict(\addr{y});0.25) +(1 - \dirac{\dict(\addr{x})})  \right) \\
        &\cdot \left(\dirac{\dict(\addr{x})} +(1 - \dirac{\dict(\addr{x})}) \pdf_\textnormal{Bernoulli}(\dict(\addr{z});0.75)  \right).
\end{align*}

As the addresses in \cref{lst:semantics} are constant and identical to the program variables, the program is well-described by many existing PPL semantics.
However, in \cref{lst:semantics2}, we rewrite the program with \emph{dynamic addresses} which results in the same above density.
Most formal PPLs studied in prior work do not support user-defined addresses for sample statements and cannot interpret a program with three sample statements as a four-dimensional joint density.
{\revcolor In contrast, address-based semantics like those of \citet{lee2019towards} and \citet{lew2023stochasticprobs}, denotationally modelling Pyro and Gen, respectively, correctly interpret the program with dynamic addresses.
This is also true for our extensions of \citet{gorinova2019slicstan}.
While closely related to the aforementioned denotational models, we opted for an operational approach to more directly model the implementations of imperative PPLs.
}

\begin{figure}[h]
\centering
    \begin{minipage}[b]{.45\columnwidth}
\begin{lstlisting}[style=Python, numbers=none, caption=Program with random address defining an unbounded number of variables., label=lst:poisson]
n = sample("n", Poisson(5))
x = sample("x_"+str(n), Normal(0.,1.))
\end{lstlisting}
    \end{minipage}%
    \hspace{5mm}
    \begin{minipage}[b]{.45\columnwidth}
\begin{lstlisting}[style=Python, numbers=none, caption= Program with while loop implementing a Geometric distribution., label=lst:geometric]
b = true; i = 0;
while b do
  i = i + 1
  b = sample("b_"+i, Bernoulli(0.25))
\end{lstlisting}  
    \end{minipage}
\vspace{-5mm}
\end{figure}

{\revcolor Note that probabilistic programs with dynamic addresses and while loops pose a significant theoretical challenge, because both cases easily lead to an unbounded number of addresses as can be seen in \cref{lst:poisson} and \cref{lst:geometric}}.

%
%

\subsection{Remark on a Measure-Theoretic Interpretation}\label{sec:measure}
Even though it suffices to consider the explicit computation of the density of a probabilistic program in this work, a measure-theoretic interpretation of the presented operational semantics is still worth exploring.
The density of a measure $\mu$ with respect to some reference measure $\nu$ is formally given by the Radon-Nikodym derivative $p$ such that $\mu(A) = \int_A p \;d\nu$.
Thus, the measure-theoretic interpretation of the operational semantics is described by a construction of a measure space of traces $(\dicts, \Sigma_\dicts, \nu)$ such that the density $p_S$ is measurable.
Then, $\mu_S(A) = \int_A p_S \;d\nu$ is the measure denoted by the program $S$.
As our results do not rely on measure theoretic assumptions, we refer to Appendix~C for a sketch of this measure space.

\subsection{Remark on Conditioning}
In the presented semantics conditioning is only supported via soft-constraints~\cite{staton2016semanticshigherorder}.
This means that the input trace $\dict$ is assumed to be fixed to observed data at certain addresses.
If those addresses are encountered during execution at sample statements the data likelihood is multiplied to $\sigma(\probvar)$.
Informally, the program density $p_S$ corresponds to the unnormalised posterior density.
Again, a rigorous treatment of the normalisation of $p_S$ is not required for our results.

\section{Static Factorisation for Programs without Loops}\label{sec:without-loops}
The goal of this work is to statically find a factorisation of the density implicitly defined by a probabilistic program in terms of addresses.
This static analysis will operate on the control-flow graph (CFG) of a program.
We begin by considering programs without loops first and describe how we can formally translate such a program to its corresponding CFG.
We will equip this CFG with semantics equivalent to the operational semantics of the original program.
The following construction is fairly standard but necessary for formally verifying the correctness of the factorisation.

\subsection{Control-Flow Graph}\label{sec:cfg}

A CFG has five types of nodes: start, end, assign, branch, and join nodes. Branch nodes have two successor nodes, end nodes have no successor, all other nodes have one.
Only join nodes have multiple predecessors.
The CFG contains exactly one start node from which every other node is reachable. Further, it contains exactly one end node that is reachable from any other node.
We describe the recursive translation rules for program $S$ as diagrams, which serve to convey the general idea of the construction.
For a formal description see Appendix~A.
Sub-graphs are drawn by circular nodes.
Branch and join nodes are drawn as diamond nodes, the latter filled black.
Assign, start, and end nodes have rectangular shape.
In text, branch nodes are written as $\text{Branch}(E)$.
Assign nodes are denoted with $\text{Assign}(x = E)$ and $\text{Assign}(x = \kw{sample}(E_0, \dots))$.

\begin{itemize}
    \item The CFG of a skip-statement, $S = \kw{skip}$, consists of only start and end nodes:\newline
    \includegraphics[scale=0.9]{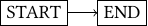}

    \item  The CFG of an assignment, $x = E$, or sample statement $x = \kw{sample}(E_0, f(E_1,\dots,E_n))$, is a sequence of start, assign, and end nodes:\newline
    \raisebox{-0.5\height}{\includegraphics[scale=0.9]{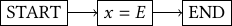}}
    \hspace{5mm}
    \raisebox{-0.5\height}{\includegraphics[scale=0.9]{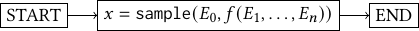}}

    \item The CFG for a sequence of statements, $S = S_1; S_2$, is recursively defined by "stitching together" the CFGs $\cfg_1$ and $\cfg_2$ of $S_1$ and $S_2$:\newline
    \includegraphics[scale=0.9]{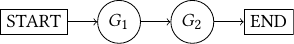}


     \item The CFG of an if statement, $S=(\kw{if}\;E\;\kw{then}\;S_1\,\kw{else}\;S_2)$, is also defined in terms of the CFGs of $S_1$ and $S_2$ together with a branch node $B = \textnormal{Branch}(E)$ and a join node $J$:  \newline
    \includegraphics[scale=0.9]{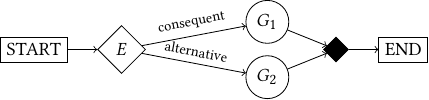}

    
    We denote the successor nodes of the branch node $B$ as $\textnormal{consequent}(B)$ and $\textnormal{alternative}(B)$.
    Furthermore, we say the tuple $(B, J)$ is a branch-join pair and define $\text{BranchJoin}(\cfg)$ to be the set of all branch-join pairs of $\cfg$.
    {\revcolor Lastly, let $\textnormal{condexp}(B)$ denote the conditional expression~$E$.}
\end{itemize}

\paragraph{CFG Semantics.}
For a control-flow graph $\cfg$ we define the small-step semantics for a trace~$\dict$ as a relation that models the transition from node to node depending on the current program state:
\[(\progstate, N) \cfgrel (\progstate',N')\]
The semantics for the CFG of program $S$ are set-up precisely such that they are equivalent to the operational semantics of program $S$.
The transition rules depending on node type and program state are given below, {\revcolor where we denote the unique child of non-branch nodes $N$ in $\cfg$ with \text{successor}(N)}:

\begin{displaymath}
    \frac{N = \text{START} \quad N' = \text{successor}(N) }{ (\progstate, N) \cfgrel (\progstate, N')} \quad \frac{N = \text{Assign}(x=E) \quad N' = \text{successor}(N)}{(\progstate, N) \cfgrel (\progstate[x \mapsto \progstate(E)],N')}
\end{displaymath}

\begin{displaymath}
    \frac{N = \text{Join} \quad N' = \text{successor}(N)}{(\progstate, N) \cfgrel (\progstate, N')} \, \frac{N = \text{Branch}(E) \quad \progstate(E) =  \true}{(\progstate, N) \cfgrel (\progstate,\text{consequent}(N))} \, \frac{N = \text{Branch}(E) \quad \progstate(E) = \false}{(\progstate, N) \cfgrel (\progstate,\text{alternative}(N))}
\end{displaymath}
\begin{equation}\label{eq:sample-semantics-cfg}
    \frac{
    \begin{aligned}
    {N = \text{Assign}(x = \kw{sample}(E_0, f(E_1,\dots,E_n))) \quad N' = \text{successor}(N)} \\
    {\forall i \colon \progstate(E_i) = V_i \land V_i \neq \none \quad V_0 \in \strings\quad V = \dict(V_0) \land V \neq \none}
    \end{aligned}}%
    {(\progstate, N) \cfgrel (\progstate[x \mapsto V, \probvar \mapsto \progstate(\probvar) \cdot \pdf_f(V; V_1,\dots,V_n)],N')}
\end{equation}
Finally, the meaning of a CFG is a mapping from trace to density.
It is defined for traces $\dict$ if there exists a path of node transitions from start to end node.
\begin{definition}
{\revcolor
The density of a CFG~$\cfg$ is the function $\bar{p}_\cfg\colon \dicts \times \progstates \to \pdfvalues$ given by
\begin{displaymath}
    \bar{p}_\cfg(\dict,\sigma_0) \coloneqq
    \begin{cases}
        \sigma(\probvar) & \text{if}\;\;  \exists \; \progstate \in \progstates, n \in \Nats \colon \forall i=1,\dots,n\colon \exists  \; \progstate_i, N_i\colon  \\
        & (\progstate_0, \text{START}) \cfgrel \dots \cfgrel(\progstate_i, N_i) \cfgrel\dots \cfgrel(\progstate, \text{END}) \\
        \undefval & \text{otherwise.}
    \end{cases}
\end{displaymath}
For initial program state $\progstate_0$ as in \Cref{def:prog-semantics},  we define $p_G\colon \dicts \to \pdfvalues$ with $p_G(\dict) = \bar{p}_G(\dict,\sigma_0).$
}
\end{definition}

\begin{proposition}\label{theorem:CFG-semantics-equivalence}
For all programs $S$ without while loops and corresponding control-flow graph $\cfg$, it holds that for all traces $\dict$ {\revcolor and all initial states $\sigma_0$}
\[{\revcolor \bar{p}_S(\dict,\sigma_0) = \bar{p}_\cfg(\dict,\sigma_0).}\]
\end{proposition}
\begin{proof}
By structural induction. See Appendix B.1.
\end{proof}

\subsection{Static Provenance Analysis}\label{sec:staticprov}

The provenance of variable $x$ in CFG node $N$ are all the addresses $A \subseteq \strings$ that contribute to the computation of its value in node $N$.
In this section, we define an algorithm that operates on the CFG and statically determines the provenance for any variable $x$ at any node $N$.
As this is a static algorithm, it is an over-approximation of the true provenance.
In this section, we will prove soundness of the presented algorithm.
We show that indeed the value of variable $x$ at node~$N$ can be computed from the values in the trace $\dict$ at the addresses which make up the provenance.
But first, we introduce some useful concepts which are similar to definitions found in data-flow analysis~\cite{cooper2011engineeringacompiler, hennessy2011computerdatacontrol}.
\begin{definition}
    For a CFG $\cfg$, let $\textnormal{AssignNodes}(\cfg, x)$ be all nodes $N'$ in $\cfg$ that assign variable $x$, $N' = \textnormal{Assign}(x = \dots)$.
    The set of \emph{reaching definitions} for variable $x$ in node $N$ is defined as
    \begin{displaymath}
        \RD(N,x) = \{N' \in \textnormal{AssignNodes}(\cfg, x) : \exists \text{ path } (N', N_1, \dots, N_n, N), \; N_i \notin \textnormal{AssignNodes}(\cfg,x)\}.
    \end{displaymath}
\end{definition}
Intuitively, the set of reaching definitions are all nodes that could have written the value of $x$ in the program state before executing node $N$.

\begin{definition}
    For CFG $\cfg$ and node $N$ the set of \emph{branch parents} is defined as
\begin{displaymath}
    \CP(N) = \{B: (B, J) \in \textnormal{BranchJoin}(\cfg) \land \exists \text{ path } (B,\dots, N,\dots,J)\}.
\end{displaymath}
    A branch node $B$ and join node $J$ are a branch-join pair $(B, J) \in \textnormal{BranchJoin}(\cfg)$ if they belong to the same if-statement (see the translation rule for if-statements to CFGs).
\end{definition}
The set of branch parents are all parent branch nodes of $N$ that determine if the program branch that $N$ belongs to is executed.

Lastly, we define the set of all strings that the address expression of a sample statement may evaluate to.

\begin{definition}
    For a sample CFG node $N = \textnormal{Assign}(x = \kw{sample}(E_0, f(E_1,\dots,E_n)))$, we define the set of all possible addresses as
    \begin{displaymath}
        \textnormal{addresses}(N) = \{\progstate(E_0): \progstate \in \progstates\} \subseteq \strings.
    \end{displaymath}
\end{definition}

\begin{algorithm}[h]
\caption{Computing the provenance set $\staticprov(N,x)$ statically from the CFG.}\label{alg:staticprov}
\begin{algorithmic}[1]
\STATE \textbf{Input} CFG node $N$, variable $x$
\STATE ${\revcolor \textnormal{result}} \gets \emptyset$
\STATE queue $\gets [(N,x)]$
\WHILE{queue is not empty}
    \STATE $(N, x)$ $\gets$ queue.pop()
    \FOR{$N' \in \RD(N,x)$}
        \IF{$N' = \text{Assign}(x = \kw{sample}(E_0, f(E_1,\dots,E_n)))$}
            \STATE ${\revcolor \textnormal{result}} \gets {\revcolor \textnormal{result}}\cup \text{addresses}(N')$\hfill $\rhd$ in practice: ${\revcolor \textnormal{result}} \gets {\revcolor \textnormal{result}} \cup \{N'\}$
            \STATE $E' \gets E_0$ \hfill $\rhd$ $E_0$ comes from pattern-matching $N'$
        \ELSIF{$N' = \text{Assign}(x = E)$}
            \STATE  $E' \gets E$ \hfill $\rhd$ $E$ comes from pattern-matching $N'$
        \ENDIF
        \FOR{$y \in \text{vars}(E')$}
            \IF{is\_unmarked($(N',y)$)}
                \STATE mark($(N',y)$)
                \STATE queue.push($(N',y)$)
            \ENDIF
        \ENDFOR
        \FOR{$N_\text{bp} \in \CP(N')$ \hfill $\rhd$ {$N_\text{bp} = \text{Branch}(E_{N_\text{bp}})$}}
            \FOR{$y \in \text{vars}(E_{N_\text{bp}})$}
                 \IF{is\_unmarked($(N_\text{bp},y)$)}
                    \STATE mark($(N_\text{bp},y)$)
                    \STATE queue.push($(N_\text{bp},y)$)
                \ENDIF
            \ENDFOR
        \ENDFOR
    \ENDFOR
\ENDWHILE
\RETURN {\revcolor \textnormal{result}}
\end{algorithmic}
\end{algorithm}

Now we are able to define the algorithm to statically approximate the provenance of any variable~$x$ at any node $N$ denoted by $\staticprov(N,x)$ in \cref{alg:staticprov}.
This algorithm is inspired by standard methods for dependence analysis.
In the algorithm, we use $\textnormal{vars}(E)$, which is defined as the set of all program variables in expression $E$.
It is often useful to determine the provenance of an entire expression instead of only a single variable.
For expressions $E$, we lift the definition of $\staticprov$
\begin{equation}\label{eq:prov-expr}
    \staticprov(N,E) = \medcup_{y \in \textnormal{vars}(E)} \staticprov(N,y).
\end{equation}

The algorithm works as follows.
If we want to find $\staticprov(N,x)$, we first have to find all reaching definitions for $x$.
If the reaching definition $N'$ is a sample node, $x = \kw{sample}(E_0,\dots)$, then the value of $x$ can only {\revcolor depend} on the addresses generated by $E_0$, $\textnormal{addresses}(N')$, and the provenance of $E_0$.
If $N'$ is an assignment node, $x = E$, we have to recursively find the provenance of the expression~$E$.
Lastly, there can be multiple reaching definitions for $x$ in different branches and the value of $x$ also depends on the branching condition.
Thus, we also find all branch parents and recursively determine their provenance and add it to the final result.

In practice the potentially infinite set $\textnormal{addresses}(N)$ can be represented by the CFG node $N$ itself.
As remarked on line 8, an implementation of \cref{alg:staticprov} would return a set of CFG sample nodes rather than a potentially infinite set of strings. {\revcolor We call this version of the algorithm $\staticprovnode$ and we have $\staticprov(N,x) = \bigcup_{N' \in \staticprovnode(N,x)} \textnormal{addresses}(N')$}.

We get a clearer sense of the introduced definitions and \Cref{alg:staticprov} by considering a simple program in \Cref{fig:rd_bp_prov}.
The right panel shows the CFG with reaching definitions and branch parents for several nodes drawn as arrows.
These arrows illustrate how \Cref{alg:staticprov} computes the provenance of variables \texttt{m} and \texttt{s} at the sample node for {\revcolor address~$\addr{x}$}.
Note that if we were to change the if statement to \texttt{(if (b == 1) then (m = 1) else (m = 1))}, then the {\revcolor address} $\addr{b}$ would still be in the provenance set of \texttt{m}.
This illustrates that the presented algorithm can only over-approximate the provenance and not compute it exactly.

\begin{figure}[h]
    \centering
\begin{minipage}[c]{0.5\linewidth}

\begin{lstlisting}[style=Python, numbers=none]
b = sample("b", Bernoulli(0.5))
s = sample("s", InverseGamma(1.,1.))
if b == 1 then
    m = sample("mu", Normal(0.,1.))
else
    m = 1
x = sample("x", Normal(m, s))
\end{lstlisting}
{\small
\begin{align*}
    &N_7 = \textnormal{Assign}(\mathtt{x} = \dots),\,\, 
    \evalf{N_7}{\mathtt{s}}(\dict) = \dict(\addr{s}),\\
    &\evalf{N_7}{\mathtt{m}}(\dict) = {\revcolor \begin{cases}
                \dict(\addr{mu}) & \text{if } \dict(\addr{b}) = \true, \\
                1 & \text{otherwise.}
            \end{cases}} \\
\end{align*}
}
\end{minipage}
\begin{minipage}[c]{0.45\linewidth}
    \includegraphics[scale=0.85]{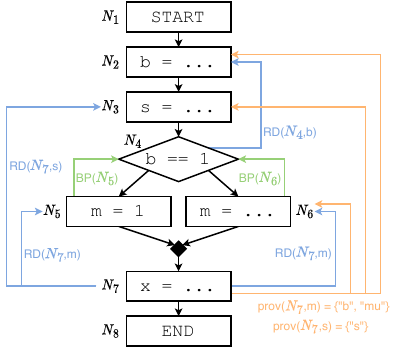}
\end{minipage}
    \caption{Reaching definitions, branch parents, and the provenance set for a simple example. The evaluation functions introduced in \cref{sec:staticprov} are also given for node $N_7$.}
    \label{fig:rd_bp_prov}
\end{figure}

\pagebreak

\paragraph{Soundness.}

Recall that the purpose of finding the provenance of variable $x$ at node $N$ is to determine which addresses contribute to the computation of its value.
To prove that \cref{alg:staticprov} indeed finds a sufficient set of addresses, first, we have to precisely define what we mean that a set of addresses contribute to the computation.
Suppose that there is a function $f$ that maps a trace $\dict$ to a value $v$.
The set of addresses $A$ contribute to the computation of value $v$, if any other trace $\dict'$ that has the same values as $\dict$ for all addresses $\alpha \in A$, is mapped to the same value $f(\dict') = v$.
In other words, if we change the trace $\dict$ at an address $\beta \in \strings \setminus A$, the value of $f(\dict)$ remains unchanged.

\begin{definition}\label{def:evalf}
    For a function from traces to an arbitrary set $B$, $f: \dicts \to B$, we write $f \in \restricedfs{A}{B}$ for a subset of addresses $A \subseteq \strings$, if and only if
    \begin{displaymath}
        \forall \alpha \in A\colon \dict(\alpha) = \dict'(\alpha) \implies f(\dict) = f(\dict').
    \end{displaymath}
    That is, changing the values of $\dict$ at addresses $\beta \in \strings \setminus A$ does not change the value $f(\dict)$.
\end{definition}
The next lemma states that we can always over-approximate the set of addresses that contribute to the computation of the values $f(\dict)$.
\begin{lemma}
    If $A \subseteq A'$ and $f \in \restricedfs{A}{B}$, then $f \in \restricedfs{A'}{B}$. In other words,
\begin{equation}\label{eq:restrictedfs-subset}
    A \subseteq A' \implies  \restricedfs{A}{B} \subseteq \restricedfs{A'}{B}.
\end{equation}
\end{lemma}
If we combine two functions with provenance $A_1$ and $A_2$ respectively, then the resulting function has provenance $A_1 \cup A_2$. This is shown in the following lemma.
\begin{lemma}\label{lemma:restricted-composition}
    Let $f_1 \in \restricedfs{A_1}{B_1}$, $f_2 \in \restricedfs{A_2}{B_2}$, $h\colon B_1 \times B_2 \to C$. Then \[\dict \mapsto h(f_1(\dict), f_2(\dict)) \in \restricedfs{A_1 \cup A_2}{C}.\]
\end{lemma}

This brings us to the main soundness result of \cref{alg:staticprov}.
\cref{theorem-evalf} states that we can equip each CFG node $N$ with evaluation functions that depend on the addresses determined by $\staticprov$ in the sense of \cref{def:evalf}.
The evaluation functions directly map the trace $\dict$ to the values of variables at node $N$ such that the values agree with the small-step CFG semantics relation~$\cfgrel$.
Importantly, while the operational semantics are defined for each individual trace, the evaluation functions work for all traces.
In \cref{fig:rd_bp_prov}, concrete examples of these evaluation functions for a simple program can be found.

\begin{proposition}\label{theorem-evalf}
    Let $\cfg$ be the CFG of a program $S$ without while loop statements.
    For each node $N$ and variable~$x$, there exists an evaluation function
    \[\evalf{N}{x} \in \restricedfs{\staticprov(N,x)}{\values},\] such that for all traces $\dict$, {\revcolor initial program state $\progstate_0$ as in \Cref{def:prog-semantics}}, and execution sequence
    \[(\progstate_0, \text{START}) \cfgrel \cdots  \cfgrel (\progstate_i, N_i) \cfgrel \cdots \cfgrel (\progstate_l, \text{END})\]
    we have
    \begin{displaymath}
    \progstate_i(x) = \evalf{N_i}{x}(\dict).
    \end{displaymath}
    \end{proposition}
\begin{proof}
   The CFG $G$ of a program without while loops is a directed acyclic graph and the statement is proven by induction on $G$. See Appendix B.2.
\end{proof}
Intuitively, $\evalf{N}{x}$ computes the value of $x$ before executing $N$.
These evaluation functions $\evalf{N}{x}$ can be lifted to evaluation functions for expressions $\evalf{N}{E}$:
\begin{displaymath}
    \evalf{N}{c}(\dict) = c, \quad  \evalf{N}{g(E_1,\dots,E_n)}(\dict) = g(\evalf{N}{E_1}(\dict),\dots,\evalf{N}{E_n}(\dict)).
\end{displaymath}
By \cref{lemma:restricted-composition}, if $\evalf{N}{y} \in \restricedfs{\staticprov(N,y)}{\values}$ for all $y \in \textnormal{vars}(E)$, then  $\evalf{N}{E} \in \restricedfs{\staticprov(N,E)}{\values}$.

\subsection{Static Factorisation Theorem for Programs without Loops}\label{sec:theorem-without-loops}
With \cref{theorem-evalf}, we can prove the first factorisation theorem by combining the evaluation functions $\evalf{N}{x}$ to express the program density $p$ in a factorised form.
\begin{theorem}\label{theorem:simple-factor}
    Let $\cfg$ be the CFG for a program $S$ without while loop statements.
    Let $N_1,\dots,N_K$ be all sample nodes in $\cfg$.
    For each sample node $N_k = \textnormal{Assign}(x_k = \kw{sample}(E_0^k, f^k(E_1^k,\dots,E_{n_k}^k)))$, let
    \[A_k = \textnormal{addresses}(N_k) \cup \medcup_{i = 0}^{n_k} \staticprov(N_k, E_i^k) \cup \medcup_{N' \in \CP(N_k)}\staticprov(N',{\revcolor \textnormal{condexp}(N')}).\]
    Then, there exist functions $p_k \in  \restricedfs{A_k}{\pdfvalues}$ such that {\revcolor for all $\dict\in\dicts$}, if $p_S(\dict) \neq \undefval$, then
    \[p_S(\dict) = p_\cfg(\dict) = \prod_{k=1}^K p_k(\dict).\]
\end{theorem}
\begin{proof}
For each sample node $N_k$, define  $b_k(\dict) \coloneqq \bigwedge_{N' \in \CP(N_k)} t^k_{N'}(\evalf{N'}{{\revcolor \textnormal{condexp}(N')}}(\dict))$ where $t^k_{N'}(v) = v$ if $N_k$ is in the consequent branch of $N'$ and $t^k_{N'}(v) = \neg v$ if $N_k$ is in the alternate branch.
By \Cref{lemma:restricted-composition}, {\revcolor we have
$b_k \in \restricedfs{\medcup_{N' \in \CP(N_k)}\staticprov(N', \textnormal{condexp}(N'))}{\{\true,\false\}}.$

Above, $\evalf{N'}{ \textnormal{condexp}(N')}(\dict)$} are precisely the evaluations of the branch conditions for node $N_k$.
Thus, $b_k(\dict)$ is $\true$ if $N_k$ is in the execution sequence for $\dict$ else $\false$.
Define the factor $p_k$ as
\begin{equation}\label{eq:factor-construction}
    p_k(\dict) \coloneqq \dirac{b_k(\dict)} \pdf_{f^k}\left(\dict(\evalf{N_k}{E_0^k}(\dict));\evalf{N_k}{E_1^k}(\dict), \dots, \evalf{N_k}{E_{n_k}^k}(\dict)\right) + (1-\dirac{b_k(\dict)}).
\end{equation}
By \Cref{theorem-evalf}, $\evalf{N_k}{E_i^k} \in \restricedfs{\staticprov(N_k, E_i^k)}{\values}$ and $\evalf{N_k}{E_0^k} \in \restricedfs{\staticprov(N_k, E_0^k)}{\textnormal{addresses}(N_k)}$.
Thus, it follows from \Cref{lemma:restricted-composition} that $p_k \in \restricedfs{A_k}{\pdfvalues}$ by construction.
Lastly, {\revcolor from \Cref{theorem-evalf}} we know that for the end-state $\sigma_l(\probvar) = \evalf{\text{END}}{\probvar}(\dict) = p_\cfg(\dict)$.
Since $b_k(\dict) = \true$ if and only if $N_k$ is in the execution sequence for $\dict$ and {\revcolor since the $\pdf_{f^k}(\dots)$ term in \eqref{eq:factor-construction} is precisely the value that is multiplied to $\probvar$ at sample statements \eqref{eq:sample-semantics}}, we have $p_\cfg(\dict) = \prod_{k=1}^K p_k(\dict)$.
\end{proof}

Again, consider the example program in \cref{fig:rd_bp_prov}. The density of this model factorises like below:
\begin{align*}    
    p(\dict) = &\,\,\pdf_\text{Bernoulli}(\dict(\addr{b}); 0.5)  & \in \restricedfs{\{\addr{b}\}}{\pdfvalues}\\
                & \cdot\pdf_\text{InverseGamma}(\dict(\addr{s}); 1, 1) & \in \restricedfs{\{\addr{s}\}}{\pdfvalues}\\
                & \cdot\left(\dirac{\dict(\addr{b})}\pdf_\text{Normal}(\dict(\addr{mu}); 0, 1) + (1-\dirac{\dict(\addr{b})})\right) & \in \restricedfs{\{\addr{b},\addr{mu}\}}{\pdfvalues}\\
                &\cdot \pdf_\text{Normal}(\dict(\addr{x}); \evalf{N_7}{\mathtt{m}}(\dict), \evalf{N_7}{\mathtt{s}}(\dict)) & \in \restricedfs{\{\addr{x},\addr{b},\addr{mu}, \addr{s}\}}{\pdfvalues}
\end{align*}
This is exactly the factorisation obtained from \cref{theorem:simple-factor}.
However, if we would replace the if-statement with a nonsensical one, \texttt{(if (b == 1) then (m = 1) else (m = 1))}, then the fourth factor is in $\restricedfs{\{\addr{x},\addr{b}, \addr{s}\}}{\pdfvalues}$ with $\evalf{N_7}{\mathtt{m}}(\dict) = \textnormal{ife}(\dict(\addr{b}), 1, 1) = 1$, {\revcolor where we use $\textnormal{ife}$ as a short-hand for the if-else case distinction.}
In this case, the static analysis over-approximates the true provenance with the spurious dependency on $\addr{b}$.

\subsubsection{Equivalence to Bayesian Networks}\label{sec:equivalence-to-bayesian-networks}
If we restrict the way sample statements are used in our PPL, we can establish a known equivalence to Bayesian networks.
A Bayesian network~\cite{koller2009pga} is a directed acyclic graph $\mathcal{G}$ over a finite set of nodes $X_1, \dots, X_n$ which represent random variables, such that their joint distribution factorises according to the graph $\mathcal{G}$,
\[P(X_1,\dots,X_n) = \prod_{i=1}^n P(X_i\;|\; \textnormal{parents}_\mathcal{G}(X_i)).\]
If we assume that all sample statements of a program have a constant unique address~$\alpha_k$, $N_k = \textnormal{Assign}(x_k = \kw{sample}(\alpha_k, f^k(E_1^k,\dots,E_{n_k}^k)))$, then we can identify each factor $p_k\in \restricedfs{A_k}{\pdfvalues}$ of \cref{theorem:simple-factor} with the unique address $\alpha_k \in \strings$.
We construct the Bayesian network $\mathcal{G}$ by mapping each address to a random variable $X_{\alpha_k}$.
Since $\alpha_k$ is unique, $p_k$ is the only relevant factor for $X_{\alpha_k}$ and can be interpreted as conditional density function of $X_{\alpha_k}$.
To see this, rewrite
\[p_k(\dict) = \dirac{b_k(\dict)} \pdf_{f^k}\left(\dict(\alpha_k);\evalf{N_k}{E_1^k}(\dict), \dots, \evalf{N_k}{E_{n_k}^k}(\dict)\right) + (1-\dirac{b_k(\dict)})\mathbf{1}_{\none}\left(\dict(\alpha_k)\right),\]
where the function $\mathbf{1}_{\none}(v)$ equals 1 if $v = \none$ else 0.
If the node $N_k$ is not executed, then the value of $X_{\alpha_k}$ is  assumed to be $\none$ and we add the $\mathbf{1}_{\none}\left(\dict(\alpha_k)\right)$ term to $p_k$.
As $\pdf_{f_k}$ is a probability density function and $\mathbf{1}_{\none}$ is the density of a Dirac distribution centered at $\none$, the factor $p_k$ is also a density function for each choice of values in $\dict$ at addresses $\alpha \in A_k \setminus \{\alpha_k\}$.
Thus, we introduce the edge $\alpha_j \to \alpha_k$ if $\alpha_j \in A_k \setminus \{\alpha_k\}$, such that $\textnormal{parents}_\mathcal{G}(X_{\alpha_k}) = A_k \setminus \{\alpha_k\}$ and
\[P(X_{\alpha_k}\,|\, \textnormal{parents}_\mathcal{G}(X_{\alpha_k})) = p_k(\{\alpha \mapsto X_\alpha\colon \alpha \in A_k\}).\]
In the above construction, it is important that each sample statement has a \emph{unique} address.
This guarantees no cyclic dependencies and a well-defined conditional probability distributions.

The Bayesian network equivalent to the example program of \cref{fig:rd_bp_prov} is drawn below.
Since the value of the program variable $\mathtt{m}$ depends on which program branch is taken during execution, the random variable $X_\addr{x}$ not only depends on $X_\addr{mu}$, but also on $X_\addr{b}$.\newline
\begin{minipage}[c]{0.2\linewidth}
\includegraphics[scale=1.]{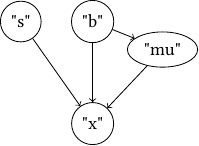}
\end{minipage}
\begin{minipage}[c]{0.75\linewidth}
\begin{align*}
    P(X_\addr{b}) &= \pdf_\textnormal{Bernoulli}(X_\addr{b}; 0.5)\\
    P(X_\addr{s}) &= \pdf_\textnormal{InverseGamma}(X_\addr{s}; 1, 1)\\
    P(X_\addr{mu} | X_\addr{b}) &= \dirac{X_\addr{b}}\pdf_\text{Normal}(X_\addr{mu}; 0, 1) + (1-\dirac{X_\addr{b}}) \mathbf{1}_{\none}(X_\addr{mu})\\
    P(X_\addr{x} | X_\addr{b}, X_\addr{mu}, X_\addr{s}) &= \pdf_\textnormal{Normal}(X_\addr{x}; \textnormal{ife}(X_\addr{b}, X_\addr{mu}, 1), X_\addr{s})\\ 
\end{align*}
\end{minipage}

\subsubsection{Equivalence to Markov Networks}\label{sec:markov-without-loops}

If the sample addresses are not unique or not constant, then the factorisation is in general not a Bayesian network.
However, we can show equivalence to Markov networks.
A Markov network~\cite{koller2009pga} is an \emph{undirected} graph $\mathcal{H}$ over random variables $\vv{X}$ that represents their dependencies.
We consider Markov networks where the joint distribution of $\vv{X}$ factorises over a set of cliques $\mathcal{D} = \{D_1,\dots,D_n\} \subseteq \textnormal{cliques}(\mathcal{H})$, with $\bigcup_{i=1}^n D_i = \vv{X}$:
\[P(\vv{X}) = \prod_{i = 1}^n \phi_i(D_i).\]  
We again identify a random variable $X_\alpha$ for each address $\alpha \in \medcup_{k=1}^K A_k$ and construct a Markov network $\mathcal{H}$ by connecting node $X_\alpha$ to $X_\beta$ if $\alpha \in A_k \land \beta \in A_k$ for any $k$.
Thus, $D_k=\{X_\alpha \colon \alpha \in A_k\}$ forms a clique and $p_S$ factorises over $\mathcal{H}$ with $\phi_{k}(D_k) = p_k(\{\alpha \mapsto X_\alpha \colon \alpha \in A_k\})$.

Consider the program with constant but non-unique sample addresses in \cref{lst:hurricane1}.
With nine sample statements we get a Markov network over $\{X_\addr{F}, X_\addr{P0}, X_\addr{P1}, X_\addr{D0}, X_\addr{D1}\}$. 
The nine factors correspond to the address sets $\{\addr{F}\}$, $\{\addr{F},\addr{P0}\}$, $\{\addr{F},\addr{P1}\}$, $\{\addr{F},\addr{D0},\addr{P1}\}$, $\{\addr{F},\addr{D1},\addr{P0}\}$,\linebreak $2\times\{\addr{F},\addr{P0},\addr{D0}\}$, and $2\times\{\addr{F},\addr{P1},\addr{D1}\}$.
In the consequent branch of the program, the distribution in the sample statement with address $\addr{P1}$ depends on $\addr{D0}$, while in the alternative branch, the distribution of $\addr{P0}$ depends on $\addr{D1}$.
Thus, there is the dependency cycle $\addr{P0} \to \addr{D0} \to \addr{P1} \to \addr{D1} \to \addr{P0}$, which means the obtained factorisation is not a Bayesian network representation for the program.
Note that the program can also be rewritten with dynamic addresses to eliminate the if statement. 
In this case, the factorisation still admits cyclic dependencies. 
\begin{figure}[h]
    \centering
\begin{lstlisting}[style=Python, numbers=none, caption=Probabilisitic program for the hurricane example of \cite{milch2005contingentbayesnet}., label=lst:hurricane1]
first_city_ixs = sample("F", Bernoulli(0.5))
if first_city_ixs == 0 then
    prep_0 = sample("P0", Bernoulli(0.5))
    damage_0 = sample("D0", Bernoulli(prep_0 == 1 ? 0.20 : 0.80))
    prep_1 = sample("P1", Bernoulli(damage_0 == 1 ? 0.75 : 0.50))
    damage_1 = sample("D1", Bernoulli(prep_1 == 1 ? 0.20 : 0.80))
else
    prep_1 = sample("P1", Bernoulli(0.5))
    damage_1 = sample("D1", Bernoulli(prep_1 == 1 ? 0.20 : 0.80))
    prep_0 = sample("P0", Bernoulli(damage_1 == 1 ? 0.75 : 0.50))
    damage_0 = sample("D0", Bernoulli(prep_0 == 1 ? 0.20 : 0.80))
\end{lstlisting}
\end{figure}




Lastly, while the programs considered up until this point were equivalent to finite Bayesian networks or finite Markov networks, note that the small program in \cref{lst:poisson} is equivalent to a Markov network with an 
infinite number of nodes $\{X_\addr{n}\} \cup \{X_{\textnormal{"x}\_i\textnormal{"}}\colon i \in \Nats\}$.


\section{Static Factorisation for Programs with Loops}\label{sec:loops}

In this section, we describe how the argument for proving the factorisation theorem for programs without loops, can be extended to generalise the result to programs with loops.
\subsection{Unrolled Control-Flow Graph}\label{sec:unrolled-cfg}
The translation rules for programs to their CFG can be extended to support while loops as follows.
Let $\cfg_S$ be the sub-CFG of statement $S$. The CFG of $(\kw{while}\;E\;\kw{do}\;S)$ is given {\revcolor below on the left:}\newline
\includegraphics[scale=0.9]{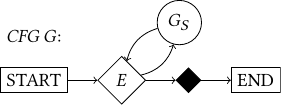}\hspace{3mm}\includegraphics[scale=0.9]{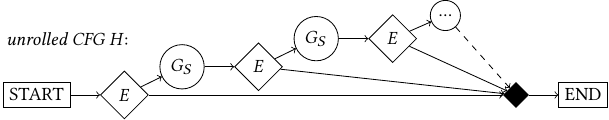}\newline
For this CFG, $\staticprov(N,x)$ is still well-defined by \cref{alg:staticprov}, but the proof of \cref{theorem-evalf} does not work anymore since it requires a directed acyclic graph.

To establish \cref{theorem-evalf} for programs with while loops, we introduce a second type of control-flow graph, the \emph{unrolled CFG}, which again is a directed acyclic graph.
All translation rules are as before except for while loops, which is given below {\revcolor and illustrated above on the right}.\newline
The resulting graph $\cfg$ contains for every $i\in\Nats$ branch nodes $B_i = \textnormal{Branch}(E)$ and \emph{copies} of the sub-CFG $\cfg_i$ of statement $S$.
It contains a single join node $J$ for which we have $(B_i,J) \in \textnormal{BranchJoin}(G)$ for all $i$.
Let $N_\textnormal{first}^i$ be successor of the start node in $G_i$ and $N_\textnormal{last}^i$ the predecessor of the end node in $G_i$.
In the new graph $G$, the first branch node $B_1$ is successor of the start node and the end node is the successor of the join node $J$.
The branch nodes $B_i$ have two successors $N_\textnormal{first}^i$ and $J$.
$N_\textnormal{last}^i$, the last node of sub-graph $G_i$, resumes in the next branch node $B_{i+1}$.

This definition implies that the number of nodes in $\cfg$ is countable.
Furthermore, since the unrolled CFG consists only of start, end, branch, join, and assign nodes, we can equip it with the same semantics as the standard CFG. As before, the semantics are equivalent.

\begin{proposition}\label{theorem:unrolled-CFG-semantics-equivalence}
For all program $S$ with CFG $\cfg$ and unrolled CFG $\unrolledcfg$, it holds that for all traces $\dict$ and {\revcolor initial states $\sigma_0$}
\[{\revcolor \bar{p}_S(\dict,\sigma_0) = \bar{p}_{\cfg}(\dict,\sigma_0) = \bar{p}_{\unrolledcfg}(\dict,\sigma_0).}\]
\end{proposition}
\begin{proof}
    By structural induction. See Appendix B.3.
\end{proof}

As mentioned, \cref{alg:staticprov} still works for the unrolled CFG.
However, since the unrolled CFG contains an infinite number of nodes, in practice, we run this algorithm on the standard CFG with a finite number of nodes avoiding loops by marking visited nodes (line 14 and 19).
We will show that the computations on the standard CFG produce an over-approximation of the provenance set in the unrolled CFG.
But first, we define the connection between the two graph types.
\begin{definition}
    Let $\cfg$ be the CFG and $\unrolledcfg$ the unrolled CFG of program $S$. For each node $M \in \unrolledcfg$, we define $\cfgnode(M) \in \cfg$ as the node from which $M$ was copied in the construction of $\unrolledcfg$.
\end{definition}
The definition of $\cfgnode$ is best understood when considering the example in \cref{fig:unrolled-to-standard}, where the CFG and unrolled CFG of a program are depicted side-by-side and connected with the mapping $\cfgnode$.
\begin{figure}[H]
    \centering
    \includegraphics[scale=1.]{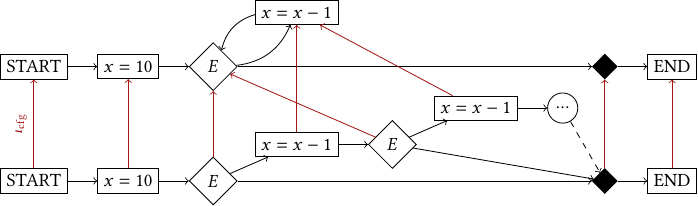}
    \caption{CFG and unrolled CFG for program \texttt{(x = 10; while (x < 10) do (x = x - 1))} and their connection via the map $\cfgnode$ shown with red edges.}
    \label{fig:unrolled-to-standard}
\end{figure}
The following lemma states that we can over-approximate the provenance set for variable $x$ at node~$M$ in the unrolled CFG $\unrolledcfg$ by applying \cref{alg:staticprov} to the corresponding node $\cfgnode(M)$ in $\cfg$.
\begin{lemma}\label{lemma:staticprov-unrolled-cfg}
    Let $\cfg$ be the CFG and $\unrolledcfg$ the unrolled CFG of program $S$. For each node $M \in \unrolledcfg$ following statements hold:
    \begin{equation}\label{eq:rd-subset}
         \{\cfgnode(M')\colon M' \in\RD(M,x)\} \subseteq \RD(\cfgnode(M), x)
    \end{equation}
    \begin{equation}\label{eq:cp-subset}
        \{\cfgnode(M')\colon M' \in\CP(M)\} = \CP(\cfgnode(M))
    \end{equation}
    \begin{equation}\label{eq:staticprov-subset}
        \staticprov(M,x) \subseteq \staticprov(\cfgnode(M),x)
    \end{equation}
\end{lemma}
\begin{proof}
    {\revcolor See Appendix B.4.}
\end{proof}

\subsection{Static Factorisation Theorem for Programs with Loops}\label{sec:theorem-with-loops}

The unrolled CFG has the important property that it is a directed acyclic graph with a potentially infinite number of nodes, but with a single root node.
Further, each node $N\neq \textnormal{START}$ still has exactly one predecessor except for join nodes.
Therefore, \cref{theorem-evalf} can be generalised to unrolled CFGs with only a slight modification to the proof.

\begin{proposition}\label{theorem-evalf-loopy}
    Let $\unrolledcfg$ be the \emph{unrolled} CFG of a program $S$. 
    For each node $M$ and variable $x$, there exists an evaluation function
    \[\evalf{M}{x} \in \restricedfs{\staticprov(M,x)}{\values},\] such that for all traces $\dict$, {\revcolor initial program state $\progstate_0$ as in \Cref{def:prog-semantics}}, and execution sequence
    \[(\progstate_0, \text{START}) \cfgrel \cdots  \cfgrel (\progstate_i, M_i) \cfgrel \cdots \cfgrel (\progstate_l, \text{END})\]
    we have
    \begin{displaymath}
    \progstate_i(x) = \evalf{M_i}{x}(\dict).
    \end{displaymath}
\end{proposition}
\begin{proof}
    The unrolled CFG $\unrolledcfg$ is constructed in a way such that we can prove the statement by well-founded induction on the relation $M_1 \prec M_2 \Leftrightarrow M_1 \textit{ is parent of } M_2$ (i.e. there is an edge from $M_1$ to $M_2$ in $\unrolledcfg$) similar to the proof of \Cref{theorem-evalf}.
     The full proof is given in Appendix B.5.
\end{proof}

As before, \cref{theorem-evalf-loopy} allows us to explicitly construct the factorisation of the program density.

\begin{theorem}\label{theorem:loopy-factor}
    Let $\cfg$ be the CFG for a program $S$.
    Let $N_1,\dots,N_K$ be all sample nodes in $\cfg$.
    For each sample node $N_k = \textnormal{Assign}(x_k = \kw{sample}(E_0^k, f^k(E_1^k,\dots,E_{n_k}^k)))$, let
    \begin{equation}\label{eq:factor-def}
        [A_k = \textnormal{addresses}(N_k) \cup  \medcup_{i = 0}^{n_k} \staticprov(N_k, E_i^k) \cup \medcup_{N' \in \CP(N_k)}\staticprov(N',{\revcolor \textnormal{condexp}(N')}).
    \end{equation}
    Then, there exist functions $p_k \in  \restricedfs{A_k}{\pdfvalues}$ such that {\revcolor for all $\dict\in\dicts$}, if $p_S(\dict) \neq \undefval$, then
    \[p_S(\dict) = p_\cfg(\dict) = p_\unrolledcfg(\dict) = \prod_{k=1}^K p_k(\dict).\]
\end{theorem}
\begin{proof}
    We give a proof sketch and refer to Appendix B.6 for the full proof.
    For each sample node $M_j$ in the \emph{unrolled} CFG $\unrolledcfg$, we construct $\tilde{p}_j \in \restricedfs{\tilde{A}_j}{\pdfvalues}$ as in the proof of \Cref{theorem:simple-factor}.
    Next, we group the nodes $M_j$ by their corresponding CFG node:
    $p_k = \prod_{j\colon \cfgnode(M_j) = N_k} \tilde{p}_j$.
    From \Cref{lemma:staticprov-unrolled-cfg}, we have that if $\cfgnode(M_j) = N_k$, then $\tilde{A}_j \subseteq A_k$.
    Thus,  $p_k \in \restricedfs{A_k}{\pdfvalues}$.
    Like in the proof of \Cref{theorem:simple-factor} it follows from \Cref{theorem-evalf-loopy} that $p_S(\dict) = p_\cfg(\dict) = \prod_{k=1}^K p_k(\dict)$.
\end{proof}

\paragraph{Equivalence to Markov Networks.}
As in \cref{sec:markov-without-loops}, \cref{theorem:loopy-factor} implies that a program $S$ with $K$ sample statements  is equivalent to a Markov network $\mathcal{H}$ with nodes $\{X_\alpha \colon \alpha \in \medcup_{k=1}^K A_k\}$ even if it contains loops.
Again, the sets $D_k = \{X_\alpha \colon \alpha \in A_k\}$ form cliques in $\mathcal{H}$ and $p_S$ factorises over $\mathcal{H}$ with $\phi_{k}(D_k) = p_k(\{\alpha \mapsto X_{\alpha}\colon \alpha \in A_k\})$:
$P(\vv{X}) = \prod_{k=1}^K \phi_{k}(D_k).$

\section{Optimising Posterior Inference Algorithms}\label{sec:casestudy}
While the statically obtained factorisation is provably correct, it is natural to ask whether the quality of the factorisation is good enough to be of practical use.
To answer this question, we have {\revcolor developed} a program transformation based on the static factorisation that generates sub-programs for each factor (described in \Cref{sec:slicing}).
This program transformation operates on probabilistic programs that are implemented in our research PPL which models the formal language introduced in \Cref{sec:semantics}.
It is implemented in Julia and can be viewed as a subset of PPLs like Gen, Turing, or Pyro.
{\revcolor Many posterior inference algorithms can be improved by exploiting the factorisation of the model density.
In \Cref{sec:speed-up}, we demonstrate that if the change to the current program trace is small as in the light-weight Metropolis Hastings kernel~\cite{wingate2011lightweightmh}, then we can leverage the sub-programs to compute the acceptance probability faster.
In \Cref{sec:var-reduce}, we show that we can reduce the variance of the ELBO gradient estimation in variational inference.
Lastly, in \Cref{sec:smc}, we compile sub-programs to enable an iterative implementation of sequential Monte Carlo.
The realised optimisation tricks are well-known but typically restricted to programs with a finite number of random variables, PPLs that require the user to explicitly construct the graphical structure, or dynamic approaches.
In contrast, our novel transformation operates on the source code and makes the optimisations statically available for programs containing while loops and dynamic addresses.
We evaluate our implementation~\citerepl{} on a benchmark set of 16 probabilistic programs, the majority of which make use of while loops to declare an unbounded number of random variables.}

\subsection{Generating Sub-programs from the Static Factorisation}\label{sec:slicing}

We generate a sub-program for each sample statement (and thus for each factor) by deriving a program slicing method from the static provenance analysis.
We want to generate a sub-program for each sample node $N_k = \textnormal{Assign}(x_k = \kw{sample}(E_0^k, f^k(E_1^k,\dots,E_{n_k}^k)))$ that
{\revcolor
\begin{enumerate}
    \item enables us to continue execution directly from $N_k$ from any state $\sigma$,
    \item stops execution once all dependencies of $N_k$ were reached,
    \item supports multiple inference algorithms.
\end{enumerate}
}

{\revcolor To realise (1) and (2)}, we collect all sample nodes that depend on $N_k$.
{\revvcolor From \Cref{theorem:loopy-factor} we know} that a sample node~$N_j$ depends on $N_k$ if $N_k$ contributes to the provenance $A_j$ of the factor $p_j \in \restricedfs{A_j}{\pdfvalues}$ {\revvcolor as defined in \cref{eq:factor-def}.}
As explained in \Cref{sec:staticprov}, these nodes can be found by modifying \Cref{alg:staticprov} on line 8 to collect sample nodes instead of their addresses.
For this modified algorithm $\staticprovnode$, we have $\staticprov(N,x) = \bigcup_{N' \in \staticprovnode(N,x)} \textnormal{addresses}(N')$ and $N_j$ depends on $N_k$ if any variable $y$ in any expression $E_i^j$ in the sample statement of $N_j$ { or in the conditional expression of any branch parent of $N_j$} depends on $N_k$.
We write $j \in \mathcal{J}_k$ if $N_j$ depends on $N_k$.

{\revcolor
We slice the program to only retain nodes that lie on a path from $N_k$ to $N_j$ for some $j\in \mathcal{J}_k$. 
More formally, let $\cfg_k$ be the CFG that contains node $N$ if there exists $j\in \mathcal{J}_k$ such that there is a path from $N_k$ to $N_j$ that visits $N$ and does not revisit $N_k$.
The start node of $\cfg_k$ has successor $N_k$ and the nodes $N_j$, $j\in \mathcal{J}_k$, are connected to the end node of $\cfg_k$ if they have no successor in $\cfg_k$.
Lastly, sample nodes $N_j$ in $G_k$ that do not depend on $N_k$, $j \notin \mathcal{J}_k$, are transformed to \emph{read} nodes for which
$(\progstate, x =\kw{read}(E_0)) \opsemrel \progstate[x \mapsto \dict(\progstate(E_0))]$.
This construction is made precise in Appendix~D.1, where we also prove following correctness theorems {\revvcolor as consequences of \Cref{theorem:loopy-factor}}.
\begin{theorem}\label{theorem:slicing-correctness}
Let $\dict_1$ and $\dict_2$ be traces that only differ at address $\alpha$, i.e. $\dict_1(\alpha) \neq \dict_2(\alpha)$ and $\forall \beta \neq \alpha\colon \dict_1(\beta) = \dict_2(\beta).$
Let the execution sequences in the unrolled CFG $\unrolledcfg$ be equal up to sample node $M_\alpha = \textnormal{Assign}(x = \kw{sample}(E_0, f(E_1,\dots,E_{n})))$ and state $\sigma$, such that $\alpha = \sigma(E_0)$:
\begin{align*}
    (\progstate_0, \text{START}) \cfgreloverset{\dict_1} \cdots  \cfgreloverset{\dict_1} (\progstate, M_\alpha) \cfgreloverset{\dict_1} \cdots \cfgreloverset{\dict_1} (\progstate_{l_1}, \text{END}),\\
    (\progstate_0, \text{START}) \cfgreloverset{\dict_2} \cdots  \cfgreloverset{\dict_2} (\progstate, M_\alpha) \cfgreloverset{\dict_2} \cdots \cfgreloverset{\dict_2} (\progstate_{l_2}, \text{END}),
\end{align*}
If $M_\alpha$ is the only sample node in the execution sequences whose address expression evaluates to $\alpha$, then for the sliced CFG $\cfg_k$ of sample node $N_k = \cfgnode(M_\alpha)$, it holds that
\[
\log p_G(\dict_1) - \log  p_G(\dict_2) = \log \bar{p}_{\cfg_k}(\dict_1, \sigma) - \log  \bar{p}_{\cfg_k}(\dict_2, \sigma).\footnote{{\revcolor
As most PPLs perform computations in terms of log-densities, we will often write $\log a - \log b$ instead of $a / b$.}}
\]
\end{theorem}
\begin{proof}
    {\revvcolor Follows from \Cref{theorem:loopy-factor} and two technical lemmas}. See Appendix~D.2.
\end{proof}
\begin{corollary}\label{theorem:slicing-correctness-corollary}
Assume that for each trace $\dict$ there is at most one sample node $M_\alpha$ in its execution sequences such that the address expression of $M_\alpha$ evaluates to $\alpha$.
For initial program state $\progstate_0$ as in \Cref{def:prog-semantics}, define
    \begin{equation*}
        p_{\alpha}(\dict) \coloneqq \begin{cases}
            \bar{p}_{G_k}(\dict, \progstate)& \text{ if } \exists\,\progstate\in\progstates, M_\alpha\colon (\progstate_0, \text{START}) \cfgreloverset{\dict} \cdots  \cfgreloverset{\dict} (\progstate, M_\alpha) \land \cfgnode(M_\alpha) = N_k\\
            1 & \text{ otherwise.}
        \end{cases}
    \end{equation*}
    Then, $\Delta_\alpha(\dict) \coloneqq \log p_G(\dict) - \log p_{\alpha}(\dict)$ is independent of address $\alpha$, $\Delta_\alpha \in \restricedfs{\strings \setminus \{\alpha\}}{\pdfvalues}.$
\end{corollary}
\begin{proof}
    {\revvcolor Follows directly from \Cref{theorem:slicing-correctness}}. See Appendix~D.3.
\end{proof}

}

{\revcolor
Lastly, to fulfill (3), we further transform the program to take as input an \emph{execution context} \texttt{ctx}.
The context controls the behavior when executing a sample statement, e.g. to resample a variable or to accumulate the density.
To actually handle the effects of different inference algorithms, we replace sample statements in the original program with three abstract methods that delegate the behavior to the context \texttt{ctx}.
In $\cfg_k$, the statement of $N_k$ is transformed to \textcolor{lightblue}{\texttt{visit}}, the statements of the nodes that depend on $N_k$ to \textcolor{lightblue}{\texttt{score}}, and all other sample nodes that are on a path between $N_k$ and some $N_j$, $j\in \mathcal{J}_k$, but do not depend on $N_k$, are transformed to \textcolor{lightblue}{\texttt{read}}.
Concrete instances of such contexts will be presented in the following sections and formalised in Appendix~D.

We have implemented the sub-program generation as a compilation to the Julia programming language making use of multiple dispatch for contextualised execution.
To enable continuation from any program state $\sigma$ we also inject a mutable state struct $\rho$.
As the CFG slicing implementation is a fairly standard graph traversal algorithm, we refer to our implementation~\citerepl{} for further details like the translation from CFGs to source code.}

The sub-program generation process is best illustrated with three examples. 
Consider Program 1 in \Cref{fig:sub-programs}.
To generate a sub-program for the factor of sample statement~"B", we omit the sample statements "A" and "E" as they do not lie on a path from "B" to "D" --- the only sample statement that depends on "B".
However, we still need to read the value of program variable \texttt{A} from the cached program state $\rho$ and we read the value for random variable "C" from the trace.
Continuing execution directly from a sample statement by caching the program state means that we can continue a loop from a specific iteration as exemplified by Program 2 in \Cref{fig:sub-programs}.
This is particularly powerful, if the only dependency of a sample statement belongs to the same loop iteration.
In this case, an entire while loop can be replaced by a single iteration which is shown for Program 3 in \Cref{fig:sub-programs}.
Note that this form of loop-iteration independence is not captured by the factorisation of \Cref{theorem-evalf-loopy}.
Instead, it is product of the program slicing implementation.

\begin{figure}[H]

\begin{minipage}{0.45\linewidth}
\centering
{{\revcolor Original formal program}}
\end{minipage}%
\begin{minipage}{0.55\linewidth}
\centering
{{\revcolor Julia sub-program for factor (input: \texttt{ctx} and $\rho$)}}
\end{minipage}
\vspace{-5mm}

\hrulefill
\vspace{-3mm}

\begin{minipage}[t]{0.45\linewidth}
\begin{lstlisting}[style=Python, numbers=none]
# Program 1
A = sample("A", Norm(0.0,1.0))
B = sample("B", Norm(A,1.0))
C = sample("C", Norm(A,1.0))
D = sample("D", Norm(B+C,1.0))
E = sample("E", Norm(A,1.0))
\end{lstlisting}
\end{minipage}%
\begin{minipage}[t]{0.55\linewidth}
\begin{lstlisting}[style=Python, numbers=none, escapeinside=@@]
# sub-program for factor of variable "B"
@$\rho$@.B = visit(ctx,@$\rho$@,"B",Norm(@$\rho$@.A,1.0))
@$\rho$@.C = read(ctx,@$\rho$@,"C")
@$\rho$@.D = score(ctx,@$\rho$@,"D",Norm(@$\rho$@.B+@$\rho$@.C,1.0))
\end{lstlisting}
\end{minipage}

\hrulefill
\vspace{-3mm}

\begin{minipage}[t]{0.45\linewidth}
\begin{lstlisting}[style=Python, numbers=none]
# Program 2
i = 0; b = 1;
while b do
  b = sample("b"+str(i),Bern(0.5))
  i = i + 1
\end{lstlisting}
\end{minipage}%
\begin{minipage}[t]{0.55\linewidth}
\begin{lstlisting}[style=Python, numbers=none, escapeinside=@@]
# sub-program for factor of variable "b"
@$\rho$@.b = visit(ctx,@$\rho$@,"b"*str(@$\rho$@.i),Bern(0.5))
@$\rho$@.i = @$\rho$@.i + 1
while b
  @$\rho$@.b = score(ctx,@$\rho$@,"b"*str(@$\rho$@.i),Bern(0.5))
  @$\rho$@.i = @$\rho$@.i + 1
end
\end{lstlisting}
\end{minipage}
\vspace{-3mm}

\hrulefill
\vspace{-3mm}

\begin{minipage}[t]{0.45\linewidth}
\begin{lstlisting}[style=Python, numbers=none]
# Program 3
i = 0;
while i < N do
  z = sample("z"+str(i),Bern(0.5))
  m = (z == 1) ? -2.0 : 2.0
  x = sample("x"+str(i),Norm(m,1.0))
  i = i + 1
\end{lstlisting}
\end{minipage}%
\begin{minipage}[t]{0.55\linewidth}
\begin{lstlisting}[style=Python, numbers=none, escapeinside=@@]
# sub-program for factor of variable "z"
@$\rho$@.z = visit(ctx,@$\rho$@,"z"*str(@$\rho$@.i),Bern(0.5))
@$\rho$@.m = (@$\rho$@.z == 1) ? -2.0 : 2.0
@$\rho$@.x = score(ctx,@$\rho$@,"x"*str(@$\rho$@.i),Norm(@$\rho$@.m,1.0))
\end{lstlisting}
\end{minipage}
\hrulefill

\caption{{\revcolor Example Julia sub-programs generated from the static factorisation (with short-hands \texttt{Norm} = \texttt{Normal}, \texttt{Bern} = \texttt{Bernoulli}, \texttt{str} = \texttt{string}). Program variables are read from and stored to the program state struct~$\rho$.}}
\label{fig:sub-programs}

\end{figure}

\subsection{Speeding-up Acceptance Probability Computation for Single-Site Trace Updates}\label{sec:speed-up}


{\revcolor In this section we introduce execution contexts to speed-up the single-site or light-weight Metropolis Hastings inference (LMH) algorithm~\cite{wingate2011lightweightmh}.}
The LMH algorithm works by selecting a single address~$\alpha$ in each iteration at random and proposing a new value $v$ {\revcolor for the current trace $\dict(\alpha)$}.
Even though we select a single address for the update, the proposal may have to produce additional values for sample statements that were not executed for trace $\dict$, but due to stochastic branching or dynamic addressing are now executed for the new trace~$\dict'$.
For instance in \Cref{fig:sub-programs} Program 2, if we change the value of the last random variable $\textnormal{"b}\_i\textnormal{"}$ from 0 to 1, then we may have to sample values for multiple  $\textnormal{"b}\_j\textnormal{"}$ where $j > i$ until we get $\dict'(\textnormal{"b}\_j\textnormal{"}) = 0$.

The proposed trace $\dict'$ is accepted {\revcolor(made current)} with probability
$A = \min\left(1, \frac{p(\dict') \cdot Q_\alpha(\dict|\dict')}{p(\dict) \cdot Q_\alpha(\dict'|\dict)}\right),$ 
where $Q_\alpha$ captures the probability of the proposal process (sampling a value for $\alpha$ and all additional values as explained above).
{\revcolor Then, the set of all generated traces is an approximation to the posterior distribution of the model}.

{\revcolor In our approach, if the address comes from sample node $N_k$, then we can optimise the runtime of this algorithm by noting that $\log p(\dict') - \log p(\dict) =  \sum_{j =1}^K \log p_j(\dict') - \log p_j(\dict) = \sum_{j \in \mathcal{J}_k \cup \{k\}}\log p_j(\dict') - \log p_j(\dict)$.
Thus, by being able to compute the factors $p_j(\dict')$ and $p_j(\dict)$ separately, we can compute the ratio of densities $p(\dict') / p(\dict)$ by summing fewer terms.}

{\revcolor By designing execution contexts for the sub-programs to handle the effect of the LMH kernel, we statically enable this optimisation for programs with an unbounded number of random variables and dynamic dependencies.
In the following we informally describe the implementation of two execution contexts for resampling address~$\alpha$ at sample node $N_k$.
A formal treatment can be found in Appendix D.1.
}

When executed with \texttt{ForwardCtx}, the sub-program is evaluated with respect to the proposed trace $\dict'$.
At \lstinline[columns=fixed,style=Python,mathescape]{visit(ForwardCtx,$\rho$,$E_0$,$f(\dots$))} we sample a new value $v$ for address $\alpha=\rho(E_0)$ from some proposal distribution $q$ and return $v$. {\revcolor As before, $\rho(E_0)$ denotes $E_0$ evaluated with respect to~$\rho$.}
In the background, we accumulate the log-probability $\Delta P \gets \Delta P + \log \pdf_f(v;\dots)$ and proposal log-probability $\Delta Q \gets \Delta Q - \log \pdf_q(v;\dots)$.
At the start, $\Delta P$ and $\Delta Q$ are set to $0$.
At \lstinline[columns=fixed,style=Python,mathescape]{score(ForwardCtx,$\rho$,$E_0$,$f(\dots$))} we propose a new value $v$ \emph{only} if the address $\rho(E_0)$ is not sampled for the current trace.
We always update $\Delta P$ as before and only update $\Delta Q$ if we have proposed a new value.
At \lstinline[columns=fixed,style=Python,mathescape]{read(ForwardCtx,$\rho$,$E_0$)}, we simply return $\dict(\rho(E_0))$.
At each of the above statements we cache the program state in $\rho$.
This introduces book-keeping cost, but as we will see in \Cref{sec:eval-results}, the cost is outweighed by the performance gained from only running sub-programs.

When executed with \texttt{BackwardCtx}, the sub-program is evaluated with respected to the current trace $\dict$.
At \lstinline[columns=fixed,style=Python,mathescape]{visit(BackwardCtx,$\rho$,$E_0$,$f(\dots$))} we update $\Delta P \gets \Delta P - \log \pdf_f(v;\dots)$ and proposal log-probability $\Delta Q \gets \Delta Q + \log \pdf_q(v;\dots)$ with the current value $v = \dict(\rho(E_0))$.
At \lstinline[columns=fixed,style=Python,mathescape]{score(BackwardCtx,$\rho$,$E_0$,$f(\dots$))} the same $\Delta P$-updates are always performed while the $\Delta Q$-updates are only performed if the address $\rho(E_0)$ is not sampled for the \emph{proposed} trace $\dict'$.
At statement \lstinline[columns=fixed,style=Python,mathescape]{read(BackwardCtx,$\rho$,$E_0$)} we again return $\dict(\rho(E_0))$.
The program state does not have to be cached when executing with \texttt{BackwardCtx}.

Lastly, after executing the sub-program in both contexts, we have $A = \min(1, \exp(\Delta P + \Delta Q))$.
We accept the proposed trace with probability $A$.
In this case, we perform some book-keeping to update the inference state to the proposed trace with the cached program states. {\revcolor Building on the soundness of our provenance analysis and \Cref{theorem:slicing-correctness}, we proof correctness of this approach in Appendix D.4:
\begin{corollary}\label{theorem:lmh-correctness}
    Assume that the address expression of at most one sample node in the execution sequence evaluates to $\alpha$ for all traces.
    Let $\dict'$ be the LMH proposal at address $\alpha$ for $\dict$.
    Then, the execution sequences of $\dict$ and $\dict'$ are equal up to $(\sigma, M_\alpha)$, where $M_\alpha$ is the first and only sample node whose address expression evaluates to~$\alpha$.
    For the sliced CFG $\cfg_k$, where $\cfgnode(M_\alpha) = N_k$, we have \[\log p(\dict') - \log p(\dict) = \log \bar{p}_{G_k}(\dict', \sigma) - \log \bar{p}_{G_k}(\dict, \sigma).\]
    The quantities $\log \bar{p}_{G_k}(\dict', \sigma)$, and $\log Q_\alpha(\dict'|\dict)$ are computed by \texttt{ForwardCtx} and the quantities  $\log \bar{p}_{G_k}(\dict, \sigma)$ and $\log Q_\alpha(\dict|\dict')$ are computed by \texttt{BackwardCtx}. The new trace $\dict'$ can be constructed from $\dict$ and the output of both contexts.
\end{corollary}
}

\subsubsection{Experiment Results}\label{sec:eval-results}
\begin{table}[t]
    \caption{Average runtime of computing a single LMH iteration in microseconds (see \Cref{sec:speed-up}). The column "Baseline" reports the runtime when computing $p(\dict')$ from scratch. The column "Static Analysis" reports the runtime when updating the old density by computing $\log p(\dict') - \log p(\dict)$ with the generated sub-programs (see \Cref{sec:slicing}). The column "Fixed-Finite" reports the runtime of updating the density with a hand-written factorisation that is representative of PPLs that explicitly construct a graphical model. The speed-up in relation to computing from scratch is given in the "speed-up" and "sp.up" columns. The number of random variables defined in each model can be found in column "\#RVs". The number of samples produced and the acceptance rate are reported in the $N$ and $AR$ column, respectively.}
    \label{tab:results}
    \centering
    \begin{tabular}{l r r r r r r c r r}
        \hline
         \textbf{} & \textbf{} & \textbf{} & \textbf{} & \textbf{Baseline} &  \multicolumn{3}{c}{\textbf{Static Analysis}} &  \multicolumn{2}{c}{\textbf{Fixed-Finite}} \\
         \textbf{Model} & \textbf{\#RVs} & $N$ &$AR$ & $\microseconds$ & $\microseconds$ &  \multicolumn{2}{c}{speed-up} & $\microseconds$ & sp.up \\
         \hline
Aircraft & $\infty$ & $1\mathrm{e}{5}$ & 0.47 & 5.861 & 4.605 & \textcolor{ForestGreen}{1.29} & \textcolor{ForestGreen}{$\fastarrow$} & - & - \\
Bayesian Network & 37 & $1\mathrm{e}{5}$ & 0.93 & 25.580 & 4.680 & \textcolor{ForestGreen}{5.47} & \textcolor{ForestGreen}{$\megafastarrow$} & 3.344 & 7.65 \\
Captcha & $\infty$ & $1\mathrm{e}{3}$ & 0.81 & 774.287 & 355.382 & \textcolor{ForestGreen}{2.18} & \textcolor{ForestGreen}{$\megafastarrow$} & - & - \\
Dirichlet process & $\infty$ & $1\mathrm{e}{4}$ & 0.53 & 89.939 & 42.664 & \textcolor{ForestGreen}{2.12} & \textcolor{ForestGreen}{$\megafastarrow$} & - & - \\
Geometric &  $\infty$ & $5\mathrm{e}{5}$ & 0.62 & 2.419 & 2.475 & \textcolor{Gray}{0.97} & \textcolor{Gray}{$\samearrow$} & - & - \\
GMM (fixed \#clusters) & 209 & $5\mathrm{e}{4}$ & 0.60 & 71.136 & 7.027 & \textcolor{ForestGreen}{10.01} & \textcolor{ForestGreen}{$\megafastarrow$} & 4.564 & 15.58 \\
GMM (variable \#clusters)\hspace{-3mm} & $\infty$ & $5\mathrm{e}{4}$ & 0.57 & 74.709 & 9.549 & \textcolor{ForestGreen}{7.81} & \textcolor{ForestGreen}{$\megafastarrow$} & - & - \\
HMM & 101 & $1\mathrm{e}{5}$ & 0.74 & 33.729 & 24.016 & \textcolor{ForestGreen}{1.39} & \textcolor{ForestGreen}{$\fastarrow$} & 2.551 & 13.22 \\
Hurricane & 5 & $1\mathrm{e}{6}$ & 0.86 & 2.121 & 1.918 & \textcolor{Gray}{1.13} & \textcolor{Gray}{$\samearrow$} & - & - \\
LDA (fixed \#topics) & 551 & $1\mathrm{e}{4}$ & 0.81 & 209.046 & 24.526 & \textcolor{ForestGreen}{8.52} & \textcolor{ForestGreen}{$\megafastarrow$} & 18.341 & 11.40 \\
LDA (variable \#topics) & $\infty$ & $1\mathrm{e}{4}$ & 0.83 & 218.614 & 25.732 & \textcolor{ForestGreen}{8.50} & \textcolor{ForestGreen}{$\megafastarrow$} & - & - \\
Linear regression & 102 & $1\mathrm{e}{5}$ & 0.53 & 4.477 & 8.006 & \textcolor{OrangeRed}{0.56} & \textcolor{OrangeRed}{$\megaslowarrow$} & 7.504 & 0.60 \\
Marsaglia & $\infty$ & $5\mathrm{e}{5}$ & 0.93 & 1.699 & 2.159 & \textcolor{OrangeRed}{0.79} & \textcolor{OrangeRed}{$\slowarrow$} & - & - \\
PCFG & $\infty$ & $1\mathrm{e}{5}$ & 0.75 & 10.220 & 8.437 & \textcolor{Gray}{1.21} & \textcolor{Gray}{$\samearrow$} & - & - \\
Pedestrian & $\infty$ & $1\mathrm{e}{5}$ & 0.37 & 3.709 & 3.645 & \textcolor{Gray}{1.02} & \textcolor{Gray}{$\samearrow$} & - & - \\
Urn & $\infty$ & $1\mathrm{e}{5}$ & 0.52 & 11.979 & 6.354 & \textcolor{ForestGreen}{1.88} & \textcolor{ForestGreen}{$\fastarrow$} & - & - \\
         \hline
    \end{tabular}
\end{table}

To evaluate the speed-up made possible by our static factorisation, we gathered a benchmark set of 16 probabilistic programs.
We specifically collected probabilistic programs {\revcolor from prior research} that make use of while loops to define an unbounded number of random variables.
{\revcolor References for the programs can be found in our replication package~\citerepl{}}.
Only 5 out of the 16 models define just a finite number of random variables with fixed dependency structure.
It takes less than five seconds to generate the sub-programs for all models.

For each model, we run two implementations of the LMH inference algorithm.
As a baseline, the standard implementation of LMH computes the acceptance probability $A$ by computing the density $p(\dict')$ from scratch.
Our optimised version of LMH leverages the sub-programs to compute $\log p(\dict') - \log p(\dict)$ when calculating $A$.
For the five finite models with fixed structure, we additionally implemented a factorisation that is representative of PPLs that explicitly construct a probabilistic graphical model.
This factorisation enables LMH to exploit the full dependency structure of the model and serves as comparison.
We have tested correctness of our LMH implementations by asserting that \emph{exactly} the same acceptance probabilities and samples are computed.

In \Cref{tab:results}, we report the results of the experiment run a M2 Pro CPU.
For each model and implementation, we ran LMH with 10 random starting traces for a total of $N$ samples {\revcolor and average the measurements over 10 experiment repetitions}.

The first observation that we want to highlight is that computing $\log p(\dict') - \log p(\dict)$ with our approach can lead to longer runtime if there is no structure to exploit in the model.
Consider the linear regression model in which the 100 observations depend on the slope and intercept variables.
{\revcolor When changing the value of one of these variables, computing $\log p(\dict') - \log p (\dict) $ requires 100 additions and 100 subtractions, while $\log p(\dict')$ only takes 100 additions.}
The same is true for the Marsaglia model which implements a rejection sampling method for the Normal distribution such that all variables depend on each other.

The factorisation approach shines when there is a lot of structure in the models like for the Aircraft, Captcha, Dirichlet Process, and Urn models, where the runtime is decreased up to a factor of 2.
It is particularly powerful, when there are latent variables that only affect independent data variables.
This is the case for the Gaussian mixture model (GMM) where there is one cluster allocation variable per data point and for the Latent Dirichlet Allocation model (LDA) where there is similar dependency structure.
For these models we observe 8-10x speed-ups.
Note that the handwritten "fixed-finite" factorisations of the models are faster than the "static analysis" factorisations, because LMH has to do less book-keeping for them.
They do not leverage additional dependency structure.


The only model where our approach fails to effectively leverage the true dependency structure is the Hidden Markov Model (HMM).
Even though the current state only affects the next state, our static analysis over-approximates this structure and outputs that the current state affects all future states.
Now, if we modify the program by unrolling the while loop, i.e. repeating its body fifty times, our factorisation approach results in an improved runtime of 2.605~$\microseconds$ (13x speed-up), because in this case the dependency structure is not over-approximated anymore due to loops.
Note that the PCFG model essentially has the same structure as HMM, but can transition to a terminal state at every loop iteration.
In this case, indeed the current state affects all future states and our static analysis is exact, but no significant improvements in runtime are measurable.

Lastly, we note that the Hurricane model is the only model in our benchmark with a finite number of random variables but a dynamic dependency structure (see \Cref{sec:markov-without-loops}).
This model can also not be represented in a conventional Bayesian network structure, but our method is applicable and can slightly reduce the runtime.

\subsubsection{Comparison to Related Methods}\label{sec:comparsion-to-related}
The approach presented in \Cref{sec:casestudy} is the first to generate optimised MCMC kernels of models with unbounded number of random variables via a static program transformation.
Importantly, its correctness follows from the theoretical foundation established in \Cref{sec:loops}.
Nevertheless, there exist methods that aim to accelerate MCMC kernels for a similar class of models with other approaches.
In this section, our goal is to make a fair comparison to those methods by focusing on the relative speed-ups achieved.
A direct runtime comparison is not meaningful due to differences in programming languages and the implementation of inference algorithms.
More related approaches are described in \Cref{sec:related-work}.


\begin{table}[h]
    \centering
    \caption{{\revcolor Relative speed-up for the LMH algorithm achieved with our static approach, combinators in Gen, and the C3 algorithm in WebPPL.}}
    \label{tab:gen-wppl}
    \begin{tabular}{|l r r r | l r r r|}
\hline

\textbf{Model} & ours & Gen & C3 & \textbf{Model} & ours & Gen & C3 \\
\hline
Aircraft & 1.29 & \textbf{1.41} & 1.02 & Hurricane & \textbf{1.13} & 0.51 & 0.93 \\
Bayesian Network & \textbf{5.47} & 4.14 & 0.93 & LDA (fixed \#topics) & \textbf{8.52} & 6.50 & 4.97 \\
Captcha & 2.18 & \textbf{2.93} & 1.05 & LDA (variable \#topics) & \textbf{8.50} & 5.92 & 7.26 \\
Dirichlet process & \textbf{2.12} & - & 1.24 & Linear regression & \textbf{0.56} & 0.43 & 0.44 \\
Geometric & \textbf{0.97} & - & 0.74 & Marsaglia & \textbf{0.79} & - & 0.53 \\
GMM (fixed \#clusters) & \textbf{10.01} & 6.26 & 1.90 & PCFG & \textbf{1.21} & - & 0.79 \\
GMM (variable \#clusters) & \textbf{7.81} & 5.45 & 1.60 & Pedestrian & \textbf{1.02} & - & 0.52 \\
HMM & 1.39 & 1.31 & \textbf{3.82} & Urn & 1.88 & \textbf{2.05} & 0.69 \\
\hline
    \end{tabular}
\end{table}

First, \citet{ritchie2016c3} presented the C3 approach to enable incrementalised re-execution of Metropolis Hastings kernels using continuations and function call caching.
To be effective this method requires programs to be implemented with recursion instead of while loops.
We re-implemented the models in recursive form in WebPPL~\cite{webppl} to report the relative speed-up achieved with C3 in \Cref{tab:gen-wppl}.
Our approach achieves better speed-ups for all models except for the HMM.
As explained before, our static analysis greatly over-approximates the dependency structure of this model.
Note that \citet{ritchie2016c3} report even more significant speed-ups for the GMM, LDA, and HMM models for larger data-set sizes.
Our approach scales the same for GMM and LDA.

Further, modern probabilistic programming systems like Gen~\cite{cusumano2019gen} include language constructs like \texttt{Map}, \texttt{Unfold}, and \texttt{Recurse} to model independence structure.
Thus, we re-implemented LMH and the 16 models only in Gen with the so-called combinator constructs and report the results in \Cref{tab:gen-wppl}.
Gen's inference backend can only leverage the combinator constructs if the root model function is written in the \emph{static modelling language} -- a non-universal subset of the language with no control statements.
This effectively abstracts away loops as \texttt{Map}, \texttt{Unfold}, or \texttt{Recurse} nodes in an internal graphical representation.
Note that at the time of writing Gen does not implement incremental computation for the \texttt{Recurse} combinator, a construct required for models where the condition of the while loop depends on the body like the Geometric model.
We cannot apply our general LMH implementation to those models and mark them with "-".
Overall, the order of magnitude of speed-up achieved with combinators is similar to our static approach.
There seems to be more overhead in Gen as it supports general programmable inference in contrast to our research PPL.
We emphasise that our approach achieves these speed-ups fully automatically from source code without the need to explicitly model independence structure with language constructs.


\subsection{Reducing the Variance of an ELBO Gradient Estimator in Variational Inference}\label{sec:var-reduce}
{\revcolor
Variational inference takes a different approach to MCMC methods by fitting a so-called variational distribution $Q_\phi$ to the posterior, where $\phi$ is a set of parameters. 
The objective is to maximise the evidence lower-bound $\textnormal{ELBO}(\phi) = \expectation{z \sim Q_\phi}{\log  P(z) - \log Q_\phi(z)}$ via gradient ascent, where $ P$ denotes a potentially unnormalised density.
In this section, we consider Black-Box Variational Inference (BBVI) \cite{ranganath2014blackboxvi}, which estimates the gradient of the ELBO with following identity $\nabla_\phi \textnormal{ELBO}(\phi) = \expectation{z \sim Q_\phi}{\nabla_\phi \log Q_\phi(z) \cdot (\log  P(z) - \log Q_\phi(z))}$.
Unfortunately, this estimator often exhibits high variance.
To reduce the variance \citet{ranganath2014blackboxvi} use a mean-field variational distribution  $Q_\phi(z) = \prod_{i=1}^n q_i(z_i|\phi_i)$ to derive the improved estimator
\begin{equation}\label{eq:bbvi-rao}
\nabla_{\phi_i} \textnormal{ELBO}(\phi) =\expectation{z_{(i)} \sim Q^i _\phi}{\nabla_{\phi_i} \log q_i(z_i|\phi_i) \cdot (\log  P^i(z_{(i)}) - \log q_i(z_i|\phi_i))}.
\end{equation}
Note that the gradient is only with respect to $\phi_i$ and the expectation is with respect to the \emph{Markov blanket} $z_{(i)}$ of $z_i$.
Including $z_i$, this is the set of variables that directly influence $z_i$ according to $ P$, variables that are directly influenced by $z_i$ -- \emph{children} of $z_i$ -- and variables that directly influence a child of $z_i$.
Then, $ P^i$ is the density factoring out all other variables and $Q^i_\phi$ the variational distribution restricted to $z_{(i)}$.
This reduces the variance and we will replicate this trick with our factorisation.

In our sub-program setting, we achieve the same optimisation by defining the execution context \texttt{BBVICtx}, which we formalise in Appendix D.5.
First, assume that a trace $\dict$ is generated from a variational distribution, where each address $\alpha$ is sampled from $\dict(\alpha) \sim q_\alpha(\dict(\alpha) | \phi_\alpha)$.
The \texttt{BBVICtx} executes the program with respect to $\dict$.
Similar to the \texttt{ForwardCtx}, it accumulates $ \Delta \gets \Delta + \log \pdf_f(v;\dots)$ at 
\lstinline[columns=fixed,style=Python,mathescape]{score(BBVICtx,$\rho$,$E_0$,$f(\dots$))}
and returns $v = \dict(\rho(E_0))$ at \lstinline[columns=fixed,style=Python,mathescape]{read(BBVICtx,$\rho$,$E_0$))}. 
At visit statement
\lstinline[columns=fixed,style=Python,mathescape]{visit(BBVICtx,$\rho$,$E_0$,$f(\dots$))} we additional substract the proposal density $ \Delta \gets \Delta + \log \pdf_f(v;\dots) - \log  q_\alpha(v | \phi_\alpha)$, where we have $\alpha = \rho(E_0)$ (compare \cref{eq:bbvi-rao}).
In Appendix~D.5. we prove following correctness result {\revvcolor as a consequence of \Cref{theorem:slicing-correctness-corollary}}:
\begin{corollary}\label{theorem:bbvi-correctness}\errnote{\hyperref[sec:errata]{$\triangleright$ see\\ Errata}}
    Let $Q_\phi$ be a variational distribution over traces as formalised in Appendix D.5.
    Assume that the address expression of at most one sample node in the execution sequence evaluates to~$\alpha$ for all traces.
    Under measurability assumptions given in Appendix D.5 and with $p_\alpha$ as in \Cref{theorem:slicing-correctness-corollary},
    \begin{align*}
    &\nabla_{\phi_\alpha}\expectation{\dict \sim  Q_\phi}{\log p(\dict) - \log  Q_\phi(\dict)} = \\
    &\quad\quad\quad\expectation{\dict \sim  Q_\phi}{\delta_{\dict(\alpha)\neq\none}\cdot\nabla_{\phi_\alpha} \log q_\alpha(\dict(\alpha)|\phi_\alpha) \cdot \big(\log p_\alpha(\dict) -\log  q_\alpha(\dict(\alpha)|\phi_\alpha)\big)}.
    \end{align*}
    The quantity $\log p_\alpha(\dict) -\log  q_\alpha(\dict(\alpha)|\phi_\alpha)$ is computed by \texttt{BBVICtx}.
\end{corollary}
{\revvcolor Note that in contrast to \Cref{eq:bbvi-rao}, where the expectation is taken explicitly with respect to the Markov blanket $z_{(i)}$, in \Cref{theorem:bbvi-correctness} the dependencies are accounted for via $p_\alpha(\dict)$ (computed with a sub-program).
Replacing $p$ in the integrand with $p_\alpha$  reduces the variance, because $p_\alpha$ is constant with respect to addresses that are independent from a change to the trace at address~$\alpha$.}

\subsubsection{Experiment Results}

To evaluate the approach we have implemented our BBVI optimisation and compare the improved estimator against the standard estimator with respect to the average gradient variance across all parameters.
In 10 experiment repetitions, we performed 1000 optimisation steps and used 100 samples to estimate the gradient in each step.
In \Cref{tab:bbvi-results}, we report our results, where we omit programs with random variables that have stochastic support dimension and are incompatible with BBVI.

We compare against Pyro's \texttt{TraceGraph\_ELBO} variance reduction technique, which follows \citet{schulman2015stochasticcomputationgraphs} and makes use of dynamic data provenance analysis.
This comes with the benefit of having no over-approximation issues, but also comes with the drawback of not being able to support programs with control flow dependencies.
For Pyro we decreased the number of samples to 100 and number of repetitions to 1 due to long runtime.
In \Cref{tab:bbvi-results}, we see that for programs with only data dependencies we achieve comparable variance reduction with lower runtime penalty.
Furthermore, with our approach we were able to reduce the variance for two models with control flow dependencies Aircraft and Marsaglia.
Note that for the same reasons as in LMH, the static optimisation is only effective for the unrolled version of the HMM.
\begin{table}[h]
    \caption{Average variance of the standard and improved ELBO gradient estimator in BBVI. The variance is averaged over all parameters, 1000 iterations, and 100 samples per iteration.
    The "cost" column reports the factor by which the runtime increased for the improved estimator.
    For reference, we include the variance reduction factor achieved by Pyro's \texttt{TraceGraph\_ELBO}, where 
     unsupported programs are denoted with "-".}
    \label{tab:bbvi-results}
    \errnote{\hyperref[sec:errata]{$\triangleright$ see\\ Errata}}
    \centering
\begin{tabular}{l r r r r c r r r}
        \hline
         \textbf{}& \textbf{Baseline} &  \multicolumn{4}{c}{\textbf{Static Analysis}} &  \multicolumn{2}{c}{\textbf{Pyro}}\\
         \textbf{Model} & avg. var. & avg. var. &  \multicolumn{2}{c}{red. factor} & cost & red. factor & cost \\
         \hline
Aircraft & 7.922e+03 & 7.969e+01 & \textcolor{ForestGreen}{99.29} & \textcolor{ForestGreen}{$\megaslowarrow$} & 1.89 & - & - \\
Bayesian Network & 4.160e+02 & 4.371e+00 & \textcolor{ForestGreen}{95.06} & \textcolor{ForestGreen}{$\megaslowarrow$} & 1.94 & 95.09 &  7.52 \\
Geometric & 1.634e-01 & 1.726e-01 & \textcolor{Gray}{0.94} & \textcolor{Gray}{$\samearrow$} & 0.53 & - & - \\
GMM (fixed \#clusters) & 8.043e+09 & 1.014e+06 & \textcolor{ForestGreen}{8092.80} & \textcolor{ForestGreen}{$\megaslowarrow$} & 2.26 & 9848.79 &  8.69 \\
HMM & 2.116e+03 & 1.279e+03 & \textcolor{Gray}{1.65} & \textcolor{Gray}{$\samearrow$} & 7.49 & 1254.24 &  8.67 \\
HMM (no loop) & 2.116e+03 & 2.349e+00 & \textcolor{ForestGreen}{901.01} & \textcolor{ForestGreen}{$\megaslowarrow$} & 1.50 & 1254.24 &  8.64 \\
Hurricane & 3.362e-01 & 3.609e-01 & \textcolor{Gray}{0.93} & \textcolor{Gray}{$\samearrow$} & 0.81 & - & - \\
LDA (fixed \#topics) & 7.229e+04 & 1.172e+00 & \textcolor{ForestGreen}{61646.02} & \textcolor{ForestGreen}{$\megaslowarrow$} & 5.92 & 59986.61 &  11.00 \\
Linear regression & 3.350e+05 & 3.341e+05 & \textcolor{Gray}{1.00} & \textcolor{Gray}{$\samearrow$} & 1.14 & 1.00 &  14.39 \\
Marsaglia & 7.274e-02 & 1.283e-02 & \textcolor{ForestGreen}{5.38} & \textcolor{ForestGreen}{$\slowarrow$} & 0.51 & - & - \\
PCFG & 8.838e-03 & 1.989e-02 & \textcolor{Gray}{0.56} & \textcolor{Gray}{$\samearrow$} & 0.94 & - & - \\
Pedestrian & 3.108e+06 & 3.111e+06 & \textcolor{Gray}{1.00} & \textcolor{Gray}{$\samearrow$} & 2.01 & - & - \\
\hline
\end{tabular}
\end{table}
}
\subsection{Iteratively Running Sequential Monte Carlo}\label{sec:smc}
{\revcolor
Lastly, we briefly describe how our slicing technique can be adapted to the setting of sequential Monte Carlo (SMC) ~\cite{chopin2020introductiontosmc}, which yet again takes a different approach to inference.
In SMC a number of traces {\revvcolor $\dict_{n}$}, so-called \emph{particles}, are evaluated in parallel.
Data points are added iteratively to the inference process and the particles are resampled according to their data likelihood.
More precisely at time $t$ we have data set $\mathcal{D}_{1:t}$ and define the corresponding density $p_t$.
At inference step $t$, the particle $\dict$ is logarithmically weighted according to $\log p_t({\revvcolor \dict_{n}}) - \log p_{t-1}({\revvcolor \dict_{n}})$ and the particles are resampled {\revvcolor (duplicated or discarded)} with respect to these weights.

A "naive" implementation of SMC would compute $\log p_t({\revvcolor \dict_{n}})$ from scratch, while an iterative implementation is able to compute $\log p_t({\revvcolor \dict_{n}}) - \log p_{t-1}({\revvcolor \dict_{n}})$ directly.
In probabilistic programming this is often realised by some mechanism that enables the interruption and continuation of program execution. 
We enable this optimisation by slightly adapting the slicing method of \Cref{sec:slicing} to generate sub-programs which compute the difference directly.

Namely, for sample node $N_k$ let $\cfg_k^\texttt{SMC}$ be the CFG that contains node $N$ if there exists a sample node $N_j$ such that $N$ lies on a path from $N_k$ to $N_j$ which contains no other sample nodes.
Again, $N_k$ will be the successor of the start node of $G_k^\texttt{SMC}$, but \emph{all} other samples nodes will have the end node of $G_k^\texttt{SMC}$ as sole successor.
As we only care about continuing execution from $N_k$ to the next sample node, this construction is independent of the provenance analysis {\revvcolor and \Cref{theorem:loopy-factor}}.

In our implementation, we again make use of an evaluation context which simply accumulates the density like the standard semantics. 
In Appendix~D.6, we prove correctness:
\begin{proposition}\label{theorem:smc-correctness}
    Let $\cfg$ be a CFG such that for initial program state $\progstate_0$ and trace~$\dict$, we have execution sequence
    $(\progstate_0, \text{START}) \cfgrel \cdots  \cfgrel (\progstate_i, M_i) \cfgrel \cdots \cfgrel (\progstate_l, \text{END})$ in its unrolled version $\unrolledcfg$.
    Let $(\progstate'_1, M'_1) \cfgrel \cdots  \cfgrel (\progstate'_2, M'_2)$ be a sub-sequence such that $\cfgnode(M'_1) = N_k$ and $\cfgnode(M'_2) = N_j$ are the only sample nodes.
    Then, {\revvcolor the execution sequence in the unrolled version of the sliced CFG $\cfg_k^\texttt{SMC}$  matches the sub-sequence}:
    $(\progstate'_1, \text{START}) \cfgrel (\progstate'_1, M'_1) \cfgrel \cdots  \cfgrel (\progstate'_2, M'_2)\cfgrel (\progstate'_2, \text{END}).$
\end{proposition}
{\revvcolor \Cref{theorem:smc-correctness} allows us to evolve the execution sequence of a particle $\dict_{n}$ from one sample node $N_t^{n} = N_k$ to the next $N_{t+1}^{n}=N_j$ by executing the corresponding sliced CFG $\cfg_k^\texttt{SMC}$.
The weight of the particle is updated with $\log p_t(\dict_{n}) - \log p_{t-1}(\dict_{n}) = \log \sigma_2'(\probvar) - \log \sigma_1'(\probvar)$ and particles may be resampled even if the execution sequences were stopped at different sample nodes $N_t^{n}$~\cite{lunden2021correctnesssmc}.
}
 
\begin{table}[h]
    \caption{{\revcolor Average runtime of naive versus iterative implementations of SMC with 100 particles. For reference, we report the speed-up achieved by WebPPL's CPS transformation.}}
    \label{tab:smc-results}
    \centering
\begin{tabular}{l r r r r c r}
        \hline
         \textbf{} & \textbf{} & \textbf{Baseline} &  \multicolumn{3}{c}{\textbf{Static Analysis}} & \textbf{WebPPL} \\
         \textbf{Model} & \textbf{\#Data} & $\milliseconds$ & $\milliseconds$ &  \multicolumn{2}{c}{speed-up} & speed-up \\
         \hline
Dirichlet process & 50 & 238.323 & 22.374 & \textcolor{ForestGreen}{10.57} & \textcolor{ForestGreen}{$\megafastarrow$} & 8.96 \\
GMM (fixed \#clusters) & 100 & 166.480 & 31.260 & \textcolor{ForestGreen}{5.26} & \textcolor{ForestGreen}{$\megafastarrow$} & 9.23 \\
GMM (variable \#clusters) & 100 & 189.529 & 31.263 & \textcolor{ForestGreen}{5.95} & \textcolor{ForestGreen}{$\megafastarrow$} & 9.18 \\
HMM & 50 & 31.951 & 7.394 & \textcolor{ForestGreen}{4.13} & \textcolor{ForestGreen}{$\megafastarrow$} & 7.66 \\
HMM (no loop) & 50 & 49.649 & 7.420 & \textcolor{ForestGreen}{6.63} & \textcolor{ForestGreen}{$\megafastarrow$} & 6.79 \\
LDA (fixed \#topics) & 262 & 1319.553 & 397.074 & \textcolor{ForestGreen}{3.33} & \textcolor{ForestGreen}{$\megafastarrow$} & 7.25 \\
LDA (variable \#topics) & 262 & 1412.182 & 410.781 & \textcolor{ForestGreen}{3.44} & \textcolor{ForestGreen}{$\megafastarrow$} & 8.18 \\
\hline
\end{tabular}
\end{table}
}

\subsubsection{Experiment Results}

To measure the performance gains achievable with our static approach, we compare the runtime of the iterative implementation of SMC based on the generated sub-programs against a naive implementation which computes $\log p_t(\dict)$ from scratch.
In \Cref{tab:smc-results} we report the speed-ups on the subset of models which can make use of iteratively adding data points in inference.
For all models we set the number of particles to 100 and repeated the experiment 10 times.
The measured speed-ups are significant and range from 3x - 10x.
For reference, we benchmarked the runtimes achieved with WebPPL's continuation passing style (CPS) approach against a naive SMC implementation in JavaScript.
For most models this resulted in greater speed-ups compared to our static approach.
This is due to the fact that, like in \Cref{sec:speed-up}, in our approach computational overhead is introduced by caching and copying programs states when resampling particles.

\section{Related Work}\label{sec:related-work}
\subsection{Semantics of Probabilistic Programming Languages}
One of the first works to formally define the semantics of a probabilistic programming language was~\citet{kozen1979semantics}.
Derivatives of this semantics can be found throughout literature, for example see \citet{hur2015provablycorrectsampler}, \citet{dahlqvist2019semanticshigherorder}, or \citet{barthe2020foundations}.
We give a brief (and incomplete) overview of existing semantic models for PPLs.

In recent years, a lambda calculus for probabilistic programming was presented by~\citet{borgstrom2016lambda}.
This work and a domain theoretic approach for probabilistic programming~\cite{vakar2019domain} had an influence on the development of the statistical PCF~\cite{mak2021densities} (programming computable functions), which was used to prove differentiability of terminating probabilistic programs.
Probabilistic programming was also examined from a categorical perspective~\cite{heunen2017convenientcategory} and the semantics of higher order probabilistic programs is also researched~\cite{dahlqvist2019semanticshigherorder,staton2016semanticshigherorder}.


As mentioned, the semantics presented in this work extend those of \citet{gorinova2019slicstan} which formalises the core language constructs of Stan by interpreting a program as a function over a fixed set of finite variables.
In our semantics, we interpret programs as functions over traces which are mappings from addresses to values similar to dictionaries.
Another dictionary based approach to probabilistic programming can be found in \citet{cusumano2020automatingimcmc} and denotational semantics based on trace types are introduced in \citet{lew2019tracetypes}.
To the best of our knowledge there are only two existing semantic models that support user-defined addresses~\cite{lew2023stochasticprobs,lee2019towards}.
Note that denotational semantics of a simplified version of the address-based PPL Gen are given in \citet{cusumano2020genthesis}.

Furthermore, the connection of simple probabilistic programs to Bayesian networks is known~\cite{van2018introppl}.
\citet{borgstrom2011measure-transformer} consider a small imperative language and give semantics in terms of factor graphs -- a concept related to Bayesian networks.
\citet{sampson2014assertions} compile probabilistic programs to Bayesian networks by treating loops as black-box functions in order to verify probabilistic assertion statements.
Previous work extended Bayesian networks to so-called contingent Bayesian networks in order to support models with an infinite number of random variables or cyclic dependencies~\cite{milch2005contingentbayesnet}, but they cannot be straightforwardly applied to PPLs with  user-labelled sample statements and while loops.
\citet{paquet2021bayesianstrategies} presented a interpretation of probabilistic programs in game semantics, which can be seen as a refinement of Bayesian networks.
Most recently, \citet{faggian2024higher} {\revcolor introduced a higher-order functional language for discrete Bayesian networks.}

\subsection{Graphical Representations, Optimisations, and Dependency Analysis}

The advantages of representing a probabilistic model graphically are well-understood~\cite{koller2009pga}.
For this reason, some PPLs like PyMC~\cite{salvatier2016pymc}, Factorie~\cite{mccallum2009factorie}, or Infer.NET~\cite{minka2012infer} require the user to explicitly define their model as a graph.
This is not possible if we want to support while loops in our programs.
There are also other graphical representations that allow for optimised inference like binary decision diagrams~\cite{holtzen2020dice} or sum-product networks~\cite{saad2021sppl}.
Venture~\cite{mansinghka2018venture} represents traces graphically to facilitate efficient recomputation of subparts of the program.
Similar to Gen~\cite{cusumano2019gen}, Pyro~\cite{bingham2019pyro} implements the \texttt{plate} and \texttt{markov} constructs to model independence and reduce variance in {\revcolor variational inference}.

Modifying the Metropolis Hastings algorithm to exploit the structure of the model is also common practice.
\citet{wu2016swift} optimised BLOG to dynamically keep track of the Markov blanket of a variable among other improvements.
For their Metropolis Hastings implementation, they reported significant speed-ups for the Hurricane model, instances of the Urn model, and GMMs with unknown number of clusters.
The Shred system~\cite{yang2014generatingefficient} traces Church~\cite{goodman2012church} programs  {\revcolor and} slices them to identify the minimal computation needed to compute the MH acceptance probability, and translates them to C++.
However, their approach comes with high compilation cost, {\revcolor is not proven correct nor} open-source.
\citet{hur2014slicing} {\revcolor developed a different theory for slicing probabilistic programs.}
In their unpublished implementation, they observe significant runtime improvements if  statements irrelevant to the return expression {\revcolor are removed from the program}.
In contrast, we slice programs with respect to the dependencies between sample statements in the program and there are no redundancies in the programs of our benchmark set.
\citet{castellan2019intensional} present a theoretical foundation for incremental computation in the Metropolis Hastings algorithm by graphically modelling the data-dependencies between sample statements.
However, they {\revcolor only consider a first-order probabilistic programs and did not implement their approach}.

{\revcolor Furthermore, there are many more optimisation opportunities to reduce the variance in variational inference like control variates~\cite{paisley2012vicontrolvariates} or path-derivatives~\cite{roeder2017stickingthelanding}.
For a review on advances in variational inference see \citet{zhang2018advancesinvi}.
Also, there are many methods to facilitate iterative implementations of SMC: WebPPL~\cite{webppl} and Anglican~\cite{tolpin2015anglican} use a continuation-passing-style transform; Turing~\cite{ge2018turing} uses a tape of instructions; some PPLs like Birch~\cite{murray2018birch} require the user to explicitly write the model in an iterative fashion.}



Lastly, the probabilistic dependency structure of programs has also been exploited for applications different from optimising inference.
\citet{baah2008probabilisticdependencygraph} leverage the structure for fault diagnosis and \citet{bernstein2020transforming} extract the factor graph of programs to automate model checking for models implemented in Stan.
In general, static analysis for probabilistic programming is an active research field.
For instance, \citet{lee2019towards} developed a static analysis to verify so-called guide programs for stochastic variational inference in Pyro.
\citet{gorinova2021conditionalindependence} developed an information flow type system and program transformation to automatically marginalise out discrete variables in Stan.
Another example is the extension of the Stan compiler with a pedantic mode that statically catches common pitfalls in probabilistic programming~\cite{bernstein2023abstractions}.
Recently, a framework for statically analysing probabilistic programs in a language-agnostic manner was presented~\cite{boeck2024lasapp}.
For a 2019 survey on static analysis for probabilistic programming see \citet{bernstein2019static}.

\section{Conclusion}\label{sec:conclusion}
In this work, we addressed the open question to what extent a probabilistic program with user-labelled sample statements and while loops can be represented graphically.
We built on existing formal semantics to support these language features. We developed a sound static analysis that approximates the dependency structure of random variables in the program and used this structure to factorise the implicitly defined program density.
The factorisation is {\revcolor equivalent to the known Bayesian network factorisation in a restricted case, but novel for programs with dynamic addresses and loops.
Lastly, we demonstrated that by generating a sub-program for each factor based on the static dependency structure, we can 1) speed-up a single-site Metropolis Hastings algorithm; 2) reduce gradient variance in variational inference; 3) facilitate an iterative implementation of sequential Monte Carlo.
This work is the first to make these optimisations statically available for the considered program class.
We proved correctness and empirically showed that our approach performs comparably or better than existing methods.
}

\section*{Data-Availability Statement}
The source code to reproduce the experiments of \Cref{sec:casestudy} is made available on Zenodo~\citerepl{} and for reuse on \url{https://github.com/ipa-lab/PPLStaticFactor}.
The shared repository includes the implementation of the static provenance analysis (\Cref{sec:staticprov}) for a research PPL, the generation of sub-programs according to the dependency structure (\Cref{sec:casestudy}), as well as a standard and optimised implementations of the LMH, {\revcolor BBVI, and SMC} algorithms.
It also includes the 16 benchmark programs (\cref{sec:eval-results}) written in our research PPL, with and without combinators in Gen, in recursive fashion in WebPPL, {\revcolor and in Pyro}.

\begin{acks}
    The authors gratefully acknowledge the financial support under the scope of the COMET program within the K2 Center “Integrated Computational Material, Process and Product Engineering (IC-MPPE)” (Project No \grantnum{FFG}{886385}). This program is supported by the Austrian Federal Ministries for Economy, Energy and Tourism (BMWET) and for Innovation, Mobility and Infrastructure (BMIMI), represented by the \grantsponsor{FFG}{Austrian Research Promotion Agency (FFG)}{https://www.ffg.at/}, and the federal states of Styria, Upper Austria and Tyrol.
    Additionally, this research has been partially funded by the \grantsponsor{FWF}{Austrian Science Fund (FWF)}{https://www.fwf.ac.at} under grant \grantnum[https://doi.org/10.55776/PIN3275223]{FWF}{10.55776/PIN3275223}.
\end{acks}

\bibliographystyle{ACM-Reference-Format}
\bibliography{references}

\newpage
\appendix
\section{Control-Flow Graph Construction}

\begin{definition}
A CFG is a tuple $(N_\textnormal{start}, N_\textnormal{end}, \mathcal{N}, \mathcal{E})$ comprised of start and end node, a set of branch, join, and assign nodes, $\mathcal{N}$, and a set of edges, $\mathcal{E} \subseteq \mathcal{N}^*\times \mathcal{N}^*$, where $\mathcal{N}^*=\mathcal{N} \cup \{N_\textnormal{start}, N_\textnormal{end}\}$.
\end{definition}

\begin{itemize}
    \item The CFG of a skip-statement, $S = \kw{skip}$, consists of only start and end node:\newline
    \includegraphics[scale=1.]{tikz/tikz-skip.pdf}
    
    $G = \big(N_\textnormal{start}, N_\textnormal{end}, \emptyset, \{(N_\textnormal{start},N_\textnormal{end})\}\big)$
    
    \medskip
    
    \item  The CFG of an assignment, $x = E$, or sample statement $x = \kw{sample}(E_0, f(E_1,\dots,E_n))$, is a sequence of start, assign, and end node:\newline
    
    \includegraphics[scale=1.]{tikz/tikz-assign.pdf}\newline
    \includegraphics[scale=1.]{tikz/tikz-sample.pdf}

    \[G = \big(N_\textnormal{start}, N_\textnormal{end}, \{A\}, \{(N_\textnormal{start},A), (A,N_\textnormal{end})\}\big),\]
    where $A = \textnormal{Assign}(x = \dots)$ is the corresponding assign node.

    \medskip
    \item The CFG for a sequence of statements, $S = S_1; S_2$, is recursively defined by "stitching together" the CFGs $\cfg_1=(N_\textnormal{start}^1, N_\textnormal{end}^1, \mathcal{N}^1, \mathcal{E}^1)$ and $\cfg_2=(N_\textnormal{start}^2, N_\textnormal{end}^2, \mathcal{N}^2, \mathcal{E}^2)$ of $S_1$ and $S_2$:\newline

    \includegraphics[scale=1.]{tikz/tikz-seq.pdf}
    \begin{align*}
        G = \big(&N_\textnormal{start}^1,\, N_\textnormal{end}^2, \, \mathcal{N}^1 \cup \mathcal{N}^2,\, \mathcal{E}^1 \setminus \{(N^1_e,N_\textnormal{end}^1)\} \cup  \mathcal{E}^2 \setminus \{(N_\textnormal{start}^2,N^2_s)\} \cup \{(N^1_e,N^2_s)\} \big),
    \end{align*}
    where $N^1_e$ is the predecessor node of $N_\textnormal{end}^1$ in $\cfg_1$, $(N^1_e,N_\textnormal{end}^1) \in \mathcal{E}^1$, and $N^2_s$ is the successor node of $N_\textnormal{start}^2$ in $\cfg_2$, $(N_\textnormal{start}^2,N^2_s) \in \mathcal{E}_2$.

    \medskip
    \item The CFG of an if statement, $S=(\kw{if}\;E\;\kw{then}\;S_1\,\kw{else}\;S_2)$, is also defined in terms of the CFGs of $S_1$ and $S_2$ together with a branch node $B = \textnormal{Branch}(E)$ and a join node $J$:  \newline
    \includegraphics[scale=1.]{tikz/tikz-if.pdf}
    \begin{align*}
        G = \big(&N_\textnormal{start},\, N_\textnormal{end}, \, \mathcal{N}^1 \cup \mathcal{N}^2 \cup \{B, J\},\\
        &\mathcal{E}^1 \setminus \{(N_\textnormal{start}^1,N^1_s), (N^1_e,N_\textnormal{end}^1)\} \cup  \mathcal{E}^2 \setminus \{(N_\textnormal{start}^2,N^2_s), (N^2_e,N_\textnormal{end}^2)\} \, \cup\\
        &\quad\{(N_\textnormal{start},B), (B,N^1_s), (B,N^2_s), (N^1_e,J), (N^2_e,J), (J,N_\textnormal{end})\}\big),
    \end{align*}
    The nodes $N^1_s$, $N^1_e$, $N^2_s$, $N^2_e$ are defined as in the last case.
    This construction assumes that $\mathcal{E}^1 \neq \emptyset$ and $\mathcal{E}^2 \neq \emptyset$.
    Otherwise, the branch would be a $\kw{skip}$ statement and we would connect $B$ with $J$ directly.
    The two predecessors of $J$ are denoted with $\textnormal{predcon}(J)$ and $\textnormal{predalt}(J)$ depending on the path.

    \medskip
    \item The CFG of $(\kw{while}\;E\;\kw{do}\;S)$ is given in terms of the sub-cfg $\cfg_S$ of statement $S$. \newline

    \includegraphics[scale=1.]{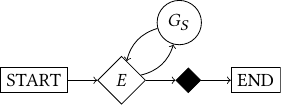}\newline
 
    \begin{align*}
        G = \big(&N_\textnormal{start},\, N_\textnormal{end}, \, \mathcal{N}^S \cup \{B, J\},\\
        &\mathcal{E}^S \setminus \{(N_\textnormal{start}^S,N^S_s), (N^S_e,N_\textnormal{end}^S)\} \cup \{(N_\textnormal{start},B), (B,N^S_s), (B,J), (N^S_e,B),(J,N_\textnormal{end})\}\big),
    \end{align*}
    where $\cfg_S=(N_\textnormal{start}^S, N_\textnormal{end}^S, \mathcal{N}^S, \mathcal{E}^S)$.

    Again, if $S = \kw{skip}$, then we would have the edge $(B,B)$ instead of $(B,N^S_s)$ and $(N^S_e,B)$.
    \medskip

    \item The unrolled CFG of $(\kw{while}\;E\;\kw{do}\;S)$ is defined as follows:\newline
    
    \includegraphics[scale=1.]{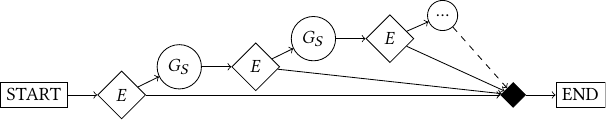}

    \begin{align*}
        G = \big(&N_\textnormal{start},\, N_\textnormal{end}, \\&
        \medcup_{i\in\Nats}\mathcal{N}^{S_i} \cup \{B_i\colon i\in\Nats\}\cup \{J\},\\
        &\medcup_{i\in\Nats}\mathcal{E}^{S_i} \setminus \{(N_\textnormal{start}^{S_i},N^{S_i}_s), (N^{S_i}_e,N_\textnormal{end}^{S_i})\} \cup \{(N_\textnormal{start},B_1), (B,J),(J,N_\textnormal{end})\} \, \cup \\ &\quad \{(B_i, N^{S_i}_s)\colon i\in\Nats\} \cup \{(N^{S_i}_e, B_{i+1})\colon i\in\Nats\}\big),
    \end{align*}
    where $\cfg_{S_i}=(N_\textnormal{start}^{S_i}, N_\textnormal{end}^{S_i}, \mathcal{N}^{S_i}, \mathcal{E}^{S_i})$ is the $i$-th copy of sub-graph $\cfg_{S}$.
    $N^{S_i}_s$ and $N^{S_i}_e$ are defined as in the cases before.
    
    If $S = \kw{skip}$, then we would have the edges $(B_i,B_{i+1})$ instead of the edges $(B_i, N^{S_i}_s)$ and $(N^{S_i}_e, B_{i+1})$.

\end{itemize}

\newpage

\section{Proofs of Sections 2-4}
\subsection{Proof of \cref{theorem:CFG-semantics-equivalence}}\label{proof:CFG-semantics-equivalence}
\begin{statement}
For all programs $S$ without while loops and corresponding control-flow graph $\cfg$, it holds that for all traces $\dict$ {\revcolor and all initial states $\sigma_0$}
\[{\revcolor \bar{p}_S(\dict,\sigma_0) = \bar{p}_\cfg(\dict,\sigma_0).}\]
\end{statement}
\begin{proof}
With structural induction we prove the more general equivalence:
\[(\progstate, S) \opsemrel \progstate' \iff  (\progstate, \text{START}) \underset{S}{\cfgrel} \dots \underset{S}{\cfgrel} (\progstate', \text{END}).\]

We write $\underset{S}{\cfgrel}$ to indicate that the predecessor-successor relationship follow the CFG of program~$S$.

For programs consisting of a single skip, assignment, or sample statement, the equality can be seen by directly comparing the operational semantic rules to the graph semantics.

For sequence $S = (S_1;S_2)$, by induction assumption we have
\begin{align*}
(\progstate, S_1) \opsemrel \progstate' &\iff  (\progstate, N_\text{start}^1) \underset{S_1}{\cfgrel} \dots \underset{S_1}{\cfgrel} (\progstate'_1, N^1) \underset{S_1}{\cfgrel}  (\progstate', N_\text{end}^1)\\
(\progstate', S_2) \opsemrel \progstate'' &\iff  (\progstate', N_\text{start}^2) \underset{S_2}{\cfgrel} (\progstate', N^2) \underset{S_2}{\cfgrel} \dots \underset{S_2}{\cfgrel} (\progstate'', N_\text{end}^2)
\end{align*}
By definition of the CFG of $S$, we see that 
\begin{displaymath}
(\progstate, S) \opsemrel \progstate'' \iff  (\progstate, N_\text{start}^1) \underset{S}{\cfgrel} \dots \underset{S}{\cfgrel} (\progstate'_1, N^1) \underset{S}{\cfgrel} (\progstate', N^2) \underset{S}{\cfgrel} \dots \underset{S}{\cfgrel}  (\progstate'', N_\text{end}^2)
\end{displaymath}

In a similar way, we proof the result for if statements $S = (\kw{if}\;E\;\kw{then}\;S_1\,\kw{else}\;S_2)$ by considering the two cases $\progstate(E) = \true$ and $\progstate(E) = \false$.
\begin{displaymath}
(\progstate, S) \opsemrel \progstate' \iff  (\progstate, N_\text{start}) \underset{S}{\cfgrel} (\progstate, \text{Branch}(E)) \underset{S}{\cfgrel} \dots \underset{S}{\cfgrel} (\progstate', N_\text{join}) \underset{S}{\cfgrel}  (\progstate', N_\text{end})
\end{displaymath}
Depending on the case, we replace the dots with the path of $S_i$ without the start and end nodes. Note that in our notation in both cases $(\progstate, S_i) \opsemrel \progstate'$.
\end{proof}

\subsection{Proof of \Cref{theorem-evalf}}

To prove soundness of \Cref{alg:staticprov}, we list properties of the provenance set that will become useful in subsequent proofs.
\begin{lemma}\label{lemma-staticprov}
Let $N$ be a CFG node and $x$ a variable. $\staticprov$ has following properties:
\begin{itemize}
    \item For all nodes $N'$ it holds that
    \begin{align}\label{eq:staticprov-same-rd}
        \RD(N,x) = \RD(N',x) \implies \staticprov(N,x) = \staticprov(N',x).
    \end{align}
    \item If $N' \in \RD(N,x)$ is an assignment node, $N' = \textnormal{Assign}(x = E)$, then
    \begin{align}\label{eq:staticprov-assign}
        \staticprov(N',E) = \medcup_{y \in \textnormal{vars}(E)} \staticprov(N',y) \subseteq \staticprov(N,x).
    \end{align}
    \item if $N' \in \RD(N,x)$ is a sample node, $N' = \textnormal{Assign}(x = \kw{sample}(E_0,\dots))$, then
    \begin{align}\label{eq:staticprov-sample}
        \staticprov(N',E_0) \subseteq \staticprov(N,x)\quad \text{and}\quad \textnormal{addresses}(N') \subseteq \staticprov(N,x).
    \end{align}
    \item For all $N' \in \RD(N,x)$ it holds that for all branch parents $ N_\text{bp} = \textnormal{Branch}(E_{N_\text{bp}}) \in \CP(N')$  we have $\staticprov(N_\text{bp},E_{N_\text{bp}}) \subseteq \staticprov(N,x)$ and 
    \begin{gather}\label{eq:staticprov-control}
    \begin{split}
        \medcup_{N_\text{bp} \in \CP(N')} \staticprov(N_\text{bp},E_{N_\text{bp}}) = \medcup_{N_\text{bp} \in \CP(N')} \medcup_{y \in \textnormal{vars}(E_{N_\text{bp}})} \staticprov(N_\text{bp},y) \subseteq \staticprov(N,x).
    \end{split}
    \end{gather}
\end{itemize}
\end{lemma}
\begin{proof}
If $(N',y)$ is pushed into the queue, then $\staticprov(N',y) \subseteq \staticprov(N,x)$ as
the algorithm for computing $\staticprov(N',y)$ starts with only $(N',y)$ in the queue.
With this fact the properties follow directly from the definition of the algorithm.
\end{proof}

\begin{statement}
    Let $\cfg$ be the CFG of a program $S$ without while loop statements.
    For each node $N$ and variable~$x$, there exists an evaluation function
    \[\evalf{N}{x} \in \restricedfs{\staticprov(N,x)}{\values},\] such that for all traces $\dict$, {\revcolor initial program state $\progstate_0$ as in \Cref{def:prog-semantics}}, and execution sequence
    \[(\progstate_0, \text{START}) \cfgrel \cdots  \cfgrel (\progstate_i, N_i) \cfgrel \cdots \cfgrel (\progstate_l, \text{END})\]
    we have
    \begin{displaymath}
    \progstate_i(x) = \evalf{N_i}{x}(\dict).
    \end{displaymath}
\end{statement}
In the proof we also use the evaluation function of expressions $\evalf{N}{E} \in \restricedfs{\staticprov(N,E)}{\values}$, where  $\evalf{N}{c}(\dict) = c$ and   $\evalf{N}{g(E_1,\dots,E_n)}(\dict) = g(\evalf{N}{E_1}(\dict),\dots,\evalf{N}{E_n}(\dict)).$
\begin{proof}\hfill\newline
    \textbf{Step 1. Defining $\evalf{N}{x}$ and proving $\evalf{N}{x} \in \restricedfs{\staticprov(N,x)}{\values}$}.
    
    We begin by defining $\evalf{N}{x}$ inductively by noting that every node $N \neq \text{START}$ has exactly one predecessor except for join nodes which have two predecessors. Further, the CFG $\cfg$ is a finite directed acyclic graph with a single root node which makes mathematical induction possible.
    \medskip

    \textbf{Base case.} For $((x_i \mapsto \none, \probvar \mapsto 1), \text{START})$, let $\evalf{\text{START}}{\probvar}(\dict) = 1$ and  $\evalf{\text{START}}{x}(\dict) = \none$ for all variables~$x$. We have $\evalf{\text{START}}{x} \in \restricedfs{\emptyset}{\values}$.
    \medskip
    
    \textbf{Induction step.}
    
    \textbf{Case 1.} Let $N$ be a node with single predecessor $N'$ ($N$ is not a join node).    
    We define $\evalf{N}{x}$ based on the node type of $N'$, for which we make the induction assumption that the evaluation function $\evalf{N'}{x} \in \restricedfs{\staticprov(N',x)}{\values}$ exists for all variables.
    \begin{itemize}
        \item \textbf{Case 1.1.} $N'$ is an assignment node, $N' = \textnormal{Assign}(z = E)$.
        Define
        \begin{displaymath}
            \evalf{N}{y}(\dict) =
            \begin{cases}
                \evalf{N'}{E}(\dict) & \text{if}\, y = z\\
                \evalf{N'}{y}(\dict) & \text{otherwise}.
            \end{cases}
        \end{displaymath}
        By assumption $\evalf{N'}{y} \in \restricedfs{\staticprov(N',y)}{\values}$ for all $y$ and thus, $\evalf{N'}{E} \in \restricedfs{\staticprov(N',E)}{\values}$.
        \medskip
        
        For $x \neq z$, $\RD(N,x) = \RD(N',x)$ and by \cref{eq:staticprov-same-rd} we have $\staticprov(N,x) = \staticprov(N',x)$, implying that $\evalf{N}{x} \in \restricedfs{\staticprov(N,x)}{\values}$.

        For $x = z$, $\RD(N,x) = \{N'\}$ since $N'$ is the single predecessor of $N$ and assigns $x$. 
        By \cref{eq:staticprov-assign} $\staticprov(N',E) \subseteq \staticprov(N,x)$, and by \cref{eq:restrictedfs-subset}
        \[\evalf{N}{x} = \evalf{N'}{E} \in \restricedfs{\staticprov(N',E)}{\values} \subseteq \restricedfs{\staticprov(N,x)}{\values}.\]

        \item \textbf{Case 1.2.} $N'$ is a sample node, $N' = \textnormal{Assign}(z = \kw{sample}(E_0, f(E_1,\dots,E_n)))$.
        Define
        \begin{displaymath}
            \evalf{N}{y}(\dict) =
            \begin{cases}
                \dict(\evalf{N'}{E_0}(\dict)) & \text{if}\, y = z\\
                \evalf{N'}{y}(\dict) & \text{otherwise.}
            \end{cases}
        \end{displaymath}

        As before, for $x \neq z$, $\evalf{N}{x} \in \restricedfs{\staticprov(N,x)}{\values}$.
        
        For $x = z$, by definition $\evalf{N'}{E_0}(\dict) \in \textnormal{addresses}(N')$.
        Since $N' \in \RD(N,x) =\{N'\}$ it follows from \cref{eq:staticprov-sample} that $\textnormal{addresses}(N') \subseteq \staticprov(N,x)$ and $\staticprov(N',E_0) \subseteq \staticprov(N,x)$. By \cref{eq:restrictedfs-subset}
        \[\evalf{N}{x} \in \restricedfs{\textnormal{addresses}(N')\cup\staticprov(N',E_0)}{\values} \subseteq \restricedfs{\staticprov(N,x)}{\values}.\]

        \item \textbf{Case 1.3.} $N'$ is a branch or join node. Define $\evalf{N}{y}=\evalf{N'}{y}$ for all variables $y$.
    \end{itemize}

    \textbf{Case 2.} Let $J$ be a join node with two predecessor nodes, $N'_1 = \text{predcons}(J)$, $N'_2 = \text{predalt}(J)$.
    Let $B = \text{Branch}(E)$ be the corresponding branching node, $(B,J) \in \textnormal{BranchJoin}(\cfg)$.
    
    For variable $x$, there are two cases:
    \begin{enumerate}
        \item There exists a reaching definition $N' \in \RD(J,x)$ on a path $(B,\dots,N',\dots,J)$.
        Define $\evalf{1}{x}$ and $\evalf{2}{x}$ depending on the node type of $N'_1$ and $N'_2$ as in step 1.
        The functions $\evalf{i}{x}$ compute the value of $x$ after executing $N'_i$.
        As before one can see that $\evalf{i}{x} \in \restricedfs{\staticprov(J,x)}{\values}$.
        Let
        
        \begin{displaymath}
            \evalf{J}{x}(\dict) \coloneqq \textnormal{ife}(\evalf{B}{E}(\dict), \evalf{1}{x}(\dict), \evalf{2}{x}(\dict)) = \begin{cases}
                \evalf{1}{x}(\dict) & \text{if } \evalf{B}{E}(\dict) = \true, \\
                \evalf{2}{x}(\dict) & \text{otherwise.}
            \end{cases}
        \end{displaymath}

        Since $B \in \CP(N')$, by \cref{eq:staticprov-control} $\staticprov(B, E) \subseteq \staticprov(J,x)$ and thus $\evalf{J}{x} \in \restricedfs{\staticprov(J,x)}{\values}$.
        
        \item All reaching definitions of $x$ (if there are any) are predecessors of $B$.
        Then, $\RD(J,x) = \RD(B,x)$ and by \cref{eq:staticprov-same-rd} $\staticprov(J,x) = \staticprov(B,x)$. Define $\evalf{J}{x} \coloneqq \evalf{B}{x}.$ 
        
    \end{enumerate}
\bigskip
\textbf{Step 2: Proving that $\progstate_i(x) = \evalf{N_i}{x}(\dict)$.}

Having defined the evaluation functions, we now have to prove that they indeed compute the correct values.
For trace $\dict$ let the corresponding execution sequence be
\[(\progstate_0, \text{START}) \cfgrel \cdots  \cfgrel (\progstate_i, N_i) \cfgrel \cdots \cfgrel (\progstate_l, \text{END}).\]
We will prove that $\progstate_i(x) = \evalf{N_i}{x}(\dict)$ by induction.
\medskip

\textbf{Base case}. The base case immediately follows from the definition $\progstate_0(x) = \evalf{\text{START}}{x}(\dict)$.

\textbf{Induction step}. For transition $(\progstate_i, N_i) \cfgrel (\progstate_{i+1}, N_{i+1})$ the assumption is that $\progstate_j(x) = \evalf{N_j}{x}(\dict)$ holds for all variables $x$ and $j\le i$.

If $N_{i+1}$ is not a join node, then $\progstate_{i+1}(x) = \evalf{N_{i+1}}{x}(\dict)$ follows directly from the CFG semantics and definition of $\evalf{N_{i+1}}{x}$ in case 1 of step 1.

Lastly, we consider the case when $N_{i+1}$ is a join node with branching node $N_j = \text{Branch}(E)$, $(N_j,N_{i+1}) \in \text{BranchJoin}(\cfg)$, for some $j \le i$. Let $N'_1 = \text{predcons}(N_{i+1})$, $N'_2 = \text{predalt}(N_{i+1})$.

By the semantics of if statements, $\evalf{B}{E}(\dict) = \true$ iff $N_i = N'_1$ and  $\evalf{B}{E}(\dict) = \false$ iff $N_i = N'_2$. 
If there is a node $N'$ between $N_j$ and $N_{i+1}$ in the execution sequence that assigns $x$, $N' \in \RD(N_{i+1},x)$, then by definitions of $\evalf{1}{x}$, $\evalf{2}{x}$, and $\evalf{N_{i+1}}{x}$ we have
\[\progstate_{i+1}(x) = \textnormal{ife}(\evalf{B}{E}(\dict), \evalf{1}{x}(\dict), \evalf{2}{x}(\dict)) = \evalf{N_{i+1}}{x}(\dict).\]
If there is no such $N'$, then
\[\progstate_{i+1}(x) = \progstate_{j}(x) = \evalf{N_j}{x}(\dict) = \evalf{N_{i+1}}{x}(\dict).\]

\end{proof}

\subsection{Proof of \cref{theorem:unrolled-CFG-semantics-equivalence}}\label{proof:unrolled-CFG-semantics-equivalence}
\begin{statement}
For all program $S$ with CFG $\cfg$ and unrolled CFG $\unrolledcfg$, it holds that for all traces $\dict$ and {\revcolor initial states $\sigma_0$}
\[{\revcolor \bar{p}_S(\dict,\sigma_0) = \bar{p}_{\cfg}(\dict,\sigma_0) = \bar{p}_{\unrolledcfg}(\dict,\sigma_0).}\]
\end{statement}
\begin{proof}
    We continue the proof of \cref{theorem:CFG-semantics-equivalence} for a while statement.   
    Note the well-known equivalence for while statements in terms of the operational semantics:
    \[(\progstate, \kw{while}\;E\;\kw{do}\;S) \opsemrel \progstate' \iff (\progstate, \kw{if}\; E\;\kw{then}\; (S;\; \kw{while}\;E\;\kw{do}\;S)\;\kw{else}\;\kw{skip}) \opsemrel \progstate'\]
    Thus, if the program terminates for $\dict$, it is equivalent to $S;\dots;S$ ($S$ repeated $n$ times), where $\progstate_0 = \progstate$, $(\progstate_i, S) \opsemrel \progstate_{i+1}$, $\progstate_n = \progstate'$, such that $\progstate_i(E) = \true$ for $i < n$, $\progstate_n(E) = \false$.

    In this case, it can be shown that the (unrolled) CFG of $S;\dots;S$ is also equivalent to the (unrolled) CFG of the while statement.

    If the program does not terminate, then ${\revcolor \bar{p}_S(\dict,\sigma_0) = \bar{p}_{\cfg}(\dict,\sigma_0) = \bar{p}_{\unrolledcfg}(\dict,\sigma_0)} = \textnormal{undefined}$.
\end{proof}

\subsection{Proof of \cref{lemma:staticprov-unrolled-cfg}}\label{proof:staticprov-unrolled-cfg}
\begin{statement}
 Let $\cfg$ be the CFG and $\unrolledcfg$ the unrolled CFG of program $S$. For each node $M \in \unrolledcfg$ following equation holds.
    \begin{equation}
         \{\cfgnode(M')\colon M' \in\RD(M,x)\} \subseteq \RD(\cfgnode(M), x)
    \end{equation}
\end{statement}
\begin{proof}
    We will prove following Lemma by structural induction:
    
    \textsc{Lemma. }{\it For every path $\vv{M} = (M',\dots,M)$ in the unrolled CFG $\unrolledcfg$, there is a path in the CFG $\cfg$, $\vv{N} = (N',\dots,N)$, such that
    \[\cfgnode(M') = N', \quad \cfgnode(M) = N,\]
    \[\textnormal{AssignNodes}(\vv{N},x) = \left\{\cfgnode(M_i)\colon M_i \in \textnormal{AssignNodes}(\vv{M},x)\right\} = \textnormal{AssignNodes}_\textnormal{cfg}(\vv{M},x),\]
    for all variables $x$, where $\textnormal{AssignNodes}(\vv{N},x)$ is the set of all CFG nodes that assign $x$ in path $\vv{N}$ and $\textnormal{AssignNodes}_\textnormal{cfg}(\vv{M},x)$ is the set of all assign nodes in path $\vv{M}$ mapped to $\cfg$ with $\cfgnode$.
    }
    \medskip
    
    With this lemma we can prove the statement:

    Let $M' \in \RD(M,x)$. Thus, there exists a path $\vv{M} = (M',\dots,M)$, such that $M'$ is the only assign node for $x$ before $M$.
    By the lemma, there is a path in the CFG $\vv{N} = (N',\dots,N)$, such that $\cfgnode(M') = N'$ and $\cfgnode(M) = N$.
    
    Further, since $\textnormal{AssignNodes}(\vv{N},x) = \textnormal{AssignNodes}_\textnormal{cfg}(\vv{M},x) \subseteq \{\cfgnode(M'), \cfgnode(M)\} = \{N',N\}$, the node~$N'$ is also the only node in the path $\vv{N}$ before $N$ that assigns $x$.
    
    Thus, $\cfgnode(M') = N' \in \RD(N,x) = \RD(\cfgnode(M),x).$
    \medskip
    
    \textsc{Proof of Lemma.} By structural induction.

    For CFGs consisting only of a single assign node, the unrolled CFGs are identical and the statement follows.

    For a sequence of two statements $(S_1; S_2)$, let $\cfg$ be the CFG with sub-graphs $\cfg_1$ and $\cfg_2$ for statements $S_1$ and $S_2$ respectively.
    Let $\unrolledcfg$ the unrolled CFG with sub-graphs $\unrolledcfg_1$ and $\unrolledcfg_2$.
    We consider the only interesting case where $M' \in \unrolledcfg_1$ and $M \in \unrolledcfg_2$:
    \begin{displaymath}
        \vv{M} = (\underbrace{M,\dots,M_i}_{\in \unrolledcfg_1}, \underbrace{M_{i+1},\dots,M}_{\in \unrolledcfg_2}).
    \end{displaymath}
    For sub-paths $\vv{M_1} \coloneqq (M,\dots,M_i, \textnormal{END})$ and $\vv{M_2} \coloneqq (\textnormal{START},M_{i+1},\dots,M)$ we apply the induction assumption to get two paths $\vv{N_1} = (N',\dots,N_1, \text{END})$, $\vv{N_2} = (\text{START}, N_2,\dots,N)$, with $\cfgnode(M') = N'$ and  $\cfgnode(M) = N$.
    The corresponding CFG path in $\cfg$ is $\vv{N} \coloneqq (N',\dots,N_1, N_2, \dots, N)$ as
    \begin{align*}
        \textnormal{AssignNodes}_\textnormal{cfg}(\vv{M},x) &= \textnormal{AssignNodes}_\textnormal{cfg}(\vv{M_1},x) \cup \textnormal{AssignNodes}_\textnormal{cfg}(\vv{M_2},x)\\
        &= \textnormal{AssignNodes}(\vv{N_1},x) \cup \textnormal{AssignNodes}(\vv{N_2},x) = \textnormal{AssignNodes}(\vv{N},x).
    \end{align*}
    
    For an if statement, let $\cfg$ be the CFG with branch node $\tilde B$, join node $\tilde J$, and sub-graphs $\cfg_1$, $\cfg_2$.
    Let $\unrolledcfg$ be the unrolled CFG with corresponding $B$, $J$, $\unrolledcfg_1$, and $\unrolledcfg_2$.
    We only show the statement for the case
    \[\vv{M} = (B,\underbrace{M',\dots,M}_{\in \unrolledcfg_i}, J).\]
    Again, we apply the induction assumption to $\vv{M_i} \coloneqq (\textnormal{START}, M',\dots,M, \textnormal{END})$ and get a path in~$\cfg_i$: $\vv{N_i} = (\textnormal{START}, N',\dots, N, \textnormal{END})$. Since $\cfgnode(B) = \tilde B$, $\cfgnode(J) = \tilde J$, which are not assign nodes, and $\textnormal{AssignNodes}(\vv{N_i}) = \textnormal{AssignNodes}_\textnormal{cfg}(\vv{M_i})$, the statement follows for the path $(\tilde B, N',\dots, N, \tilde J)$.
    \medskip
    
    Lastly, for a while loop, let $\cfg$ be the CFG with branch node $\tilde B$, join node $\tilde J$, and sub-graph $\cfg_S$. Let $\unrolledcfg$ be the unrolled CFG with $B_i$ branch nodes, $J$ join node, and $\unrolledcfg_i$ sub-graphs. For the sake of brevity we consider three cases:

    \begin{itemize}
        \item For $M' \in \unrolledcfg_i$ and $M \in \unrolledcfg_j$
        \[\vv{M} = (\underbrace{M',\dots,M_i}_{\in \unrolledcfg_i},B_{i+1},\dots,B_j,\underbrace{M_j,\dots,M}_{\in \unrolledcfg_j})\]
        we apply the induction assumption to the sub-paths and construct the path
        \[\vv{N} = (\underbrace{N',\dots,N_i}_{\in \cfg_S},\tilde B,\dots,\tilde B,\underbrace{N_j,\dots,N}_{\in \cfg_S}))\]
        where $\tilde B$ and $B_i$ are do not contribute to $\textnormal{AssignNodes}(\vv{N})$ or $\textnormal{AssignNodes}(\vv{M})$.
        \item For  $M' \in \unrolledcfg_i$ and $J$
        \[\vv{M} = (\underbrace{M',\dots,M_i}_{\in \unrolledcfg_i},B_{i+1},\dots,B_l,J)\]
        we construct the path
        \[\vv{N} = (\underbrace{N',\dots,N_i}_{\in \cfg_S},\tilde B,\dots, \tilde B, \tilde J)\]
        \item For the path $\vv{M} = (B_l,J)$  the path in the CFG is $\vv{N} = (\tilde B, \tilde J)$.
    \end{itemize}    
\end{proof}

\begin{statement}
 Let $\cfg$ be the CFG and $\unrolledcfg$ the unrolled CFG of program $S$. For each node $M \in \unrolledcfg$ the following equation holds.
    \begin{equation}
        \CP_\textnormal{cfg}(M) \coloneqq \{\cfgnode(M')\colon M' \in\CP(M)\} = \CP(\cfgnode(M))
    \end{equation}
\end{statement}
\begin{proof}
    We prove this statement by structural induction.
    \medskip
    
    For CFGs only consisting of a single assign node, the unrolled CFGs are identical, there are no branch parents, and the statement follows.
    \medskip
    
    The sequencing of two statements $(S_1;S_2)$ does not change the branch parents of any node.
    \medskip

    For an if statement, let $\cfg$ be the CFG with branch node $\tilde B$, join node $\tilde J$, and sub-graphs $\cfg_1$, $\cfg_2$.
    Let $\unrolledcfg$ be the unrolled CFG with corresponding $B$, $J$, $\unrolledcfg_1$, and $\unrolledcfg_2$.
    Every node $M$ in $\unrolledcfg_i$ has branch parents $\left(\CP(M) \cap \unrolledcfg_i\right) \cup \{B\}$.
    
    By induction assumption $\CP_\textnormal{cfg}(M) \cap \cfg_i = \CP(\cfgnode(M)) \cap \cfg_i$ and thus
    \begin{align*}
        \CP_\textnormal{cfg}(M) &= \left\{ \cfgnode(M')\colon M' \in \CP(M) \cap \unrolledcfg_i \right\} \cup \{\cfgnode(B)\}\\
        &= \left(\CP(\cfgnode(M)) \cap \cfg_i\right) \cup \{\tilde B\} = \CP(\cfgnode(M)).
    \end{align*}
    Nodes $B$, $\tilde B$, $J$, and $\tilde J$ do not have branch parents. 
    \medskip
    
    Lastly, for a while loop, let $\cfg$ be the CFG with branch node $\tilde B$, join node $\tilde J$, and sub-graph $\cfg_S$. Let~$\unrolledcfg$ be the unrolled CFG with $B_i$ branch nodes, $J$ join node, and $\unrolledcfg_i$ sub-graphs.

    Every  node $M$ in $\unrolledcfg_i$  has branch parents $\left( \CP(M) \cap \unrolledcfg_i\right) \cup \{B_1,\dots,B_i\}$.

    The CFG node $\cfgnode(M) \in \cfg_S$ has branch parents $\left( \CP(\cfgnode(M)) \cap \cfg_S\right) \cup \{\tilde B\}$.

    Since $\cfgnode(B_j) = \tilde B$ and by induction assumption $\CP_\textnormal{cfg}(M) \cap \cfg_S = \CP(\cfgnode(M)) \cap \cfg_S$, we get
    \begin{align*}
    \CP_\textnormal{cfg}(M) &= \left\{\cfgnode(M')\colon M' \in \CP(M) \cap \unrolledcfg_i \right\} \cup \{\cfgnode(B_1),\dots,\cfgnode(B_i)\}\\
    &= \left(\CP(\cfgnode(M)) \cap \cfg_S\right) \cup \{\tilde B\} = \CP(\cfgnode(M)).
    \end{align*}
\end{proof}

\begin{statement}
Let $\cfg$ be the CFG and $\unrolledcfg$ the unrolled CFG of program $S$. For each node $M \in \unrolledcfg$ the following equation holds.
    \begin{equation}
        \staticprov(M,x) \subseteq \staticprov(\cfgnode(M),x)
    \end{equation}
\end{statement}
\begin{proof}
    Follows directly from the two equations $\{\cfgnode(M')\colon M' \in\RD(M,x)\} \subseteq \RD(\cfgnode(M), x)$ and $\{\cfgnode(M')\colon M' \in\CP(M)\} = \CP(\cfgnode(M))$, since $\staticprov(M,x)$ is defined in terms of $\RD$ and $\CP$.
\end{proof}

\subsection{Proof of \Cref{theorem-evalf-loopy}}

\begin{statement}
    Let $\unrolledcfg$ be the  \emph{unrolled} CFG of a program $S$. 
    For each node $M$  and variable $x$, there exists an evaluation function
    \[\evalf{M}{x} \in \restricedfs{\staticprov(M,x)}{\values},\] such that for all traces $\dict$, {\revcolor initial program state $\progstate_0$ as in \Cref{def:prog-semantics}}, and execution sequence
    \[(\progstate_0, \text{START}) \cfgrel \cdots  \cfgrel (\progstate_i, M_i) \cfgrel \cdots \cfgrel (\progstate_l, \text{END})\]
    we have
    \begin{displaymath}
    \progstate_i(x) = \evalf{M_i}{x}(\dict).
    \end{displaymath}
\end{statement}

\begin{proof}
    
    First, we have to justify the use of mathematical induction to prove properties on a potentially infinite graph.
    For $M_1,M_2 \in \unrolledcfg$ define the relation $M_1 \prec M_2 \Leftrightarrow M_1 \text{ is parent of } M_2$ (i.e. there is an edge from $M_1$ to $M_2$ in $\unrolledcfg$).
    By construction of $\unrolledcfg$, there are no infinite descending chains, $M_i \in \unrolledcfg$ for $i\in\Nats$ such that $M_{i+1} \prec M_i$.
    This makes $\prec$ a \emph{well-founded} relation, i.e. every non-empty subset has a minimal element.
    As a result, we may apply induction with respect to this relation (as has been done in the proof of \Cref{theorem-evalf}) for graphs with infinite nodes.
    This proof principle is called well-founded or Noetherian induction.

    To complete the proof, we only need to consider one additional case:
    
    \noindent
    \textbf{Case 3.} Let $J$ be a join node of a while loop with the parents $B_i = \textnormal{Branch}(E)$, $i \in \Nats$.
    The function $\evalf{B_i}{x} \in \restricedfs{\staticprov(B_i,x)}{\values}$ exists by induction assumption and computes the value of $x$ \emph{before} executing the while loop body for the $i$-th time.
    If there is a reaching definition $M' \in \RD(J,x)$ on a path $(B_i,\dots,M',\dots,J)$, define
    \begin{displaymath}
        \evalf{J}{x}(\dict) \coloneqq \begin{cases}
            \evalf{B_i}{x}(\dict) & \text{if } \exists i \in \Nats\colon\forall j < i\colon\evalf{B_j}{E}(\dict) = \true \land \evalf{B_{i}}{E}(\dict) = \false, \\
            \none & \text{otherwise.}
        \end{cases}
    \end{displaymath}
    Further, as $M'$ belongs to the sub-CFG of the while loop body, we have that for every $B_i$ there is a $M_i' \in \RD(J,x)$ on the path from $B_i$ to $J$.
    Since $B_i \in \CP(M_i')$, by \Cref{eq:staticprov-control} $\staticprov(B_i,E) \subseteq \staticprov(J,x)$ {\revcolor for all $i$} and thus $\evalf{J}{x} \in \restricedfs{\staticprov(J,x)}{\values}$.
    If there is no reaching definition in the while loop body sub-CFG, then $\staticprov(J,x) = \staticprov(B_1,x)$ and define $\evalf{J}{x} = \evalf{B_1}{x}$.
    
\end{proof}

\subsection{Proof of \Cref{theorem:loopy-factor}}

\begin{statement}
    Let $\cfg$ be the CFG for a program $S$.
    Let $N_1,\dots,N_K$ be all sample nodes in $\cfg$.
    For each sample node $N_k = \textnormal{Assign}(x_k = \kw{sample}(E_0^k, f^k(E_1^k,\dots,E_{n_k}^k)))$, let
    \[A_k = \textnormal{addresses}(N_k) \cup  \medcup_{i = 0}^{n_k} \staticprov(N_k, E_i^k) \cup \medcup_{N' \in \CP(N_k)}\staticprov(N',{\revcolor \textnormal{condexp}(N')}).\]
    \newline
    Then, there exist functions $p_k \in  \restricedfs{A_k}{\pdfvalues}$ such that {\revcolor for all $\dict\in\dicts$}, if $p_S(\dict) \neq \undefval$, then
    \[p_S(\dict) = p_\cfg(\dict) = p_{\unrolledcfg}(\dict) = \prod_{k=1}^K p_k(\dict).\]
\end{statement}

\begin{proof}
    Let $\unrolledcfg$ be the \emph{unrolled} CFG of program $S$.
    Denote the (countably many) sample nodes in $\unrolledcfg$ with
    $M_j = \textnormal{Assign}(x^j = \kw{sample}(E_0^j, f^j(E_1^j,\dots,E_{n_j}^j)))$.
    Like in the proof of Theorem~\ref{theorem:simple-factor}, define for each $M_j$, $b_j(\dict) \coloneqq \bigwedge_{M' \in \CP(M_j)} t^j_{M'}(\evalf{M'}{{\revcolor \textnormal{condexp}(M')}}(\dict))$ where
    \[b_j \in \restricedfs{\medcup_{M' \in \CP(M_j)}\staticprov(M',{\revcolor \textnormal{condexp}(M')})}{\{\true,\false\}}.\]
    Again, $b_j(\dict) = \true$ if $M_j$ is in the execution sequence for $\dict$ else $\false$. Define
    \[\tilde{p}_j(\dict) \coloneqq \dirac{b_j(\dict)} \pdf_{f^j}\left(\dict(\evalf{M_j}{E_0^j}(\dict)); \evalf{M_j}{E_1^j}(\dict), \dots, \evalf{M_j}{E_{n_j}^j}(\dict)\right) + (1-\dirac{b_j(\dict)}).\]
    We have $\tilde{p}_j \in \restricedfs{\tilde{A}_j}{\pdfvalues}$, for
    \[\tilde{A}_j = \textnormal{addresses}(M_j)\cup \medcup_{i = 0}^{n_j} \staticprov(M_j, E_i^j) \cup \medcup_{M' \in \CP(M_j)}\staticprov(M',{\revcolor \textnormal{condexp}(M')}).\]
    Finally, we group the nodes $M_j$ by their corresponding CFG node.
    \[p_k = \prod_{j\colon \cfgnode(M_j) = N_k} \tilde{p}_j\]
    If $\cfgnode(M_j) = N_k$, then $\tilde{A}_j \subseteq A_k$, since for $i=0,\dots,n$ we have $E_i^k = E_i^j$ and by \cref{eq:staticprov-subset}
    \[\staticprov(M_j,E_i^j) \subseteq \staticprov(\cfgnode(M_j),E_i^j) = \staticprov(N_k,E_i^k)\]
    and by \cref{eq:cp-subset} it holds that
    \[\{\cfgnode(M')\colon M' \in \CP(M_j)\} = \CP(\cfgnode(M_j))  = \CP(N_k)\]
    such that
    \begin{align*}
        \medcup_{M' \in \CP(M_j)}\staticprov(M',{\revcolor \textnormal{condexp}(M')}) &\subseteq \medcup_{M' \in \CP(M_j)}\staticprov(\cfgnode(M'),{\revcolor \textnormal{condexp}(\cfgnode(M'))})\\
        &= \medcup_{N' \in \CP(N_k)}\staticprov(N',{\revcolor \textnormal{condexp}(N')}).
    \end{align*}
    Thus, for all $j$, $\cfgnode(M_j) = N_k$, we have $\tilde{p}_j \in \restricedfs{A_k}{\pdfvalues}$ which implies $p_k \in \restricedfs{A_k}{\pdfvalues}.$
    Again, we know that for the end-state $\sigma_l(\probvar) = \evalf{\text{END}}{\probvar}(\dict) = p_\unrolledcfg(\dict)$ and since $b_j(\dict) = \true$ if and only if $M_j$ is in the execution sequence for $\dict$, we have $p_\unrolledcfg(\dict) = \prod_{k=1}^K p_k(\dict)$.
\end{proof}

\newpage

\section{Towards a Measure-theoretic Interpretation of the Operational Semantics}\label{appendix:measure-theory}

In this section, we will briefly discuss how one may construct a measure space on traces such that the function $\dict \mapsto p_S(\dict)$ is indeed a Radon-Nikodym derivative of a measure on this space.
This construction is related to denotational semantics based on trace types~\cite{lew2019tracetypes,lew2023stochasticprobs}. 
We first define a measure space for each value type:
\begin{align*}
    \tau &:\coloneqq \{\none\}, \Ints, \Reals, \Reals^2, \dots \\
    \measurespace_\tau = (M_\tau, \sigmaalgebra_\tau, \refmeasure_\tau) &:\coloneqq (\{\none\},\powerset(\{\none\}), \delta_\none), \; (\Ints,\powerset(\Ints),\#), \; (\Reals,\mathcal{B},\lambda), \; (\Reals^2,\mathcal{B}^2,\lambda^2),\; \dots
\end{align*}
These measure spaces model the support of common distributions, but could also be extended by building sum and product spaces.
They are comprised of the typical Dirac measure, power sets, counting measure, Borel sets, and Lebesgue measures.

Each trace $\dict$ can be assigned an unique type function $\tau_\dict\colon\strings\to\types$ such that for all $\alpha\in\strings$ we have $\dict(\alpha) \in \tau_\dict(\alpha)$.
The types are assumed to be disjoint (e.g. $\Ints \cap \Reals = \emptyset$).
Let
\[\mathcal{S}_{ < \infty} = \{s\colon\strings\to\types \colon |\{\alpha\in\strings\colon s(\alpha) \neq \{\none\}\}| < \infty\}\]
be the set of all type functions that correspond to traces with a finite number of non-$\none$ values.
We enumerate the addresses, $\strings = \{\alpha_i\colon i \in \Nats\}$, and define a measureable space for each type function, $s\colon\strings\to\types$,
$\mathscr{M}_s = (M_s, \Sigma_s) = \bigotimes_{i \in \Nats\colon s(\alpha_i) \neq \{\none\} }\mathcal{M}_{s(\alpha_i)}$.
If $s \in \mathcal{S}_{ < \infty}$, then $\mathscr{M}_s$ is a finite dimensional product space with product measure $\nu_s$.
For the set of traces $\dicts$, we define the $\sigma$-algebra $\Sigma_\dicts$ with
$A \in \Sigma_\dicts \Leftrightarrow \forall s\colon\strings\to\types \colon \pi_s(\{\dict \in A\colon \tau_\dict = s\}) \in \Sigma_s,$ where $\pi_s\colon\dicts\to M_s$ is the canonical projection.
The reference measure on $(\dicts, \Sigma_\dicts)$ is \[\nu(A) = \sum_{s \in \mathcal{S}_{ < \infty}} \nu_s(\pi_s(\{\dict \in A\colon \tau_\dict = s\}))\] which effectively is a measure of traces with a finite number of non-$\none$ values.


If a primitive distribution $f$ with $\text{support}(f) \subseteq \tau_f$ and arguments $v_1\dots,v_n$ corresponds to the measure $\mu_{f, v_1,\dots, v_n}$ with Radon-Nikodym derivative ${\mathtt{d}\mu_{f, v_1,\dots, v_n}}/{\mathtt{d}\nu_{\tau_f}}$, we define
\begin{displaymath}
    \pdf_f(v;v_1,\dots, v_n) = 
    \begin{cases}
        (\mathtt{d}\mu_{f, v_1,\dots, v_n}/\mathtt{d}\nu_{\tau_f})(v) & {\text{if } v \in \tau_f}\\
        0 & \text{otherwise}.
    \end{cases}
\end{displaymath}
The measure-theoretic interpretation of a program is the measure $\mu_S(A) = \int_{\dicts_\textnormal{valid} \cap A} p_S \, \mathtt{d}\nu$, where
\[\dicts_\textnormal{valid} \coloneqq \{\dict\colon p_S(\dict) \neq \undefval \land \forall \alpha\colon \dict(\alpha) \neq \none \Rightarrow p_S(\dict[\alpha \mapsto \none]) = \undefval\}\]
is the set of all traces with well-defined density such that changing a value at any address to $\none$ leads to an undefined density.
That is, for $\dict \in \dicts_\textnormal{valid}$ all addresses that are not used during the execution of the program are set to $\none$ and all address that are used are set to non-$\none$ values.
Such a trace necessarily has a finite number of non-$\none$ values as the density for non-termination is $\undefval$ which is the only possibility for executing an infinite number of sample statements.
This condition is important, because the countably infinite sum in definition of $\nu$ has better convergence properties on $\dicts_\textnormal{valid}$.
Finally, $p_S$ the Radon-Nikodym derivative of $\mu_S$ with respect to $\nu$ on $\dicts_\textnormal{valid}$. 

We emphasize that this interpretation demands a careful analysis of the measurability of $\dicts_\textnormal{valid}$ and $p_S$ which is beyond the scope of this work.
This requires additional assumptions about the program, at the very least that the primitives $g$ are measurable.
Such measurability analysis may be done along the lines of \citet{borgstrom2016lambda} and \citet{lee2019towards}.

\pagebreak

Note that this construction can handle programs where random variables are of mixed-type like the variable $\mathtt{y}$ in the example below:
\begin{lstlisting}[style=Python, numbers=none]
x = sample("x", Bernoulli(p))
if x == 1 then
    y = sample("y", Bernoulli(0.5))
else
    y = sample("y", Normal(0.,1.))
\end{lstlisting}
We list the density according to the semantics for different example trace of different type:
\begin{align*}
    p(\{x\mapsto 1\colon\mathbb{Z}, y\mapsto 1\colon\mathbb{Z}\}) &= \pdf_\textnormal{Bernoulli}(1; 0.5) \cdot \pdf_\textnormal{Bernoulli}(1; 0.5)  = 0.5\cdot 0.5\\
    p(\{x\mapsto 1\colon\mathbb{Z}, y\mapsto 1.0\colon\mathbb{R}\}) &= \pdf_\textnormal{Bernoulli}(1; 0.5) \cdot 0 = 0 \\
    p(\{x\mapsto 0\colon\mathbb{Z}, y\mapsto 1.0\colon\mathbb{R}\}) &= \pdf_\textnormal{Bernoulli}(0; 0.5) \cdot \pdf_\textnormal{Normal}(1.0; 0.0, 1.0) \approx 0.5 \cdot 0.24197\\
    p(\{x\mapsto 0\colon\mathbb{Z}, y\mapsto 1\colon\mathbb{Z}\}) &= \pdf_\textnormal{Bernoulli}(0; 0.5) \cdot 0 = 0
\end{align*}

\subsection{Remark on Conditioning}
We have defined density semantics such that $p_S$ denotes the joint probability distribution over all variables in the model in the sense described above.
In our construction, we do not differentiate between latent and observed variables.
To incorporate observed data, we can define the trace $\obsdict$, where addresses are mapped to the observed data points.
All addresses that do not correspond to observed data map to $\none$.
Loosely, the function $\dict \mapsto p_S(\dict \oplus \obsdict)$ then corresponds to the \emph{unnormalised} posterior density, 
where we merge the trace $\dict$, encapsulating latent variables, with the observed trace $\obsdict$ like below
\begin{displaymath}
    (A \oplus B)(\alpha) = \begin{cases}
        B(\alpha) & \text{if} \, B(\alpha)  \neq \none,\\
        A(\alpha) & \text{otherwise.}
    \end{cases}
\end{displaymath}
This is only a loose correspondence, because it is not clear whether this function can be normalized by integrating over all observed variables.
It is beyond the scope of this work to precisely define conditioning semantics.
For work discussing similar issues, see \citet{lee2019towards}.

\newpage

\section{Proofs of Section 5}
{\revcolor

\subsection{Sliced CFG Construction}

First, we formally define the CFG slicing with respect to the factorisation of \Cref{theorem:loopy-factor}.
Here, we use $\staticprovnode$ instead of $\staticprov$ to make the dependencies between CFG nodes explicit (see \Cref{sec:staticprov}).

\begin{definition}
    Let $\cfg=(N_\textnormal{start},N_\textnormal{end},\mathcal{N},\mathcal{E})$ be a CFG with $K$ sample nodes, which we enumerate with $N_j = \textnormal{Assign}(x_j = \kw{sample}(E_0^j, f^j(E_1^j,\dots,E_{n_j}^j)))$ for $j = 1,\dots,K$.\newline    
    For each $k \in \{1,\dots,K\}$, we define the set of indices of sample nodes that the factor $p_k$ depends on:
    \[\mathcal{J}_k \coloneqq \{j \in \{1,\dots,K\}\colon  N_k \in \medcup_{i=0}^{n_j}\staticprovnode(N_j,E^j_i) \cup \medcup_{N'\in \CP(N_j)} \staticprovnode(N',\textnormal{condexp}(N'))\}.\]
    Next, we define the set of nodes that lie on path from $N_k$ to some of its dependent node $N_j$:
    \[\mathcal{N}'_k \coloneqq \{N \in \mathcal{N}\colon \exists j\in\mathcal{J}_k\colon\exists\text{ path } \nu=(N_k,\dots,N_j)\colon N \in \nu \text{ and } N_k \text{ only appears at the start of }\nu\}.\]
    We transform all sample nodes that are in $\mathcal{N}_k'$, but do not depend on $N_k$ to a \emph{read} statement, which in contrast to sample statements do not contribute to the density:
    $\mathcal{N}_k \coloneqq \{\psi(N)\colon N \in \mathcal{N}_k'\}$, where \[\psi(N)=
    \begin{cases}
    \textnormal{Assign}(x_j = \kw{read}(E_0^j,f(E_1^j,\dots,E_{n_j}^j)) & \text{if}\,\, \nexists j\in\mathcal{J}_k\cup\{k\}\colon N_j = N\\
    N & \text{otherwise}
    \end{cases}\]
    The semantic of read statements is given below. Note that we have to keep all $E_i$ in the statement to check their values against $\none$ like in \Cref{eq:sample-semantics}. We omit this technicality in the main text.
    \begin{equation}\label{eq:read-semantics}
    \frac{\forall i\colon V_i = \progstate(E_i)\land V_i\neq\none \quad V_0 \in \strings \quad V = \dict(V_0) \neq \none}{(\sigma, x = \kw{read}(E_0,f(E_1,\dots,E_{n})))\opsemrel \progstate[x \mapsto V]}
    \end{equation}
    We restrict the set of edges to the subset of nodes $\mathcal{N}_k$:
    \[\mathcal{E}_k' \coloneqq \{(\psi(e_1),\psi(e_2))\colon (e_1,e_2) \in \mathcal{E}\land e_1 \in \mathcal{N}_k' \land e_2 \in \mathcal{N}_k'\}.\]
    We connect the new start node $N^k_\textnormal{start}$ to $N_k$ and the new end node
    $N^k_\textnormal{end}$ to all $N_j$, $j\in\mathcal{J}_k$, that do not have a successor in $\mathcal{E}_k'$:
    \[\mathcal{E}_k = \mathcal{E}_k'\cup \{(N^k_\textnormal{start}, N_k)\} \cup \{(N_j,N^k_\textnormal{end})\colon j\in\mathcal{J}_k \land \nexists N \in \mathcal{N}_k\colon (N_j,N)\in \mathcal{E}_k'\}.\]
    Finally, define the \emph{CFG sliced with respect to $N_k$} as
    \[\cfg_k \coloneqq (N^k_\textnormal{start}, N^k_\textnormal{end},\mathcal{N}_k, \mathcal{E}_k).\]
\end{definition}
Note that in this construction, we allow that the end node $N^k_\textnormal{end}$ has multiple predecessors.
This does not change the CFG semantics.
Next, we rewrite \Cref{lemma:staticprov-unrolled-cfg} and \Cref{lemma-staticprov} in terms of $\staticprovnode$ and prove two technical lemmas needed to establish the correctness result in the next section.

\begin{lemma}\label{lemma-staticprovnode}
Let $M$ be a CFG node and $x$ a variable. $\staticprovnode$ has following properties:
\begin{itemize}
    \item If $M' \in \RD(M,x)$ is an assignment node, $M' = \textnormal{Assign}(x = E)$, then
    \begin{align}\label{eq:staticprovnode-assign}
        \staticprovnode(M',E) = \medcup_{y \in \textnormal{vars}(E)} \staticprovnode(M',y) \subseteq \staticprovnode(M,x).
    \end{align}
    \item if $M' \in \RD(M,x)$ is a sample node, $M' = \textnormal{Assign}(x = \kw{sample}(E_0,\dots))$, then
    \begin{align}\label{eq:staticprovnode-sample}
        \staticprovnode(M',E_0) \subseteq \staticprovnode(M,x)\quad \text{and}\quad M' \in \staticprovnode(M,x).
    \end{align}
    \item For all $M' \in \RD(M,x)$ it holds that for all branch parents $ N_\text{bp} = \textnormal{Branch}(E_{M_\text{bp}}) \in \CP(M')$ 
    \begin{align}\label{eq:staticprovnode-control}
        \staticprovnode(M_\text{bp},E_{N_\text{bp}}) \subseteq \staticprovnode(M,x).
    \end{align}
    \item We have following relationship between $\staticprovnode$ on CFG $\cfg$ and on unrolled CFG $\unrolledcfg$:
    \begin{equation}\label{eq:staticprovnode-icfg}
        \staticprovnode(M,x) \subseteq \staticprovnode(\cfgnode(M),x).
    \end{equation}
\end{itemize}
\end{lemma}

\begin{lemma}\label{lemma:alpha-sample-node-in-both}
    Let $\dict_1$ and $\dict_2$ be traces that only differ at address $\alpha$, i.e. $\dict_1(\alpha) \neq \dict_2(\alpha)$ and $\forall \beta \neq \alpha\colon \dict_1(\beta) = \dict_2(\beta).$
    Further, let $M_\alpha$ bet the only sample node in the execution sequences of $\dict_1$ and $\dict_2$ whose address expression evaluates to $\alpha$.
    If CFG node $M$ appears in \emph{both} execution sequences and variable $x$,
    \begin{align*}
        (\progstate_0, \text{START}) \cfgreloverset{\dict_1} \cdots  \cfgreloverset{\dict_1} (\progstate_1, M) \cfgreloverset{\dict_1} \cdots \cfgreloverset{\dict_1} (\progstate_{l_1}, \text{END}),\\
        (\progstate_0, \text{START}) \cfgreloverset{\dict_2} \cdots  \cfgreloverset{\dict_2} (\progstate_2, M) \cfgreloverset{\dict_2} \cdots \cfgreloverset{\dict_2} (\progstate_{l_2}, \text{END}),
    \end{align*}
then the inequality $\progstate_1(x) = \evalf{M}{x}(\dict_1) \neq \evalf{M}{x}(\dict_2) = \progstate_2(x)$ implies $M_\alpha \in \staticprovnode(M,x)$.
\newline
Or equivalently, $M_\alpha \notin \staticprovnode(M,x)$ implies $ \progstate_1(x) = \evalf{M}{x}(\dict_1) = \evalf{M}{x}(\dict_2) =  \progstate_2(x)$.
\end{lemma}
\begin{proof}
    By induction on the position in the execution sequence.
    The claim is trivially true for the start node.
    Let $M$ be a node in both execution sequences with $\evalf{M}{x}(\dict_1) \neq \evalf{M}{x}(\dict_2)$.
    The induction assumption states that for all nodes $M'$ that appear sooner than $M$ in any of the execution sequences, the claim holds.
    
    Let $M'_1 \in \RD(M,x) \neq \emptyset$ be the assign node that last writes $x$ in the execution sequence of $\dict_1$ before evaluating $M$.
    Similarly, let $M'_2 \in \RD(M,x) \neq \emptyset$ be the corresponding assign node in the execution sequences of $\dict_2$.
    \begin{itemize}
        \item Case 1. $M'_1 \neq M'_2$: By definition of reaching definitions, both execution sequences must contain a different path from start node to $M$.
        This can only happen if there exists one branch parent  $B \in \CP(M'_1) \cap \CP(M'_2)$ in both execution sequences that evaluates differently, $\evalf{B}{E}(\dict_1)\neq\evalf{B}{E}(\dict_2)$ for $E = \textnormal{condexp}(B)$.
        By induction assumption, $M_\alpha \in \staticprovnode(B,E)$ and \Cref{eq:staticprovnode-control} implies $M_\alpha \in \staticprovnode(M,x)$.
        \item Case 2. $M' \coloneqq M'_1 = M'_2 = \textnormal{Assign}(x=E)$: It holds that $\evalf{M'}{E}(\dict_1) = \evalf{M}{x}(\dict_1) \neq \evalf{M}{x}(\dict_2) = \evalf{M'}{E}(\dict_2)$.
        By induction assumption, we have  $M_\alpha \in \staticprovnode(M',E)$ and \Cref{eq:staticprovnode-assign} implies $M_\alpha \in \staticprovnode(M,x)$.
        \item Case 3. $M' \coloneqq M'_1 = M'_2 = \textnormal{Assign}(x=\kw{sample}(E_0,\dots))$: It holds that $\dict_1(\evalf{M'}{E_0}(\dict_1)) = \evalf{M}{x}(\dict_1) \neq \evalf{M}{x}(\dict_2) = \dict_2(\evalf{M'}{E_0}(\dict_2))$.
        If $\evalf{M'}{E_0}(\dict_1) = \alpha$ or $\evalf{M'}{E_0}(\dict_2) = \alpha$, then $M' = M_\alpha$, since $M_\alpha$ is the only sample node whose address expression evaluates to $\alpha$ and we have $M_\alpha \in \staticprovnode(M,x)$ by \Cref{eq:staticprovnode-sample}.
        Otherwise, since the traces only differ at $\alpha$, the inequality can only hold if $\evalf{M'}{E_0}(\dict_1) \neq \evalf{M'}{E_0}(\dict_2)$.
        But then as in Case 2, by induction assumption $M_\alpha \in \staticprovnode(M',E_0)$ and by \Cref{eq:staticprovnode-sample} we have $M_\alpha \in \staticprovnode(M,x)$.
    \end{itemize}
    Lastly, if there are no such $M'_1$ and $M'_2$ in both sequences, then $\evalf{M}{x}(\dict_1) = \none = \evalf{M}{x}(\dict_2)$, a contradiction.
    If there is only $M'_1$ or only $M'_2$, then the claim follows as in Case 1.
\end{proof}

\begin{lemma}\label{lemma:alpha-sample-node-in-one}
    Let $\dict_1$ and $\dict_2$ be traces that only differ at address $\alpha$, i.e. $\dict_1(\alpha) \neq \dict_2(\alpha)$ and $\forall \beta \neq \alpha\colon \dict_1(\beta) = \dict_2(\beta).$
    Further, let  $M_\alpha$ be the only sample node in the execution sequences whose address expression evaluates to $\alpha$.
    If CFG node $M$ is in \emph{one but not both} of the execution sequences of $\dict_1$ and $\dict_2$, then there exists a branch parent $B \in \CP(M)$ such that $M_\alpha \in \staticprovnode(B,\textnormal{condexp}(B))$.
\end{lemma}
\begin{proof}
    As in the proof of \Cref{theorem:loopy-factor}, $b(\dict) = \bigwedge_{B \in \CP(M)} t_{B}(\evalf{B}{ \textnormal{condexp}(B)}(\dict))$ is $\true$ if $M$ is in the execution sequence of $\dict$ else $\false$.
    Since $M$ is not in both execution sequences we have $b(\dict_1) \neq b(\dict_2)$.
    Thus, one of the branch parents has to evaluate differently, i.e. there is a branch parent $B \in \CP(M)$ in both execution sequences such that $\evalf{B}{ \textnormal{condexp}(B)}(\dict_1) \neq \evalf{B}{ \textnormal{condexp}(B)}(\dict_2)$.
    
    Since, $\dict_1$ and $\dict_2$ differ only at $\alpha$ and $M_\alpha$ is the only sample node in the execution sequences whose address expression evaluates to $\alpha$, by \Cref{lemma:alpha-sample-node-in-both}, $M_\alpha \in \staticprovnode(B,\textnormal{condexp}(B))$.
\end{proof}
\pagebreak

\subsection{Proof of \Cref{theorem:slicing-correctness} - Correctness of Slicing}
Now, we are ready to prove the main correctness result of the sliced CFG construction. \medskip

\begin{statement}
    Let $\dict_1$ and $\dict_2$ be traces that only differ at address $\alpha$, i.e. $\dict_1(\alpha) \neq \dict_2(\alpha)$ and $\forall \beta \neq \alpha\colon \dict_1(\beta) = \dict_2(\beta).$
    Let the execution sequences in the unrolled CFG $\unrolledcfg$ be equal up to sample node $M_\alpha = \textnormal{Assign}(x = \kw{sample}(E_0, f(E_1,\dots,E_{n})))$ and state $\sigma$, such that $\alpha = \sigma(E_0)$:
\begin{align}
    (\progstate_0, \text{START}) \cfgreloverset{\dict_1} \cdots  \cfgreloverset{\dict_1} (\progstate, M_\alpha) \cfgreloverset{\dict_1} \cdots \cfgreloverset{\dict_1} (\progstate_{l_1}, \text{END}), \label{eq:slice-correctness-exec-seq-1}\\
    (\progstate_0, \text{START}) \cfgreloverset{\dict_2} \cdots  \cfgreloverset{\dict_2} (\progstate, M_\alpha) \cfgreloverset{\dict_2} \cdots \cfgreloverset{\dict_2} (\progstate_{l_2}, \text{END}).\label{eq:slice-correctness-exec-seq-2}
\end{align}
If $M_\alpha$ is the only sample node in the execution sequences whose address expression evaluates to $\alpha$, then for the sliced CFG $\cfg_k$ of sample node $N_k = \cfgnode(M_\alpha)$, it holds that
\[\log p_G(\dict_1) - \log  p_G(\dict_2) = \log \bar{p}_{\cfg_k}(\dict_1, \sigma) - \log  \bar{p}_{\cfg_k}(\dict_2, \sigma).\]
\end{statement}

\begin{proof}

Let $M_1'$ and $M_2'$ be the last sample nodes in the sequences, such that $\cfgnode(M_1') \in \cfg_k$ and $\cfgnode(M_2') \in \cfg_k$.
By definition of $\cfg_k$, the execution sequences for $\dict_1$ and $\dict_2$ in the unrolled version $\unrolledcfg_k$ of $\cfg_k$ become
\[(\progstate, \text{START})  \underset{\unrolledcfg_k}{\cfgreloverset{\dict_1}} (\progstate, M_\alpha) \underset{\unrolledcfg_k}{\cfgreloverset{\dict_1}} \cdots \underset{\unrolledcfg_k}{\cfgreloverset{\dict_1}} (\sigma'_1,M'_1) \underset{\unrolledcfg_k}{\cfgreloverset{\dict_1}} \cdots\underset{\unrolledcfg_k}{\cfgreloverset{\dict_1}} (\progstate_1'', \text{END}),\]
\[(\progstate, \text{START}) \underset{\unrolledcfg_k}{\cfgreloverset{\dict_2}} (\progstate, M_\alpha) \underset{\unrolledcfg_k}{\cfgreloverset{\dict_2}} \cdots \underset{\unrolledcfg_k}{\cfgreloverset{\dict_2}} (\progstate'_2,M'_2) \underset{\unrolledcfg_k}{\cfgreloverset{\dict_2}} \cdots \underset{\unrolledcfg_k}{\cfgreloverset{\dict_2}} (\progstate_2'', \text{END}),\]
for some states $\progstate'_1$, $\progstate'_2, \progstate''_1$, $\progstate''_2$.
Note that technically we also have $(\progstate'_i,M'_i) \underset{\unrolledcfg_k}{\cfgreloverset{\dict_i}} (\progstate_i', \text{END})$.

To prove the statement, we have to show that the contribution to the density of sample nodes $M$ that come after $M_1'$ and $M_2'$ in the original execution sequences \eqref{eq:slice-correctness-exec-seq-1} or \eqref{eq:slice-correctness-exec-seq-2} and sample nodes $M$ that come between $M_\alpha$ and $M_1'$ or $M_2'$ but were transformed to $\kw{read}$ statements in $G_k$, see \eqref{eq:read-semantics}, is equal.

By \Cref{theorem:loopy-factor}, the contribution to the density of sample node $M$ with $N_j = \cfgnode(M)$, i.e. $M = \textnormal{Assign}(x_j = \kw{sample}(E_0^j, f^j(E_1^j,\dots,E_{n_j}^j)))$, has provenance set
\begin{align*}
    \tilde{A}_j &\coloneqq \textnormal{addresses}(M) \cup  \medcup_{i = 0}^{n_j} \staticprov(M, E_i^j) \cup \medcup_{M_\textnormal{bp} \in \CP(M)}\staticprov(M_\textnormal{bp},\textnormal{condexp}(M_\textnormal{bp}))\\
    &\subseteq \textnormal{addresses}(N_j) \cup  \medcup_{i = 0}^{n_j} \staticprov(N_j, E_i^j) \cup \medcup_{N_\textnormal{bp} \in \CP(N_j)}\staticprov(N_\textnormal{bp},\textnormal{condexp}(N_\textnormal{bp})) \eqqcolon A_j,
\end{align*}
or equivalently expressed in terms of sample nodes
\begin{align*}
    \tilde{\mathcal{A}}_j &\coloneqq \{M\} \cup  \medcup_{i = 0}^{n_j} \staticprovnode(M, E_i^j) \cup \medcup_{M_\textnormal{bp} \in \CP(M)}\staticprovnode(M_\textnormal{bp},\textnormal{condexp}(M_\textnormal{bp}))
\end{align*}
such that by \Cref{lemma-staticprovnode}
\begin{align*}
    \{\cfgnode(M')\colon &M' \in \tilde{\mathcal{A}}_j\} \subseteq\\ &\{N_j\} \cup  \medcup_{i = 0}^{n_j} \staticprovnode(N_j, E_i^j) \cup \medcup_{N_\textnormal{bp} \in \CP(N_j)}\staticprovnode(N_\textnormal{bp},\textnormal{condexp}(N_\textnormal{bp})) \eqqcolon \mathcal{A}_j
\end{align*}

If $M$ comes after $M'_1$ or $M'_2$ in \eqref{eq:slice-correctness-exec-seq-1} or \eqref{eq:slice-correctness-exec-seq-2}, then $\cfgnode(M) = N_j \notin G_k$, and by definition of $\cfg_k$, we have $\cfgnode(M_\alpha) = N_k \notin\mathcal{A}_j$ and thus also $M_\alpha \notin \tilde{\mathcal{A}}_j$.
If $M$ comes between $M_\alpha$ and $M'_1$ or $M'_2$, but was transformed to a $\kw{read}$ statement in $G_k$, then by definition of $\cfg_k$, also $N_k \notin\mathcal{A}_j$ and $M_\alpha \notin \tilde{\mathcal{A}}_j$.

In both cases, by \Cref{lemma:alpha-sample-node-in-one}, $M$ has to be in both original execution sequences, since otherwise we would have the contradiction $M_\alpha \in \medcup_{B \in \CP(M)}\staticprovnode(B,\textnormal{condexp}(B))$.
Let $(\progstate_1, M)$ denote the appearance in~\eqref{eq:slice-correctness-exec-seq-1} and $(\progstate_2, M)$ the appearance in~\eqref{eq:slice-correctness-exec-seq-2}.
Since $M_\alpha \notin \medcup_{i = 0}^{n_j} \staticprovnode(M, E_i^j)$, by \Cref{lemma:alpha-sample-node-in-both}, we have $\sigma_1(y) =\sigma_2(y)$ for all variables $y \in \medcup_{i = 0}^{n_j} \textnormal{vars}(E_i^j)$.
Since it also holds that $\alpha \neq \sigma_1(E_0) = \sigma_2(E_0)$, we have $\dict_1(\sigma_1(E_0)) = \dict_2(\sigma_2(E_0))$, which means that $M$ contributes the same value to the density in both execution sequences, compare \Cref{eq:sample-semantics-cfg}.
\end{proof}
\newpage

\subsection{Proof of \Cref{theorem:slicing-correctness-corollary} - Correctness of Slicing (cont.)}
\begin{statement}
Assume that for each trace $\dict$ there is at most one sample node $M_\alpha$ in its execution sequences such that the address expression of $M_\alpha$ evaluates to $\alpha$.
For initial program state $\progstate_0$ as in \Cref{def:prog-semantics}, define
    \begin{equation*}
        p_{\alpha}(\dict) \coloneqq \begin{cases}
            \bar{p}_{G_k}(\dict, \progstate)& \text{ if } \exists\,\progstate\in\progstates, M_\alpha\colon (\progstate_0, \text{START}) \cfgreloverset{\dict} \cdots  \cfgreloverset{\dict} (\progstate, M_\alpha) \land \cfgnode(M_\alpha) = N_k\\
            1 & \text{ otherwise.}
        \end{cases}
    \end{equation*}
    Then, $\Delta_\alpha(\dict) \coloneqq \log p_G(\dict) - \log p_{\alpha}(\dict)$ is independent of address $\alpha$, $\Delta_\alpha \in \restricedfs{\strings \setminus \{\alpha\}}{\pdfvalues}.$
\end{statement}
\begin{proof}
    By definition $\Delta_\alpha \in \restricedfs{\strings \setminus \{\alpha\}}{\pdfvalues}$, if for two traces $\dict_1$ and $\dict_2$ it holds that
    \[\forall \beta \in \strings \setminus \{\alpha\}\colon \dict_1(\beta) = \dict_2(\beta) \implies \Delta_\alpha(\dict_1) = \Delta_\alpha(\dict_2).\]
    Let $\dict_1$ and $\dict_2$ be traces such that $\dict_1(\alpha) \neq \dict_2(\alpha)$ and $\forall \beta \neq \alpha\colon \dict_1(\beta) = \dict_2(\beta).$
    The execution sequences of $\dict_1$ and $\dict_2$ are equal up to $(\sigma, M_\alpha)$, where $M_\alpha$ is the first sample node  whose address expression evaluates to $\alpha$.
    If there is no such sample node, then the execution sequences are completely identical.
    In this case, $\log p_G(\dict_1) = \log p_G(\dict_2)$  and $\log p_{\alpha} (\dict_1) = \log p_{\alpha} (\dict_2) = 0$.
    If there is such a sample node $M_\alpha$, by assumption it is the only one whose address expression evaluations to $\alpha$ and by \Cref{theorem:slicing-correctness}, 
    $\log p_G(\dict_1) - \log  p_G(\dict_2) = \log \bar{p}_{\cfg_k}(\dict_1, \sigma) - \log  \bar{p}_{\cfg_k}(\dict_2, \sigma)$.
    After rearranging, we have
    \begin{align*}
     \Delta_\alpha(\dict_1) - \Delta_\alpha(\dict_2) &=
     (\log p_G(\dict_1) - \log p_\alpha(\dict_1)) - ( \log p_G(\dict_2) - \log p_\alpha(\dict_2))\\
      &=(\log p_G(\dict_1) - \log \bar{p}_{G_k}(\dict_1,\sigma)) - ( \log p_G(\dict_2) - \log \bar{p}_{G_k}(\dict_2,\sigma)) \\
     &= (\log p_G(\dict_1) - \log p_G(\dict_2)) - (\log \bar{p}_{G_k}(\dict_1,\sigma)  - \log \bar{p}_{G_k}(\dict_2,\sigma)) = 0
    \end{align*}
\end{proof}

\subsection{Proof of \Cref{theorem:lmh-correctness} - Correctness of LMH Optimisation}
To prove correctness of the LMH optimisation, we require the notion of keys and minimal traces.
\begin{definition}
    The set of keys of a trace $\dict$ is the set of addresses that map to an non-$\none$ value
    \[\textnormal{keys}(\dict) \coloneqq \{\alpha \in \strings\colon \dict(\alpha) \neq \none\}.\]
    
    A trace $\dict$ is \emph{minimal} with respect to $\cfg$ if
    \[p_\cfg(\dict) \neq \undefval \quad\text{and}\quad \forall a \in \strings\colon p_\cfg(\dict[\alpha \mapsto \none]) = \undefval.\]
\end{definition}
Note that the density is undefined if a trace is read at an address that maps to $\none$, see \eqref{eq:sample-semantics-cfg}.
Since a trace is minimal if setting a single additional address to $\none$ results in undefined density, the minimality condition captures "key errors".
For a minimal trace $\textnormal{keys}(\dict)$ is the minimal set of addresses that need to be set to non-$\none$ values such that the density is well-defined, i.e. the trace contains no redundant information.
We can immediately establish following lemma.
Here, and in the following, we often omit the dependency on $G$ and write $p(\dict) = p_\cfg(\dict)$ to declutter notation.

\begin{lemma}\label{lemma:minimal-same-exec-seq}
    Let $\dict_1$ be a minimal trace and $\dict_2$ a trace such that $\forall \alpha \in \textnormal{keys}(\dict_1)\colon \dict_1(\alpha) = \dict_2(\alpha)$.
    Then, the execution sequences of $\dict_1$ and $\dict_2$ as in \cref{theorem-evalf-loopy} are the same and $p(\dict_1) = p(\dict_2).$
\end{lemma}
\begin{proof}
    By induction on the execution sequence.
    Assume the execution sequences are equal up to sample node $N$ with address expression $E_0$ and state $\sigma$.
    Since $p(\dict_1)\neq\undefval$, we must have $\dict_1(\sigma(E_0)) \neq \none$ in \eqref{eq:sample-semantics-cfg}.
    Thus, we have $\sigma(E_0) \in \textnormal{keys}(\dict_1)$ and  $\dict_2(\sigma(E_0)) = \dict_1(\sigma(E_0)) \neq \none$ and the execution sequence continue to be equal.
    All other transition rules follow immediately.
\end{proof}
\pagebreak

Next, we formally define the semantics of the LMH kernel \cite{wingate2011lightweightmh} in our setting. 

\begin{definition}\label{def:lmh-semantics}
Let $Q(\beta)$ be a proposal distribution for each address $\beta$ and $\texttt{vs}(\beta) \sim Q(\beta)$ a sample from $Q(\beta)$.
For a minimal trace $\dict$, the semantics of the LMH kernel \cite{wingate2011lightweightmh} that resamples at address $\alpha$ is given below, where we introduce new reserved variables $\boldsymbol{tr}$ and $\boldsymbol{q}$.
\begin{equation*}
    \frac{
    \begin{aligned}
    &\forall i \colon \progstate(E_i) = V_i \land  V_i \neq \none \quad V_0 \in \strings
    \\
    &V =
    \begin{cases}
    \texttt{vs}(V_0) & \text{if}\,\, V_0 = \alpha \lor V_0 \notin \textnormal{keys}(\dict)\\
    \dict(V_0) & \text{if}\,\, V_0 \neq \alpha \land V_0 \in \textnormal{keys}(\dict)
    \end{cases}
    \quad
     q = \begin{cases}
    \pdf_{Q(V_0)}(V) & \text{if}\,\, V_0 = \alpha \lor  V_0 \notin \textnormal{keys}(\dict) \\
    1 & \text{otherwise}
    \end{cases}
    \end{aligned}}
    {
    \begin{aligned}
    (\progstate, x = \kw{sample}(E_0, f(E_1,\dots,E_n))) \Downarrow^{\dict, \texttt{vs}}_{\texttt{LMH}(\alpha)} \progstate[&x \mapsto V, \probvar \mapsto \progstate(\probvar) \cdot  \pdf_f(V; V_1,\dots,V_n), \\
    &\boldsymbol{tr} \mapsto \progstate(\boldsymbol{tr})[V_0 \mapsto V], \boldsymbol{q} \mapsto \progstate(\boldsymbol{q}) \cdot q]
    \end{aligned}
    }
\end{equation*}
Let there exists a state $\progstate \in \progstates$ such that $(\progstate_0, S) \Downarrow^{\dict, \texttt{vs}}_{\texttt{LMH}(\alpha)} \sigma$ for program $S$ and initial program state $\sigma_0 = (x_1\mapsto \none,\dots,x_n\mapsto \none, \probvar \mapsto 1, \boldsymbol{tr} \mapsto \texttt{nulltr}, \boldsymbol{q} \mapsto 1)$, where $\texttt{nulltr}$ is a trace that maps all addresses to $\none$.
Then, the new trace $\dict' \coloneqq \progstate(\boldsymbol{tr})$ is called \emph{valid LMH proposal} and $\dict'$ is a minimal trace that agrees with $\dict$ on $\textnormal{keys}(\dict)\cap\textnormal{keys}(\dict') \setminus \{\alpha\}$.
Furthermore, there exists a state $\progstate'$ such that $(\sigma_0, S) \Downarrow^{\dict',\dict}_{\texttt{LMH}(\alpha)} \sigma'$ and $\dict = \sigma'(\boldsymbol{tr})$.

With $p(\dict') = \progstate(\probvar)$, $Q_\alpha(\dict'|\dict)=\progstate(\boldsymbol{q})$,  $p(\dict) = \progstate'(\probvar)$, and $Q_\alpha(\dict|\dict') =\progstate'(\boldsymbol{q})$ the acceptance probability is given by \[A = \min\left(1, \frac{p(\dict') \cdot Q_\alpha(\dict|\dict')}{p(\dict) \cdot Q_\alpha(\dict'|\dict)}\right).\]

Lastly, let the equivalent semantics on the CFG be denoted by $\underset{\texttt{LMH}(\alpha)}{\cfgreloverset{\dict,\texttt{vs}}}$.
\end{definition}

\begin{definition}
We can translate the LMH semantics to the semantics of context \texttt{ForwardCtx}, denoted by $\Downarrow^{\dict, \texttt{vs}}_{\texttt{Fwd}(\alpha)}$, and of context \texttt{BackwardCtx}, denoted by $\Downarrow^{\dict}_{\texttt{Bck}(\alpha)}$.

Using log-densities, at the $\kw{visit}$ node, we have:
\begin{equation*}
    \frac{
    \forall i \colon \progstate(E_i) = V_i \land  V_i \neq \none \quad V_0 = \alpha \quad V = \texttt{vs}(\alpha) \quad q = \pdf_{Q(\alpha)}(V)
    }
    {
    \begin{aligned}
    (\progstate, x = \kw{visit}(E_0, f(E_1,\dots,E_n))) \Downarrow^{\dict, \texttt{vs}}_{\texttt{Fwd}(\alpha)} \progstate[&x \mapsto V, \boldsymbol{\log p} \mapsto \progstate(\boldsymbol{\log p}) +  \log\pdf_f(V; V_1,\dots,V_n), \\
    &\boldsymbol{tr} \mapsto \progstate(\boldsymbol{tr})[V_0 \mapsto V], \boldsymbol{\log q} \mapsto \progstate(\boldsymbol{\log q}) + \log q]
    \end{aligned}
    }
\end{equation*}
At $\kw{read}$ nodes, we just read the trace $\dict$:
\begin{equation*}
    \frac{
    \forall i\colon\progstate(E_i) = V_i \land V_i \neq \none \quad V_0 \in \textnormal{keys}(\dict) \quad V = \dict(V_0)
    }
    {
    (\progstate, x = \kw{read}(E_0, f(E_1,\dots,E_n))) \Downarrow^{\dict, \texttt{vs}}_{\texttt{Fwd}(\alpha)} \progstate[x \mapsto V,
    \boldsymbol{tr} \mapsto \progstate(\boldsymbol{tr})[V_0 \mapsto V]]
    }
\end{equation*}
The semantics of $\kw{score}$ nodes look similar to the semantics of sample statements in \cref{def:lmh-semantics} with the difference that $V_0 \neq \alpha$:
\begin{equation*}
    \frac{
    \begin{aligned}
    &\forall i \colon \progstate(E_i) = V_i \land  V_i \neq \none \quad V_0 \in \strings \quad V_0 \neq \alpha
    \\
    &V =
    \begin{cases}
    \texttt{vs}(V_0) & \text{if}\,\, V_0 \notin \textnormal{keys}(\dict)\\
    \dict(V_0) & \text{if}\,\, V_0 \in \textnormal{keys}(\dict)
    \end{cases}
    \quad
     q = \begin{cases}
    \pdf_{Q(V_0)}(V) & \text{if}\,\, V_0 \notin \textnormal{keys}(\dict) \\
    1 & \text{otherwise}
    \end{cases}
    \end{aligned}}
    {
    \begin{aligned}
    (\progstate, x = \kw{score}(E_0, f(E_1,\dots,E_n))) \Downarrow^{\dict, \texttt{vs}}_{\texttt{Fwd}(\alpha)} \progstate[&x \mapsto V, \boldsymbol{\log p} \mapsto \progstate(\boldsymbol{\log p}) +  \log\pdf_f(V; V_1,\dots,V_n), \\
    &\dict' \mapsto \progstate(\dict')[V_0 \mapsto V], \boldsymbol{\log q} \mapsto \progstate(\boldsymbol{\log q}) + \log q]
    \end{aligned}
    }
\end{equation*}
The semantics of the \texttt{BackwardCtx} are simply
\begin{equation*}
    \Downarrow^{\dict}_{\texttt{Bck}(\alpha)} \,=\, \Downarrow^{\dict', \dict}_{\texttt{Fwd}(\alpha)}.
\end{equation*}
\end{definition}

Now, we are ready to prove \Cref{theorem:lmh-correctness} which we make more precise here.\medskip

\begin{statement}
    For address $\alpha$ and minimal trace $\dict$, let $\texttt{vs}$ be values from the proposal distributions as in  \Cref{def:lmh-semantics} and let $\dict'$ be a valid LMH proposal.
    Assume that the address expression of at most one sample node in the execution sequence evaluates to $\alpha$ for all traces.
    Then, the execution sequences of $\dict$ and $\dict'$ are equal up to $(\sigma, M_\alpha)$, where $M_\alpha$ is the first and only sample node whose address expression evaluates to~$\alpha$.
    For the sliced CFG $\cfg_k$, where $\cfgnode(M_\alpha) = N_k$, we have \[\log p(\dict') - \log p(\dict) = \log \bar{p}_{G_k}(\dict', \sigma) - \log \bar{p}_{G_k}(\dict, \sigma).\]
    The quantities $\log \bar{p}_{G_k}(\dict', \sigma)$, and $\log Q_\alpha(\dict'|\dict)$ are computed by \texttt{ForwardCtx} and the quantities  $\log \bar{p}_{G_k}(\dict, \sigma)$ and $\log Q_\alpha(\dict|\dict')$ are computed by \texttt{BackwardCtx}. The new trace $\dict'$ can be constructed from $\dict$ and the output of both contexts.
\end{statement}

\begin{proof}
    Since $\dict'$ is a valid LMH proposal for $\dict$, we have that both traces are minimal and $\forall\beta \in \textnormal{keys}(\dict)\cap\textnormal{keys}(\dict') \setminus \{\alpha\}\colon \dict(\beta) = \dict'(\beta)$.

    Let $\dict_1 = \dict$ and $\dict_2 = \dict'$ and define two new traces
    \[
    \overline{\dict_i}(\beta) = \begin{cases}
    \dict_i(\alpha) & \text{ if } \beta = \alpha \\
    \dict_1(\beta)=\dict_2(\beta) & \text{ if } \beta \in (\textnormal{keys}(\dict_1) \cap \textnormal{keys}(\dict_2)) \setminus \{\alpha\} \\
    \dict_1(\beta) & \text{ if } \beta \in \textnormal{keys}(\dict_1) \setminus \textnormal{keys}(\dict_2) \\
    \dict_2(\beta) & \text{ if } \beta \in \textnormal{keys}(\dict_2) \setminus \textnormal{keys}(\dict_1) \\
    \none & \text{ otherwise.}
    \end{cases}
    \]
    It holds that $\forall \beta \in \textnormal{keys}(\dict_i)\colon \overline{\dict_i}(\beta) = \dict_i(\beta)$ (Note that $\alpha \in \textnormal{keys}(\dict_i)$).
    By \Cref{lemma:minimal-same-exec-seq}, $\overline{\dict_1}$ has the same execution sequence as $\dict$ and $\overline{\dict_2}$ the same as $\dict'$.
    
    Since $\overline{\dict_1}$ differs from $\overline{\dict_2}$ only at $\alpha$ the by \Cref{theorem:slicing-correctness} we have
    \begin{align*}
        \log p(\dict') - \log  p(\dict) &= \log p(\overline{\dict_2}) - \log  p(\overline{\dict_1})= \\
        &= \log \bar{p}_{G_k}(\overline{\dict_2},\progstate) - \log \bar{p}_{G_k}(\overline{\dict_1},\progstate) = \log \bar{p}_{G_k}(\dict',\progstate) - \log \bar{p}_{G_k}(\dict,\progstate).
    \end{align*}

    As the LMH semantics do not change the computation of $\progstate(\probvar)$, we immediately conclude that the \texttt{ForwardCtx} computes $\bar{p}_{G_k}(\dict',\progstate)$ and the \texttt{BackwardCtx} computes $\bar{p}_{G_k}(\dict,\progstate)$.
    
    \newcommand{\cfgrellmhfwd}{\underset{\texttt{LMH}(\alpha)}{\cfgreloverset{\dict,\texttt{vs}}}}
    \newcommand{\cfgrellmhbck}{\underset{\texttt{LMH}(\alpha)}{\cfgreloverset{\dict',\dict}}}
    
    As in the proof of \Cref{theorem:slicing-correctness}, but with LMH semantics, let $M'_1$ be the last sample node in the execution sequence of $\dict$ such that $\cfgnode(M'_1) \in \cfg_k$ and likewise $M'_2$ for $\dict'$.
    \[(\progstate_0, \text{START})  \cfgrellmhfwd \cdots \cfgrellmhfwd (\progstate, M_\alpha) \cfgrellmhfwd \cdots \cfgrellmhfwd (\sigma'_1,M'_1) \cfgrellmhfwd \cdots \cfgrellmhfwd (\progstate_1, \text{END}),\]
    \[(\progstate_0, \text{START}) \cfgrellmhbck \cdots \cfgrellmhbck (\progstate, M_\alpha)  \cfgrellmhbck \cdots  \cfgrellmhbck (\progstate'_2,M'_2)  \cfgrellmhbck \cdots \cfgrellmhbck (\progstate_2, \text{END}).\]
    For sample nodes that come before $M_\alpha$ or after $M'_1$ or $M'_2$ the address expressions evaluate to some $\beta \in \textnormal{keys}(\dict)\cap\textnormal{keys}(\dict') \setminus \{\alpha\}$.
    Thus, they do not update $\boldsymbol{q}$ and $\log Q_\alpha(\dict'|\dict) = \log \progstate_1(\boldsymbol{q})$ as well as $\log Q_\alpha(\dict|\dict') = \log \progstate_2(\boldsymbol{q})$ are fully computed in the sliced CFG by \texttt{ForwardCtx} and \texttt{BackwardCtx}, respectively (compare \Cref{def:lmh-semantics}).
    
    For the same reason, the difference between $\dict$ and $\dict'$ is also fully accounted for in the sliced CFG.
    The new trace $\dict'$ can be constructed by removing the address-value pairs encountered in the \texttt{BackwardCtx} from $\dict$ and adding the address-value pairs encountered in the \texttt{ForwardCtx}.

\end{proof}

\subsection{Proof of \Cref{theorem:bbvi-correctness} - Correctness of BBVI optimisation}

Before we can prove correctness of the BBVI optimisation, we need to formally define the variational distribution over traces $Q_\phi$.

\begin{definition}\label{def:variational-distribution}
    For every address $\alpha$, let there be a probability measure on some measurable space $\mathcal{M_\alpha} = (Z_\alpha, \Sigma_\alpha)$ parameterised by $\phi_\alpha$ that admits a density $q_\alpha(z_\alpha|\phi_\alpha)$ differentiable with respect to~$\phi_\alpha$ (the density is with respect to the Lebesgue measure on $\mathbb{R}^{d_\alpha}$ or counting measure on $\mathbb{Z}^{d_\alpha}$).\medskip
    
    For every finite address set $A$, define the joint density $q_A((z_\alpha)_{\alpha \in A}|(\phi_\alpha)_{\alpha\in A}) = \prod_{\alpha\in A} q_\alpha(z_\alpha|\phi_\alpha)$ on the product space $\mathcal{M}_A = (Z_A, \Sigma_A) = \bigotimes_{\alpha\in A} \mathcal{M_\alpha}.$\medskip
    
    On the set of traces $\dicts$, define the $\sigma$-algebra $\Sigma_\dicts$, such that $T \in \Sigma_\dicts$ if and only if for all finite address sets $A$, we have $\{(\dict(\alpha))_{\alpha\in A}\colon \dict \in T\} \in \Sigma_A$.
    Lastly, assume that for program density $p$ the set $\dicts_\textnormal{minimal} = \{\dict \in \dicts\colon \dict \text{ is minimal wrt. } p\}$ is measurable in $\Sigma_\dicts$.\medskip
    
    Then, the variational distribution is defined by the density \[Q_\phi(\dict) = 1_{\dicts_\textnormal{minimal}}(\dict)\sum_{\substack{A \subseteq \strings,\\|A|<\infty}} 1_{\dicts_A}(\dict) \prod_{\alpha\in A} q_\alpha(\dict(\alpha)|\phi_\alpha),\]
    where $\dicts_A\coloneqq\{\dict\in\dicts\colon \forall\alpha\in\strings\colon \dict(\alpha)\neq \none \Leftrightarrow \alpha \in A\}$.
    For all measurable functions $f$ we can express the expectation with respect to $Q_\phi$ as
    \begin{align*}
    \expectation{\dict \sim  Q_\phi}{f(\dict)}
    &= \int f(\dict)Q_\phi(\dict) d\dict\\
    &= \sum_{\substack{A \subseteq \strings,\\|A|<\infty}}\int 1_{\dicts_\textnormal{minimal}}(\iota_\dicts(z_A)) \, f(\iota_\dicts(z_A)) \prod_{\alpha\in A} q_\alpha(z_\alpha|\phi_\alpha) \,\,dz_A.
    \end{align*}
    Above $\iota_\dicts(z_A)$ for $z_A = (z_\alpha)_{\alpha \in A} \in Z_A$ is the trace that maps address $\alpha$ to $z_\alpha$ if $\alpha \in A$ and all other addresses to $\none$.
    Note that we may exchange integration and sum, since for a minimal trace only one term in the sum is non-zero.
\end{definition}

Before repeating the correctness statement in terms of \Cref{def:variational-distribution}, we formally defined the semantics of \texttt{BBVICtx}.

\begin{definition}
    The semantics of the \texttt{BBVICtx} are very similar to the standard semantics with the only difference that we store $\log q_\alpha(\dict(\alpha)| \phi_\alpha)$ at $\kw{visit}$ statements and the use of log-densities.
    
    \begin{equation*}
        \frac{
        (\progstate, x = \kw{sample}(E_0, f(E_1,\dots,E_n))) \Downarrow^{\dict} \progstate' \quad \progstate(E_0) =\alpha \quad V = \dict(\alpha)
        }{
        (\progstate, x = \kw{visit}(E_0, f(E_1,\dots,E_n))) \Downarrow^{\dict}_{\texttt{BBVI}(\alpha)} \progstate[x\mapsto V, \boldsymbol{\log p} \mapsto \log \progstate'(\boldsymbol{p}), \boldsymbol{\log q} \mapsto \log q_\alpha(V| \phi_\alpha)]
        }
    \end{equation*}

    The semantics of $\kw{read}$ statements remains unchanged:
    \begin{equation*}
        \frac{
        (\progstate, x = \kw{read}(E_0, f(E_1,\dots,E_n))) \Downarrow^{\dict} \progstate'
        }{
        (\progstate, x = \kw{read}(E_0, f(E_1,\dots,E_n))) \Downarrow^{\dict}_{\texttt{BBVI}(\alpha)} \progstate'
        }
    \end{equation*}

    Lastly, the semantics of $\kw{score}$ statements are only rewritten in terms of log-densities:
    \begin{equation*}
        \frac{
        (\progstate, x = \kw{sample}(E_0, f(E_1,\dots,E_n))) \Downarrow^{\dict} \progstate' \quad \progstate(E_0) \neq\alpha \quad V=\dict(\sigma(E_0))
        }{
        (\progstate, x = \kw{score}(E_0, f(E_1,\dots,E_n))) \Downarrow^{\dict}_{\texttt{BBVI}(\alpha)} \progstate[x\mapsto V, \boldsymbol{\log p} \mapsto \log \progstate'(\boldsymbol{p}))]
        }
    \end{equation*}
\end{definition}

\begin{statement}
    Let $Q_\phi$ be a variational distribution over traces as in \Cref{def:variational-distribution}.
    Assume that the address expression of at most one sample node in the execution sequence evaluates to $\alpha$ for all traces.
    Assume that the program density $p$ and the set $\dicts_\textnormal{minimal}$ are measurable with respect to $(\dicts,\Sigma_\dicts)$.
    With $p_\alpha$ as in \Cref{theorem:slicing-correctness-corollary}, it holds that
    \begin{align*}
    &\nabla_{\phi_\alpha}\expectation{\dict \sim  Q_\phi}{\log p(\dict) - \log  Q_\phi(\dict)} = \\
    &\quad\quad\quad\expectation{\dict \sim  Q_\phi}{\delta_{\dict(\alpha)\neq\none}\cdot\nabla_{\phi_\alpha} \log q_\alpha(\dict(\alpha)|\phi_\alpha) \cdot \big(\log p_\alpha(\dict) -\log  q_\alpha(\dict(\alpha)|\phi_\alpha)\big)}.
    \end{align*}
    The quantity $\log p_\alpha(\dict) -\log  q_\alpha(\dict(\alpha)|\phi_\alpha)$ is computed by \texttt{BBVICtx}.
\end{statement}

\begin{proof}
    First, we repeat the known identity for the BBVI estimator
    \begin{align*}
        \nabla_{\phi_\alpha}\expectation{\dict \sim  Q_\phi}{\log p(\dict) - \log  Q_\phi(\dict)} = 
        \expectation{\dict \sim  Q_\phi}{\nabla_{\phi_\alpha}\big[\log Q_\phi(\dict) \big]\big(\log p(\dict) - \log  Q_\phi(\dict)\big) }.
    \end{align*}
    Then, for a minimal trace $\dict$, we have
    \begin{align*}
    \nabla_{\phi_\alpha}\big[\log Q_\phi(\dict)\big] &=  \delta_{\alpha\in\textnormal{keys}(\dict)}\sum_{\beta \in \textnormal{keys}(\dict)}\nabla_{\phi_\alpha}\log q_\beta(\dict(\beta)|\phi_\beta)\\
    &=\delta_{\dict(\alpha)\neq\none}\cdot\nabla_{\phi_\alpha} \log q_\alpha(\dict(\alpha)|\phi_\alpha).
    \end{align*}

    Further,  let $\Delta_p(\dict) = \log p(\dict) - \log p_\alpha(\dict)$ and $\Delta_q(\dict) = \log Q_\phi(\dict) - \log q_\alpha(\dict(\alpha)|\phi_\alpha)$.
    For an address set $A$ with $\alpha\in A$, by \Cref{theorem:slicing-correctness-corollary} and by definition of $Q_\phi$, both $\Delta_p \circ \iota_\dicts$ and $\Delta_q \circ \iota_\dicts$ are constant with respect to $z_\alpha$ in $Z_A$.
     Thus, 
    \begin{align*}
        &\int 1_{\dicts_\textnormal{minimal}}(\iota_\dicts(z_A)) \, \nabla_{\phi_\alpha} \log q_\alpha(z_\alpha|\phi_\alpha) 
  \, \Delta_p(\iota_\dicts(z_A)) \prod_{\beta\in A} q_\beta(z_\beta)|\phi_\beta) \,\,dz_A  = \\
        &\int \Delta_p(\iota_\dicts(z_A)) \int 1_{\dicts_\textnormal{minimal}}(\iota_\dicts(z_A)) \, \nabla_{\phi_\alpha} \big[\log q_\alpha(z_\alpha|\phi_\alpha)\big] \,   q_\alpha(z_\alpha|\phi_\alpha)\,dz_\alpha \prod_{\beta\in A\setminus\{\alpha\}} q_\beta(z_\beta|\phi_\beta) \,\,dz_{A\setminus\{\alpha\}}\\
        &= 0,
    \end{align*}
    because $\int\nabla_{\phi_\alpha}\big[\log q_\alpha(z_\alpha|\phi_\alpha)\big]  q_\alpha(z_\alpha|\phi_\alpha)\,dz_\alpha =0 $.
    Similarly,
    \begin{align*}
        &\int 1_{\dicts_\textnormal{minimal}}(\iota_\dicts(z_A)) \, \nabla_{\phi_\alpha} \log q_\alpha(z_\alpha|\phi_\alpha) 
  \, \Delta_q(\iota_\dicts(z_A)) \prod_{\beta\in A} q_\beta(z_\beta)|\phi_\beta) \,\,dz_A = 0.
    \end{align*}
    Lastly, we rewrite
    \begin{align*}
    f(\dict) &\coloneqq \nabla_{\phi_\alpha}\big[\log Q_\phi(\dict) \big]\big(\log p(\dict) - \log  Q_\phi(\dict)\big) \\&
    =
    \delta_{\dict(\alpha)\neq\none}\cdot\nabla_{\phi_\alpha} \log q_\alpha(\dict(\alpha)|\phi_\alpha)\cdot \big(\log p_\alpha(\dict) -  \log q_\alpha(\dict(\alpha)|\phi_\alpha)\big) \\
    & \quad\quad+\delta_{\dict(\alpha)\neq\none}\cdot\nabla_{\phi_\alpha} \log q_\alpha(\dict(\alpha)|\phi_\alpha)\cdot\big(\Delta_p(\dict) - \Delta_q(\dict)\big).
    \end{align*}
    Since we have shown that the second term in the sum integrates to $0$ for every address set $A$, with $\alpha \in A$, by linearity of the integral, we have
    \begin{align*}
        \expectation{\dict \sim  Q_\phi}{f(\dict)} &=
        \sum_{\substack{A \subseteq \strings,\\|A|<\infty, \alpha\in A}}\int 1_{\dicts_\textnormal{minimal}}(\iota_\dicts(z_A)) \, f(\iota_\dicts(z_A)) \prod_{\beta\in A} q_\beta(z_\beta)|\phi_\beta) \,\,dz_A \\
        &= \expectation{\dict \sim  Q_\phi}{\delta_{\dict(\alpha)\neq\none}\cdot\nabla_{\phi_\alpha} \log q_\alpha(\dict(\alpha)|\phi_\alpha) \cdot \big(\log p_\alpha(\dict) -\log  q_\alpha(\dict(\alpha)|\phi_\alpha)\big)}.
    \end{align*}
    
    Finally, from the definition of the \texttt{BBVICtx} it follows immediately that this context computes the quantities $\log p_\alpha(\dict)$ and $\log  q_\alpha(\dict(\alpha)|\phi_\alpha)$.
\end{proof}

\subsection{Proof of \Cref{theorem:smc-correctness} - Correctness of SMC optimisation}

First, we formalise the sliced CFG construction for SMC.
\begin{definition}\label{def:SMC-sliced-CFG}
Let $\cfg=(N_\textnormal{start},N_\textnormal{end},\mathcal{N},\mathcal{E})$ be a CFG with $K$ sample nodes, which we enumerate with $N_j = \textnormal{Assign}(x_j = \kw{sample}(E_0^j, f^j(E_1^j,\dots,E_{n_j}^j)))$ for $j = 1,\dots,K$.

    We define the set of nodes that lie on path from $N_k$ to some of other sample node $N_j$:
    \begin{align*}
        \mathcal{N}_k \coloneqq \{N \in \mathcal{N}\colon \exists j\in\{1,\dots,K\}\colon&\exists\text{ path } \nu=(N_k,\dots,N_j)\colon\\
        &N \in \nu \text{ and } N_k \text{ and } N_j \text{ are the only sample nodes in }\nu\}.
    \end{align*}
    In this construction, we do not require $\kw{read}$ statements.
    We restrict the set of edges to the subset of nodes~$\mathcal{N}_k$:
    \[\mathcal{E}_k' \coloneqq \{(e_1, e_2)\colon (e_1,e_2) \in \mathcal{E}\land e_1 \in \mathcal{N}_k \land e_2 \in \mathcal{N}_k\}.\]
    We connect the new start node $N^k_\textnormal{start}$ to $N_k$ and the new end node
    $N^k_\textnormal{end}$ to all $N_j$, $j\in\{1,\dots,K\}$:
    \[\mathcal{E}_k = \mathcal{E}_k'\cup \{(N^k_\textnormal{start}, N_k)\} \cup \{(N_j,N^k_\textnormal{end})\colon j\in\{1,\dots,K\}\land N_j \in \mathcal{N}_k\}.\]
    Finally, define the \emph{CFG sliced with respect to $N_k$} as
    \[\cfg_k^\texttt{SMC} \coloneqq (N^k_\textnormal{start}, N^k_\textnormal{end},\mathcal{N}_k, \mathcal{E}_k).\]    
\end{definition}

\medskip
The "continuation property" of this sliced CFG follows directly from its construction.
\medskip

\begin{statement}
    Let $\cfg$ be a CFG such that for initial program state $\progstate_0$ and trace~$\dict$, we have execution sequence
    $(\progstate_0, \text{START}) \cfgrel \cdots  \cfgrel (\progstate_i, M_i) \cfgrel \cdots \cfgrel (\progstate_l, \text{END})$ in its unrolled version $\unrolledcfg$.
    Let $(\progstate'_1, M'_1) \cfgrel \cdots  \cfgrel (\progstate'_2, M'_2)$ be a sub-sequence such that $\cfgnode(M'_1) = N_k$ and $\cfgnode(M'_2) = N_j$ are the only sample nodes.
    Then, {\revvcolor the execution sequence in the unrolled version of the sliced CFG $\cfg_k^\texttt{SMC}$  matches the sub-sequence}:
    $(\progstate'_1, \text{START}) \cfgrel (\progstate'_1, M'_1) \cfgrel \cdots  \cfgrel (\progstate'_2, M'_2)\cfgrel (\progstate'_2, \text{END}).$
\end{statement}
\begin{proof}
    Follows directly from \cref{def:SMC-sliced-CFG}.
\end{proof}

}

\newpage

\newcommand{\rec}{\textnormal{rec}}

\section{Errata (added in arXiv v3)}\label{sec:errata}

As stated, Corollary~3 is incorrect for programs with control dependencies.
In particular, in the proof, $\Delta_q \circ \iota_{\dicts}$ is not constant with respect to $z_\alpha$ in $Z_A$, because $z_\alpha$ can influence the set of factors that are summed over in $\log Q_\phi$ at $\dict = \iota_{\dicts}(z_A)$.
Furthermore,  $1_{\dicts_\textnormal{minimal}} \circ \iota_{\dicts}$ may also depend on~$z_\alpha$.

To fix this, first define the product measure $\mu = \bigoplus_\beta q_\beta(.|\phi_\beta)$ on the set of all traces, which exists since all $q_\alpha$ are probability measures.
For any trace $\dict$, let $\rec(\dict)$ be the minimal trace that is recorded when the program is executed with $\dict$.
That is, $\rec(\dict)(\beta) = \dict(\beta)$ if the program reaches a sample node whose address expression evaluates to $\beta$ and otherwise $\rec(\dict)(\beta) = \none$.
Then, the variational distributions $Q_\phi$ is the distribution of $\rec(\dict)$ for $\dict \sim \mu$, such that we have $\expectation{\dict\sim Q_\phi}{f(\dict)} = \expectation{\dict \sim \mu}{f(\rec(\dict))}$ for any measurable $f$.
This definition is closely related to the construction of the original proof and implies density
$Q_\phi(\dict) = 1_{\dicts{}_\textnormal{minimal}}(\dict) \prod_{\alpha \in \textnormal{keys}(\dict)} q_\alpha(\dict(\alpha) | \phi_\alpha)$.

Next, let $S_Q$ be the \emph{guide} program that replaces each sample statement of the original program with $x = \kw{sample}(E_0, q(E_0))$, where $\textnormal{pdf}_q(v; \beta) = q_\beta(v | \phi_\beta)$. Assume almost sure termination and denote its CFG with $U$.
By construction, if $p(\dict)$ is defined, then $p_{U}(\dict) = Q_\phi(\rec(\dict))$.
Let $Q_\alpha$ be the function of Corollary~1 for $S_Q$ and sample node $N_k$, i.e., $Q_\alpha(\dict) = \bar p_{U_k}(\dict, \progstate)$ if the execution for $\dict$ reaches the sample node in the sliced CFG $U_k$ with state $\progstate$ and $Q_\alpha(\dict) = 1$ otherwise.
By Corollary~1, $\log Q_\phi(\rec(\dict)) - \log Q_\alpha(\rec(\dict))$ is constant with respect to $\alpha$.
Note that the dependencies in the guide are exactly the control dependencies in the original program.
Pure data dependencies are eliminated with the replacement of the distributions at sample statements.

Now, for this corrected construction, we get
\begin{gather}\label{eq:bbvi-corrected}
\begin{split}
    &\nabla_{\phi_\alpha}\expectation{\dict \sim  Q_\phi}{\log p(\dict) - \log  Q_\phi(\dict)} = \\
    &\quad\quad\quad\expectation{\dict \sim  Q_\phi}{\delta_{\dict(\alpha)\neq\none}\times\nabla_{\phi_\alpha} \log q_\alpha(\dict(\alpha)|\phi_\alpha) \times \big(\log p_\alpha(\dict) -  \textcolor{blue}{\log Q_\alpha(\dict(\alpha))}\big)}.
\end{split}
\end{gather}

The proof follows the original work, where we start with the BBVI estimator identity,
\begin{align*}
    &\nabla_{\phi_\alpha}\expectation{\dict \sim  Q_\phi}{\log p(\dict) - \log  Q_\phi(\dict)} = \\
    &\quad\quad\quad\expectation{\dict \sim  Q_\phi}{\delta_{\dict(\alpha)\neq\none}\times\nabla_{\phi_\alpha} \log q_\alpha(\dict(\alpha)|\phi_\alpha) \times\big(\log p(\dict) - \log  Q_\phi(\dict)\big) }.
\end{align*}
This equation relies on the assumption that address $\alpha$ is sampled at most once and that we can exchange integration and differentiation.

Next, by independency in the variational distribution, we have for any measurable $f$
\begin{align*}
    \expectation{\dict \sim Q_\phi}{f(\dict)} = \int \left(\int f(\rec(\dict[\alpha \mapsto z_\alpha])) q_\alpha(z_\alpha | \phi_\alpha) \, dz_\alpha \right) d\mu_{-\alpha}(\dict),
\end{align*}
where $\mu_{-\alpha} = \bigoplus_{\beta \neq \alpha} q_\beta(.|\phi_\beta)$ and $\dict[\alpha \mapsto z_\alpha]$ is the trace update.
By Corollary~1, \[\Delta(\dict) = \big(\log p(\dict) - \log p_\alpha(\dict)\big) - \big(\log  Q_\phi(\dict) - \log Q_\alpha(\dict)\big)\] is constant with respect to $\alpha$ and for $f(\dict) = \delta_{\dict(\alpha)\neq\none}\times\nabla_{\phi_\alpha} \log q_\alpha(\dict(\alpha)|\phi_\alpha) \times \Delta(\dict)$,
\begin{align*}
     \expectation{\dict \sim Q_\phi}{f(\dict)} = \int \Delta(\rec(\dict)) \left( \int  \nabla_{\phi_\alpha} \log q_\alpha(z_\alpha|\phi_\alpha)  q_\alpha(z_\alpha | \phi_\alpha) \, dz_\alpha\right)d\mu_{-\alpha}(\dict) = 0,
\end{align*}
as the inner integral evaluates to 0.
By linearity of the integral, \Cref{eq:bbvi-corrected} follows.
\pagebreak

To compute $\log Q_\alpha$ in the $\texttt{BBVI}$ context, we have to update the semantics of $\kw{score}$ statements to accumulate the variational density at dependencies of $\alpha$: 
\begin{equation*}
    \frac{
    (\progstate, x = \kw{sample}(E_0, f(E_1,\dots,E_n))) \Downarrow^{\dict} \progstate' \quad \beta \coloneqq \progstate(E_0) \neq\alpha \quad V=\dict(\sigma(E_0))
    }{
    \begin{aligned}
    (\progstate, x = \kw{score}(E_0, f(E_1,\dots,E_n))) \Downarrow^{\dict}_{\texttt{BBVI}(\alpha)} \progstate[x\mapsto V, &\boldsymbol{\log p} \mapsto \log \progstate'(\boldsymbol{p}), \\ 
    &\textcolor{blue}{\boldsymbol{\log q} \mapsto \boldsymbol{\log q} + \log q_\beta(V| \phi_\beta)}]
    \end{aligned}
    }
\end{equation*}
For correctness, we only need to accumulate the variational density at nodes \emph{control} dependent on $\alpha$ (since these are the only dependencies in the guide program).
But in practice, we have found that accumulating the density for all dependencies improves variance.
This may be due to the fact that the variational factor $q_\alpha$ becomes correlated with the marginal distribution of $z_\alpha$ in $p_\alpha$ during optimisation.

In \Cref{tab:corrected-bbvi-results}, we report the updated experiment results.
We observe slightly improved variance for Aircraft and Bayesian Network, 2x improvement for the HMMs and Hurricane, and a great improvement for PCFG.
The results for the other models are unchanged.

\begin{table}[h]
    \caption{Updated Table~3 with corrected variational distribution.}
    \label{tab:corrected-bbvi-results}
    \centering
\begin{tabular}{l r r r r c r r r}
        \hline
         \textbf{}& \textbf{Baseline} &  \multicolumn{4}{c}{\textbf{Static Analysis}} &  \multicolumn{2}{c}{\textbf{Pyro}}\\
         \textbf{Model} & avg. var. & avg. var. &  \multicolumn{2}{c}{red. factor} & cost & red. factor & cost \\
         \hline
Aircraft & 7.922e+03 & 6.730e+01 & \textcolor{ForestGreen}{117.57} & \textcolor{ForestGreen}{$\megaslowarrow$} & 2.11 & - & - \\
Bayesian Network & 4.160e+02 & 2.786e+00 & \textcolor{ForestGreen}{149.24} & \textcolor{ForestGreen}{$\megaslowarrow$} & 2.03 & 95.09 &  7.52 \\
Geometric & 1.634e-01 & 1.681e-01 & \textcolor{Gray}{1.00} & \textcolor{Gray}{$\samearrow$} & 0.60 & - & - \\
GMM (fixed \#clusters) & 8.043e+09 & 1.014e+06 & \textcolor{ForestGreen}{8092.82} & \textcolor{ForestGreen}{$\megaslowarrow$} & 2.37 & 9848.79 &  8.69 \\
HMM & 2.116e+03 & 5.830e+02 & \textcolor{ForestGreen}{3.63} & \textcolor{ForestGreen}{$\slowarrow$} & 8.44 & 1254.24 &  8.67 \\
HMM (no loop) & 2.116e+03 & 1.151e+00 & \textcolor{ForestGreen}{1838.75} & \textcolor{ForestGreen}{$\megaslowarrow$} & 1.66 & 1254.24 &  8.64 \\
Hurricane & 3.362e-01 & 1.524e-01 & \textcolor{ForestGreen}{2.21} & \textcolor{ForestGreen}{$\slowarrow$} & 0.90 & - & - \\
LDA (fixed \#topics) & 7.229e+04 & 1.168e+00 & \textcolor{ForestGreen}{61905.85} & \textcolor{ForestGreen}{$\megaslowarrow$} & 6.74 & 59986.61 &  11.00 \\
Linear regression & 3.350e+05 & 3.341e+05 & \textcolor{Gray}{1.00} & \textcolor{Gray}{$\samearrow$} & 1.18 & 1.00 &  14.39 \\
Marsaglia & 7.274e-02 & 1.438e-02 & \textcolor{ForestGreen}{5.08} & \textcolor{ForestGreen}{$\slowarrow$} & 0.54 & - & - \\
PCFG & 8.838e-03 & 1.236e-03 & \textcolor{ForestGreen}{7.25} & \textcolor{ForestGreen}{$\slowarrow$} & 1.18 & - & - \\
Pedestrian & 3.108e+06 & 3.110e+06 & \textcolor{Gray}{1.00} & \textcolor{Gray}{$\samearrow$} & 2.31 & - & - \\
\hline
\end{tabular}
\end{table}

\end{document}